\newif\iflncs
\lncsfalse
\newif\ifnotes
\notesfalse

\iflncs
\documentclass[runningheads,a4paper]{llncs}
\else
\documentclass[11pt]{article}
\fi

\iflncs
\usepackage{breakcites}
\else
\usepackage{amsthm}
\usepackage{fullpage}
\fi

\usepackage{subcaption}
\usepackage{algorithm}
\usepackage[noend]{algpseudocode}
\usepackage{amsfonts}
\usepackage{amsmath}
\usepackage{bbm}
\usepackage{complexity}
\usepackage{mdframed}
\usepackage{float}
\usepackage{array}
\usepackage{enumitem}
\usepackage[dvipsnames]{xcolor} 
\usepackage{graphicx}
\usepackage{tikz}
\usepackage{multirow}
\usepackage{pdfpages}
\usepackage[colorlinks=true,linkcolor=Maroon,citecolor=Maroon]{hyperref}
\usepackage{thm-restate}
\usepackage{orcidlink}                    
\usepackage[capitalize]{cleveref}  

\newcommand{\remove}[1]{}

\ifnotes
\newcommand{\hnote}[1]{[{\footnotesize \textsf{\color{purple}{\bf Haoxing:} { {#1}}}}]}
\definecolor{ao}{rgb}{0.0, 0.5, 0.0}
\newcommand{\pnote}[1]{[{\footnotesize \textsf{\color{cyan}{\bf Prashant:} { {#1}}}}]}
\else
\newcommand{\hnote}[1]{}
\definecolor{ao}{rgb}{0.0, 0.5, 0.0}
\newcommand{\pnote}[1]{}
\fi
\definecolor{mygreen}{RGB}{0,128,0} 
\newcommand{\resolved}[1]{{\color{mygreen}{[Resolved]}}}

\iflncs
\input{lncs_issues}
\else
\newtheorem{informaltheorem}{Theorem}
\newtheorem{theorem}{Theorem}[section]

\newtheorem{claim}[theorem]{Claim}
\newtheorem{corollary}[theorem]{Corollary}
\newtheorem{fact}[theorem]{Fact}
\newtheorem{lemma}[theorem]{Lemma}
\newtheorem{proposition}[theorem]{Proposition}
\newtheorem{definition}[theorem]{Definition}
\newtheorem{remark}[theorem]{Remark}
\fi

\iflncs
\renewcommand{\paragraph}[1]{\;\newline \noindent \textbf{#1}}
\fi

\newenvironment{proofof}[1]{\begin{proof}[\textit{Proof of #1}]}{\end{proof}}

\newenvironment{proofsketchof}[1]{\begin{proof}[\textit{Proof Sketch of #1}]}{\end{proof}}

\crefname{definition}{Definition}{Definitions}
\crefname{sub-definition}{Definition}{Definitions}
\crefname{example}{Example}{Examples}
\crefname{exercise}{Exercise}{Exercises}
\crefname{property}{Property}{Properties}
\crefname{question}{Question}{Questions}
\crefname{solution}{Solution}{Solutions}
\crefname{theorem}{Theorem}{Theorems}
\crefname{informaltheorem}{Theorem}{Theorems}
\crefname{proposition}{Proposition}{Propositions}
\crefname{problem}{Problem}{Problems}
\crefname{lemma}{Lemma}{Lemmas}
\crefname{conjecture}{Conjecture}{Conjectures}
\crefname{corollary}{Corollary}{Corollaries}
\crefname{fact}{Fact}{Facts}
\crefname{claim}{Claim}{Claims}
\crefname{remark}{Remark}{Remarks}
\crefname{note}{Note}{Notes}
\crefname{figure}{Figure}{Figure}
\crefname{case}{Case}{Cases}
\crefname{proofsketch}{Proof Sketch}{Proof Sketches}

\def\ddefloop#1{\ifx\ddefloop#1\else\ddef{#1}\expandafter\ddefloop\fi}

\def\ddef#1{\expandafter\def\csname bb#1\endcsname{\ensuremath{\mathbb{#1}}}}
\ddefloop ABCDEFGHIJKLMNOPQRSTUVWXYZ\ddefloop

\def\ddef#1{\expandafter\def\csname c#1\endcsname{\ensuremath{\mathcal{#1}}}}
\ddefloop ABCDEFGHIJKLMNOPQRSTUVWXYZ\ddefloop

\def\ddef#1{\expandafter\def\csname v#1\endcsname{\ensuremath{\overline{#1}}}}
\ddefloop ABCDEFGHIJKLMNOPQRSTUVWXYZabcdefghijklmnopqrstuvwxyz\ddefloop

\newcommand{\algfont}[1]{\mathsf{#1}}
\def\ddef#1{\expandafter\def\csname alg#1\endcsname{\ensuremath{\algfont{#1}}}}
\ddefloop ABCDEFGHIJKLMNOPQRSTUVWXYZ\ddefloop

\newcommand{\langfont}[1]{\mathnormal{#1}}
\def\ddef#1{\expandafter\def\csname lang#1\endcsname{\ensuremath{\langfont{#1}}}}
\ddefloop ABCDEFGHIJKLMNOPQRSTUVWXYZ\ddefloop

\def\ddef#1{\expandafter\def\csname v#1\endcsname{\ensuremath{\boldsymbol{\csname #1\endcsname}}}}
\ddefloop {alpha}{beta}{gamma}{delta}{epsilon}{varepsilon}{zeta}{eta}{theta}{vartheta}{iota}{kappa}{lambda}{mu}{nu}{xi}{pi}{varpi}{rho}{varrho}{sigma}{varsigma}{tau}{upsilon}{phi}{varphi}{chi}{psi}{omega}{Gamma}{Delta}{Theta}{Lambda}{Xi}{Pi}{Sigma}{varSigma}{Upsilon}{Phi}{Psi}{Omega}{ell}\ddefloop

\newcommand{\Nat}{\mathbb{N}}
\newcommand{\Int}{\mathbb{Z}}
\newcommand{\Reals}{\mathbb{R}}
\newcommand{\bset}{\set{0,1}}
\newcommand{\bell}{{\bar{\ell}}}
\newcommand{\ind}{\mathrm{Ind}}

\newcommand{\ra}{\rightarrow}
\newcommand{\Otilde}{\tilde{O}}
\newcommand{\Input}{\item[\textbf{Input:}]}
\newcommand{\Output}{\item[\textbf{Output:}]}
\newcommand{\FuncA}[2]{\mathbf{A}_{#1,#2}}
\newcommand{\FuncB}[2]{\mathbf{B}_{#1,#2}}
\newcommand{\FuncC}[2]{\mathbf{C}_{#1,#2}}
\newcommand{\FuncD}[2]{\mathbf{D}_{#1,#2}}
\newcommand{\FuncInd}[1]{\mathbbm{1}\left[#1\right]}
\newcommand{\Gammaone}[1]{\gamma_{1}(#1)}
\newcommand{\Gammazero}[1]{\gamma_{0}(#1)}
\newcommand{\Gamman}[1]{\gamma_{\neq}(#1)}

\DeclareMathOperator*{\Exp} {{\mathbf{E}}}

\newcommand{\pr}[1]         {\Pr\left[ #1 \right]}
\newcommand{\prob}[2]       {\Pr_{#1}\left[ #2 \right]}

\renewcommand{\exp}[1]      {\Exp\left[ #1 \right]}
\newcommand{\expec}[2]      {\Exp_{#1}\left[ #2 \right]}

\renewcommand{\poly}{\mathrm{poly}}
\renewcommand{\polylog}{\mathrm{polylog}}

\DeclareMathOperator{\mr}   {\Delta_{MR}}

\DeclareMathOperator{\diff} {\delta}

\newcommand{\set}[1]        {\left\{ #1 \right\}}
\newcommand{\abs}[1]        {\left| #1\right|}
\newcommand{\size}[1]       {\left| #1\right|}
\newcommand{\floor}[1]      {\left\lfloor #1 \right\rfloor}
\newcommand{\ceil}[1]       {\left\lceil #1 \right\rceil}

\newcommand{\range}[1]      {\langle #1 \rangle}

\newcommand{\ktree}         {\algfont{kTree}}

\title{On Wagner's $k$-Tree Algorithm Over Integers}
\author{Haoxing Lin\thanks{\orcidlink{0000-0001-9594-1871} National University of Singapore. Email: \texttt{\href{mailto:haoxingl@comp.nus.edu.sg}{haoxingl@comp.nus.edu.sg}}} \and Prashant Nalini Vasudevan\thanks{\orcidlink{0000-0001-6880-795X} National University of Singapore. Email: \texttt{\href{mailto:prashant@comp.nus.edu.sg}{prashant@comp.nus.edu.sg}}}}
\date{\today}

\begin{document}

\maketitle

\thispagestyle{empty}
\begin{abstract}
  The $k$-Tree algorithm~\cite{Wagner02} is a non-trivial algorithm for the average-case $k$-SUM problem that has found widespread use in cryptanalysis. Its input consists of $k$ lists, each containing $n$ integers from a range of size $m$.  Wagner's original heuristic analysis~\cite{Wagner02} suggested that this algorithm succeeds with constant probability if $n \approx m^{1/(\log{k}+1)}$, and that in this case it runs in time $O(kn)$. Subsequent rigorous analysis of the algorithm~\cite{Lyubashevsky05,Shallue08,JKL24} has shown that it succeeds with high probability if the input list sizes are significantly larger than this.
  
  We present a broader rigorous analysis of the $k$-Tree algorithm, showing upper and lower bounds on its success probability and complexity for any size of the input lists. Our results confirm Wagner's heuristic conclusions, and also give meaningful bounds for a wide range of list sizes that are not covered by existing analyses. We present analytical bounds that are asymptotically tight, as well as an efficient algorithm that computes (provably correct) bounds for a wide range of concrete parameter settings. We also do the same for the $k$-Tree algorithm over $\Int_m$. Finally, we present experimental evaluation of the tightness of our results. 
\end{abstract}



\newpage

\thispagestyle{empty}
\tableofcontents


\newpage
\section{Introduction}
\label{sec:intro}



The (average-case) $k$-SUM problem is a basic computational problem that comes up relatively often in cryptanalysis~\cite[\dots]{Schnorr01,Wagner02,DEFKLNS19,BLLOR22}. In this problem, given $k$ lists of $n$ integers each drawn uniformly at random from the range $\{-\floor{m/2},\dots,$ $\floor{m/2}\}$, the task is to find a set of $k$ integers, one from each list, such that their sum is $0$. Such sets exist with high probability once the list size $n$ is larger than $m^{1/k}$; and in this case there is a simple meet-in-the-middle algorithm that runs in time $\Otilde(m^{1/2})$. The only known algorithms that do significantly better than this are the $k$-Tree algorithm~\cite{Wagner02} and its extensions~\cite{MS12,NS15,Dinur19} and variants~\cite{BKW03,Lyubashevsky05,JKL24}. 

\medskip
The $k$-Tree algorithm, as formulated by Wagner~\cite{Wagner02}, is described in \cref{fig:ktree-intro}.\iflncs\else\footnote{See \cref{fig:ktree} in \cref{sec:analysis} for a more detailed presentation.}\fi \ Wagner provided heuristic arguments indicating that this algorithm succeeds with constant probability if the input lists are of size $n \approx m^{1/(\log{k}+1)}$, in which case it runs in time $O(k\cdot m^{1/(\log{k}+1)})$. Even for modest values of $k$ (say $4$ or $8$), this runtime is significantly better than the earlier $\Otilde(m^{1/2})$ for large values of $m$. For this reason, the $k$-Tree algorithm has been used repeatedly in the cryptanalysis of signatures~\cite{Wagner02,DEFKLNS19,BLLOR22}, hash functions~\cite{CJ04,BC22}, identification schemes~\cite{LF06}, symmetric-key encryption schemes~\cite{Joux03}, etc., to provide simple attacks that perform significantly better than brute force.

\begin{figure}[h!]
  \centering
  \begin{mdframed}
    \begin{center}
      \textbf{\underline{The $k$-Tree algorithm}}
    \end{center}
    \vspace{0.5em}

    \medskip
    \textbf{Parameters:} $k, n, m \in \Nat$, with $k$ being a power of $2$

    \medskip
    \textbf{Input:} Lists $L_1,\dots,L_k$, each consisting of $n$ integers from $\set{-\floor{m/2},\dots,\floor{m/2}}$

    \medskip
    \textbf{Output:} Indices $\ell_1,\dots,\ell_k\in[n]$, or a failure symbol $\bot$
    
    \medskip
    \textbf{Procedure:}
    \begin{enumerate}[topsep=5pt,itemsep=0pt]
      \item Initialization:
      \begin{itemize}[topsep=0pt,itemsep=0pt]
        \item Set $p = m^{1/(\log{k}+1)}$
        \item For each $i\in[k]$, denote the list $L_i$ by $L^0_i$
        \item Set $\tau \gets m/2$
      \end{itemize}
      \item For $d$ from $1$ to $\log{k}$:
      \begin{itemize}[topsep=0pt,itemsep=5pt]
        \item Set $\tau \gets p \cdot \tau$
        \item For $i \in \left[\frac{k}{2^d}\right]$: \newline
        compute $L^d_i \gets \set{ (a+b)\ \big|\ a\in L^{d-1}_{2i-1} \wedge b\in L^{d-1}_{2i} \wedge \abs{a+b} \leq \tau}$
      \end{itemize}
      \item If $L^{\log{k}}_1$ contains $0$, output the indices in the input lists that led to this sum. Otherwise output $\bot$.
    \end{enumerate}

    \medskip
  \end{mdframed}
  \caption{The $k$-Tree algorithm over Integers}
  \label{fig:ktree-intro}
\end{figure}

\medskip
In spite of such widespread relevance, the complexity of the $k$-Tree algorithm is still not well understood. Wagner's aforementioned heuristic analysis was based on assumptions that are not strictly true, though empirical evidence suggests that his conclusion is valid.\footnote{See, for instance, the results of our experiments presented in \iflncs\else \cref{sec:experiments}\fi .} Rigorous analysis of the algorithm has since been carried out by Lyubashevsky~\cite{Lyubashevsky05}, Shallue~\cite{Shallue08}, and Joux et al.~\cite{JKL24}, but these have not been able to confirm this conclusion (see \cref{sec:related} for details). Further, we still do not have the means to provide good answers to basic questions like the following, even heuristically:
\begin{enumerate}[itemsep=0pt]
  \item What is the probability that the algorithm will succeed if the lists are of size $(10\cdot m^{1/(\log{k}+1)})$ or $(0.1 \cdot m^{1/(\log{k}+1)})$?
  \item What size of lists is necessary or sufficient for the algorithm to succeed with probability $0.01$?
  \item What is the complexity of the algorithm for these list sizes?
\end{enumerate}
Answers to these questions for some pairs of values $(m,k)$ for which these list sizes are not too large -- e.g., $(2^{64},4)$, $(2^{256},128)$ -- can be obtained empirically, but it is not clear how to extrapolate these to other settings that may come out of concrete parameter choices in constructions -- e.g., $(2^{256},4)$. Our objective in this work is to provide a tight analysis of the $k$-Tree algorithm that can help researchers answer questions like the ones above and thus understand more accurately the power of attacks that make use of the $k$-Tree algorithm.

\paragraph{Paper Outline.} We summarize our results in \cref{sec:results}, and provide an overview of our techniques in \cref{sec:overview}. 
In \cref{sec:related}, we describe prior work on this topic, present comparisons to our results, and discuss other related work.
\iflncs\else
In \cref{sec:analysis}, we provide the complete proof of our main theorem regarding the behavior of the $k$-Tree algorithm.
\fi
In \cref{sec:program}, we present pseudo-code for computing tighter bounds on the behavior of the algorithm.
\iflncs\else
In \cref{sec:experiments}, we present experimental evaluation of our bounds for concrete parameters.
\fi

\iflncs
\else
\subsection{Our Results}
\fi
\label{sec:results}

Recall that we are interested in the performance of the $k$-Tree algorithm when given as input $k$ lists of $n$ numbers each drawn uniformly at random from $\set{-\floor{m/2},\dots,\floor{m/2}}$. Here $k$ is a power of $2$. With these parameters, we define the \emph{complexity} of the $k$-Tree algorithm to be the expected total size of all the lists that are generated in the course of its execution.\footnote{See \cref{rem:complexity} at the end of this subsection for a brief discussion of this complexity measure.}

\paragraph{Analytical Bounds.} Our first set of results are analytical bounds on the success probability and complexity of the $k$-Tree algorithm for any list size. 

\begin{informaltheorem}
  \label{infthm:ktree}
  Consider any $k, n, m\in \Nat$, where $k \geq 4$ is a power of $2$ and $m > 30^{\log{k}+1}$. Set $p = m^{\frac{-1}{\log{k}+1}}$ and $c = n / m^{\frac{1}{\log{k}+1}}$. The probability of success of the $k$-Tree algorithm is bounded as:
    \begin{align*}
      \frac{c^k}{1 + c^k\cdot \left(1+\frac{k}{n} \right)^k}\cdot (1-150p)^{k} \leq \pr{\text{$k$-Tree succeeds}} \leq c^k \cdot (1 + 37p)^k
    \end{align*}
    Its complexity is bounded as follows:
    \begin{align*}
      \text{Complexity of $k$-Tree} \in k n \cdot \left( 1 + \sum_{d\in[\log{k}]} \frac{c^{2^d-1}}{2^d} \right) \cdot \left(1 \pm 37p \right)^{k-1}
    \end{align*}
\end{informaltheorem}

Asymptotically, the following simpler bound can be inferred from the above theorem if $k$ is not too large in relation to the other parameters. \iflncs\else(For a more rigorous statement, see \cref{cor:ktree} in \cref{sec:analysis}.)\fi

\begin{corollary}
  \label{infcor:ktree}
  Consider the same parameters as in \cref{infthm:ktree}, and in addition suppose that $k = o(1/p)$ and $k = o(n^{1/2})$. Then we have:
    \begin{align*}
      \frac{c^k}{1 + c^k}\cdot (1-o(1)) \leq \pr{\text{$k$-Tree succeeds}} \leq c^k \cdot (1 + o(1))
    \end{align*}
\end{corollary}

The setting $c = 1$ corresponds to the choice of $n = m^{1/(\log{k}+1)}$ for the lists sizes suggested by Wagner's heuristic argument~\cite{Wagner02}. As long as $k$ is not too large, the above results confirm that in this case the $k$-Tree algorithm succeeds with some probability roughly between $1/2$ and $1$. In addition, they also allow one to compute list sizes and the resulting complexity that is sufficient for the algorithm to succeed with probability, say, $0.01$; or even the complexity that is \emph{necessary} for the algorithm to succeed with probability at least $0.01$. This enables answering the various kinds of questions about the algorithm presented earlier in this section.

\paragraph{Computational Bounds.} There are still some meaningful concrete values of $(m,k)$ (for example, $(2^{256},1024)$), however, for which \cref{infthm:ktree} does not give accurate bounds. This is due to two sources of inaccuracy in our results above: gaps that are inherent to our technique itself, and approximations that were made in the course of our proof to make intermediate expressions easier to handle analytically. Of these, the latter turn out to substantially dominate.

We eliminate this second source of inaccuracy by implementing our entire proof as a program that, instead of attempting to produce simple analytical expressions, explicitly computes the resulting probability and complexity bounds given values for the parameters $k$, $m$, and $n$. The pseudo-code for the relevant algorithms are presented in \cref{sec:program}, and our implementation of it is available \iflncs\else in the associated repository~\cite{repo}\fi .

\begin{informaltheorem}
  \label{infthm:program}
  Consider any $k, n, m\in \Nat$, where $k \geq 4$ is a power of $2$, and let $p = m^{\frac{-1}{\log{k}+1}}$. The pair of values output by $\textsc{ProbBounds}(m,k,n,p)$ bound the success probability of the $k$-Tree algorithm, and those output by $\textsc{SizeBounds}(m,k,n,p)$ bound its complexity. Further, these algorithms run in time $\polylog(k,m,n)$.
\end{informaltheorem}

We find that this tightening does significantly strengthen the resulting bounds for various concrete values of $(m,k)$. We plot the bounds on success probability from both these theorems for a couple of pairs of values of $(m,k)$ in \cref{fig:intro-1} for illustration. For the rest of this section, we solely use the bounds from \cref{infthm:program} as they are more accurate.

\begin{figure}[h!]
  \centering
  \includegraphics[width=0.48\textwidth]{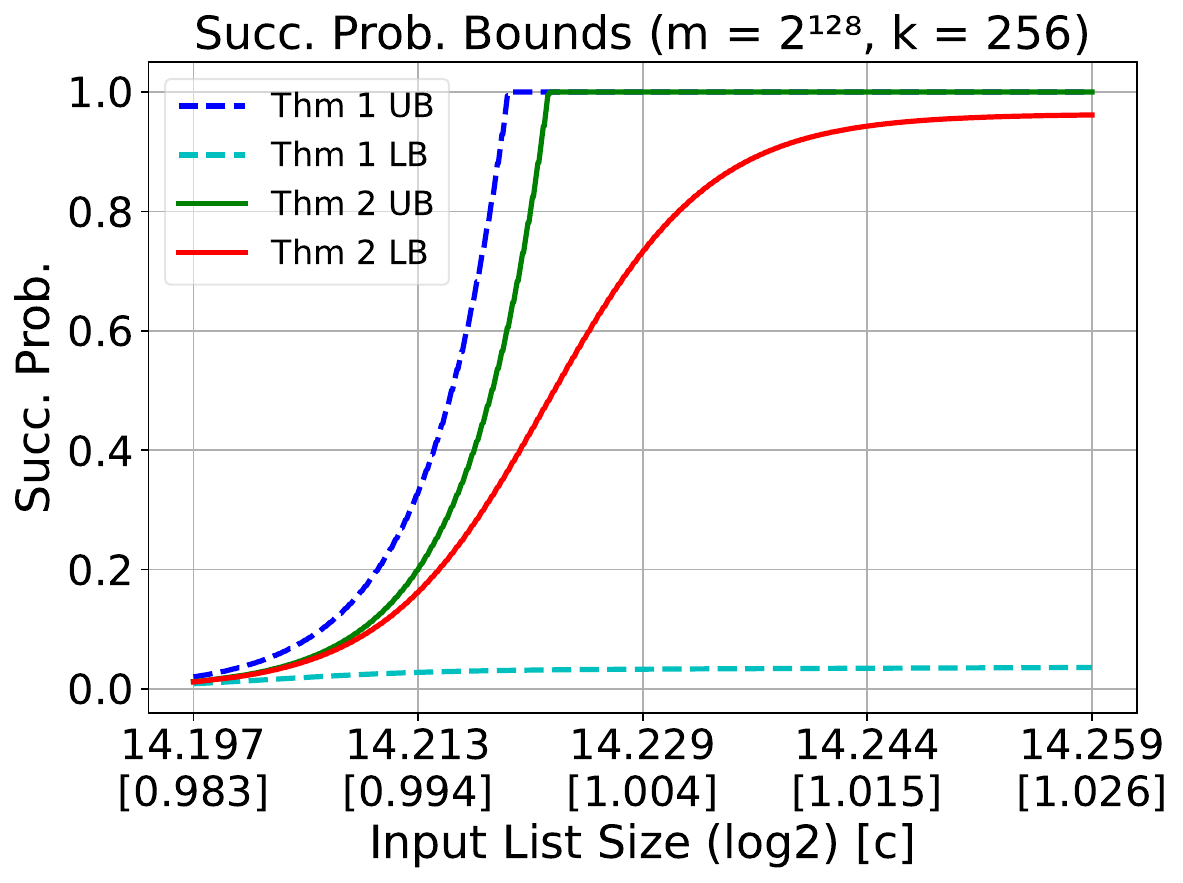}
  \hfill
  \includegraphics[width=0.48\textwidth]{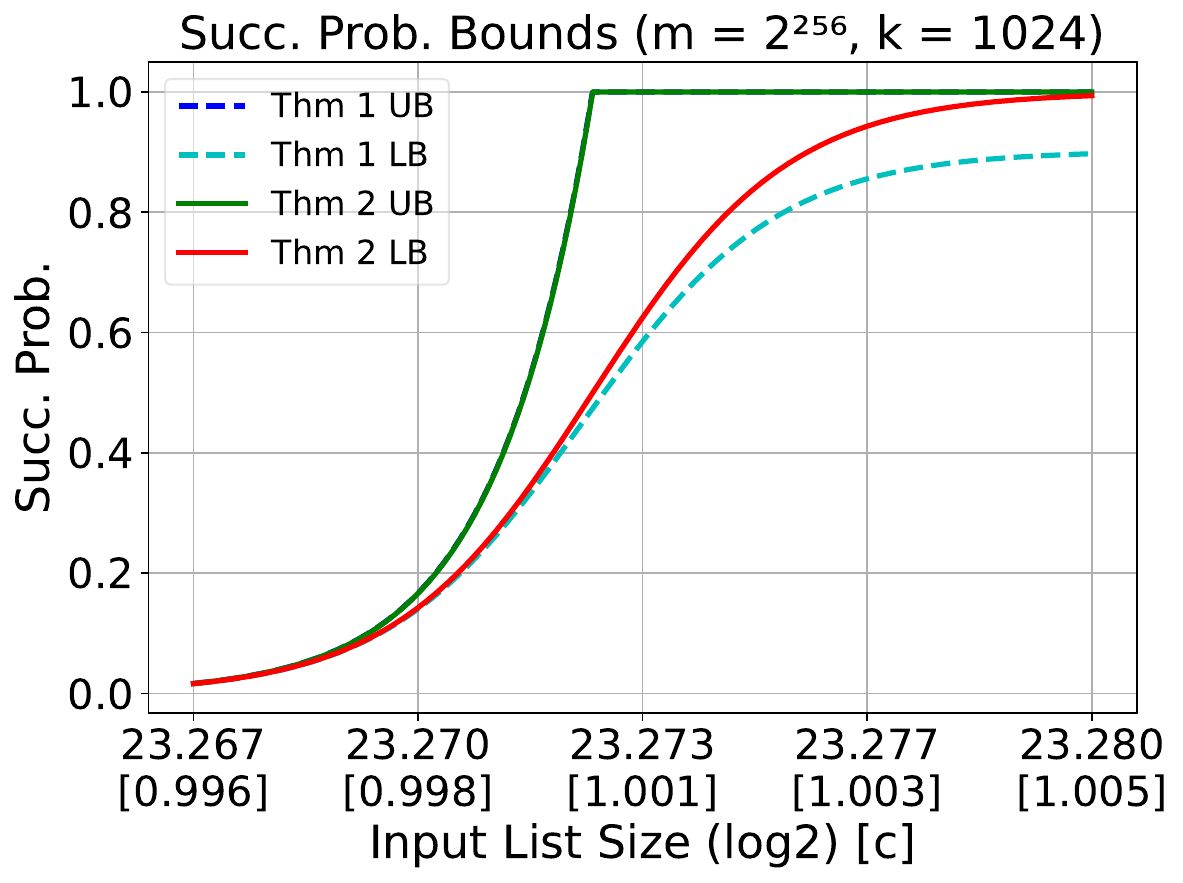}
  \caption{Plot of upper and lower bounds on success probabilities against the input list size $n$. Bounds from both \cref{infthm:ktree,infthm:program} are represented. The quantity $c$ is defined to be the ratio $n/m^{1/(\log{k}+1)}$. Labels on the x-axis are $\log_2(n)$, with the corresponding value of $c$ in square brackets.}
  \label{fig:intro-1}
\end{figure}

\paragraph{Using Our Bounds.} We expect our results to be useful in a few different ways. In applications of the $k$-Tree algorithm where the size of the range $m$ and the number of lists $k$ are fixed and not in the control of the algorithm designer (e.g. the attack in \cite{CJ04}, solving the ROS problem~\cite{Wagner02,BLLOR22}), our bounds can be used to determine the list size (and the corresponding complexity) that is sufficient to achieve a specific desired probability of success. This may be done analytically using \cref{infthm:ktree} if $k$ is relatively small, or computationally using \cref{infthm:program} (with a binary search for $n$) for larger values of $k$. 

In applications where the range $m$ is fixed but $k$ may be determined by the algorithm designer (e.g. attacks in \cite{DEFKLNS19,Joux03}, the attacks on hash functions in \cite{Wagner02}), our bounds can also be used to select the best value of $k$ together with the best value of $n$ (as the complexity depends on both) for a desired probability of success. In \cref{fig:intro-2}, we present some sample computations of this form, covering some of the parameter settings from the papers cited above as examples.

Finally, our upper bounds on the probability of success of the $k$-Tree algorithm can be used to obtain concrete lower bounds on the complexity of $k$-Tree-based attacks on candidate cryptographic constructions that achieve any given probability of success.

\begin{figure}[h!]
  \centering
  \includegraphics[width=0.48\textwidth]{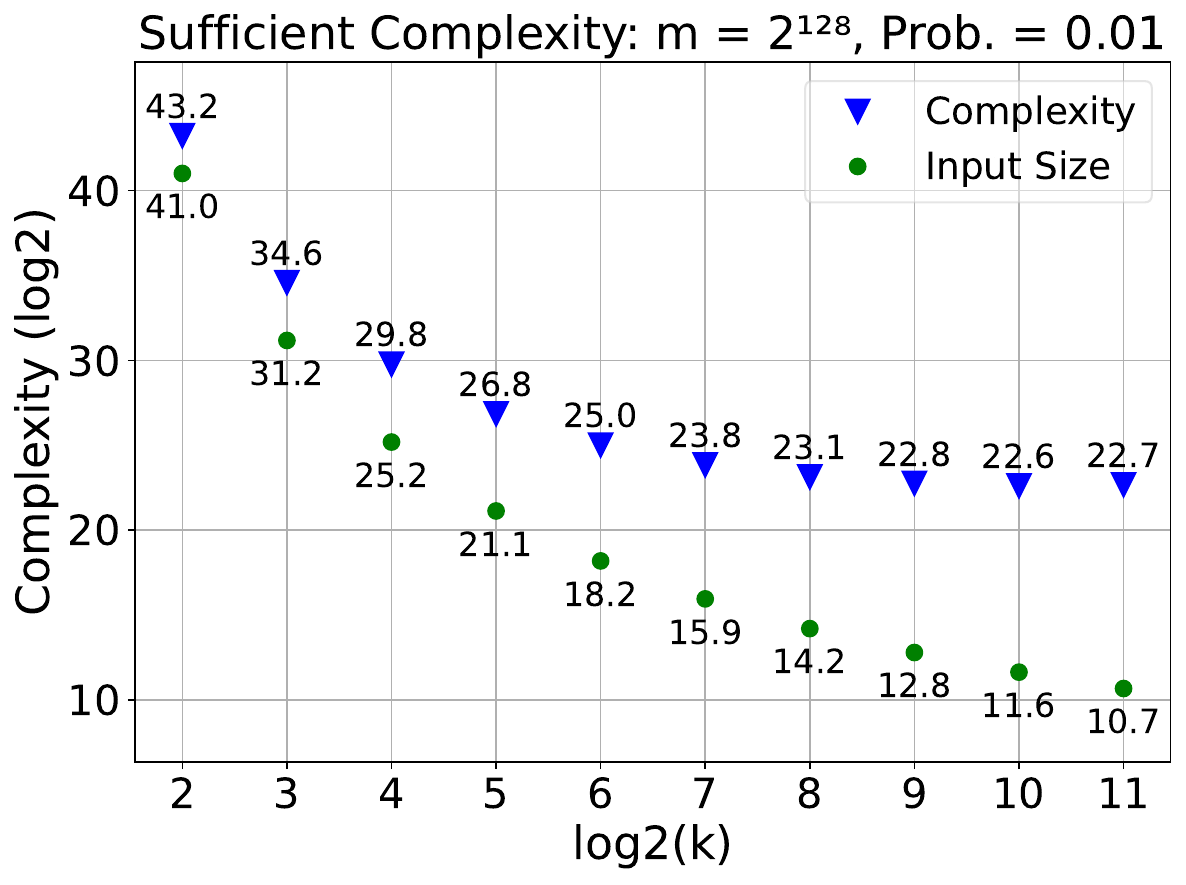}
  \hfill
  \includegraphics[width=0.48\textwidth]{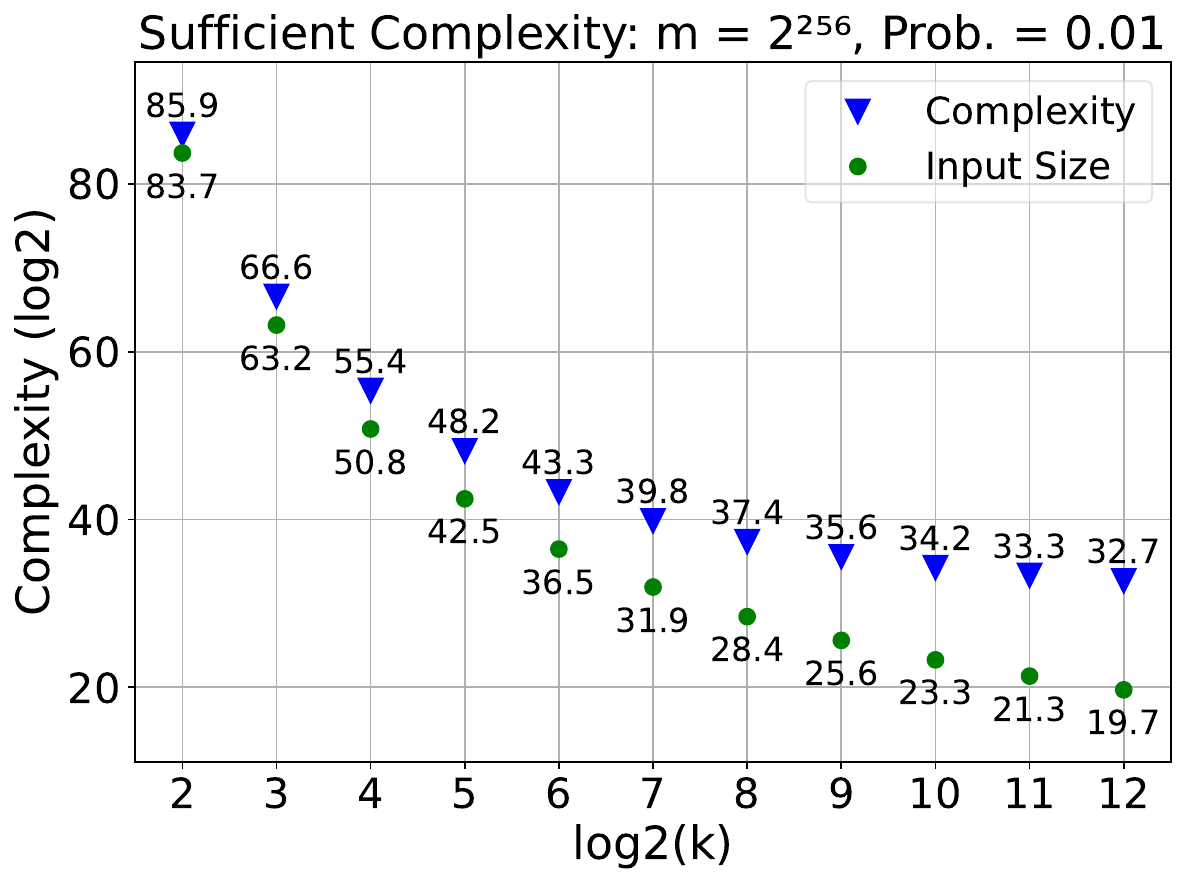}
  \caption{Plots of the complexity of the $k$-Tree algorithm that is sufficient to provably achieve success probability of $0.01$, against $k$. Input list size in each case is chosen so that the lower bound on success probability is slightly larger than $0.01$. Computed using \cref{infthm:program}.}
  \label{fig:intro-2}
\end{figure}

\paragraph{Experimental Evaluation.} To understand the tightness of our upper and lower bounds, we empirically estimate the values of the quantities bounded in the theorems above for various values of the parameters $k$, $m$, and $n$, and compare them to our predictions. We refer the reader to \iflncs\else \cref{sec:experiments} \fi for details of these experiments, and a summary of their results and implications.

\begin{remark}[Simple Optimizations]
  \label{rem:optimisations}
  There are some simple optimizations to the $k$-Tree algorithm that can significantly improve its behavior in some parameter regimes. For instance, removing duplicates when computing the intermediate lists $L^d_i$ ensures that the complexity of the algorithm is at most $O(k(n+m))$, which is smaller than the bounds produced by \cref{infthm:ktree,infthm:program} for very large values of $c$.

  For large values of $c$, the lower bound on the success probability produced by these theorems is roughly $(1-c^{-k})$. But for such $c$, simply breaking up the lists into $\floor{c}$ smaller lists with $m^{1/(\log{k}+1)}$ elements each and running $k$-Tree on each separately (effectively with ``$c$''$=1$) results in an amplified success probability of $(1-e^{-\Omega(c)})$.
\end{remark}

\begin{remark}[Measure of Complexity]
  \label{rem:complexity}
  Depending on various incidental factors such as the hardware being used, its word-size in relation to $m$, etc., the actual running time of the optimal implementation of the $k$-Tree algorithm might be anything from a constant factor to a logarithmic factor larger than the total list size. However, in all cases, the total list size is the central quantity that determines the running time. So for the sake of generality and simplicity, we directly use this as our measure of complexity.
\end{remark}

\begin{remark}[Integers vs. $\Int_m$]
  \label{rem:zm}
  So far, we have been discussing the $k$-SUM problem where the addition is over integers. In reality, the problem that is most often useful in cryptanalysis, etc., is the problem where addition is done modulo $m$ (that is, over $\Int_m$), often for some prime number $m$. In this case, the $k$-Tree algorithm is modified to perform all its additions modulo $m$ instead of over integers, with $\Int_m$ being identified with the set $\set{-\floor{m/2},\dots,\floor{m/2}}$ in the natural manner. Prior analyses by Shallue~\cite{Shallue08} and Joux et al.~\cite{JKL24} were also for $k$-Tree over $\Int_m$. However, as demonstrated in \iflncs\else \cref{sec:zm}\fi , the difference in the behavior of the $k$-Tree algorithm in these two cases is so small as to be beneath the precision with which we state our results in this section. So for simplicity, we continue to ignore this distinction in this section, though \iflncs\else later sections \fi will take it into account.
\end{remark}

\subsection{Technical Overview}
\label{sec:overview}

In this section, we present an overview of our proof of the bounds on the success probability of the $k$-Tree algorithm presented in \cref{infthm:ktree}. Our bounds on the complexity of the algorithm follow from similar arguments. The algorithms captured in \cref{infthm:program} follow the same high-level approach as this proof, but involve some careful re-framing of its terms in order to enable efficient computation of some intermediate values. The entire rigorous proof is presented in \iflncs\else \cref{sec:analysis}\fi .

Our approach to these bounds is quite elementary -- we define a random variable corresponding to the number of occurrences of ``$0$'' in the final list computed by the $k$-Tree algorithm, compute bounds on its first and second moments, and then use standard second-moment-based tail bounds to bound the probability that this variable has non-zero value. The challenge lies in computing these moments, as this variable is somewhat complex.

\paragraph{Setup.} Recall that the input to the $k$-Tree algorithm is $k$ lists of $n$ integers, with each integer being drawn uniformly at random from the set $\set{-\floor{m/2},\dots,\floor{m/2}}$. Given these parameters, define the symbols $p = m^{-1/(\log{k}+1)}$, which is the filtering parameter in the algorithm, and $c = n\cdot p$, which represents the size of the lists relative to the ``default value'' of $m^{1/(\log{k}+1)}$ suggested by heuristic arguments. For any non-negative integer $s$, we will denote by $\range{s}$ such a set $\set{\floor{s/2},\dots,\floor{s/2}}$. 

\medskip
For simplicity, we will fix the number of lists $k$ to $4$ in this overview (and so $p = m^{-1/3}$). In this case, the $k$-Tree algorithm is given as inputs $k = 4$ lists $L_1,\dots,L_4$, each of size $n$, and proceeds as follows:
\begin{enumerate}
  \item Compute the following lists:
  \begin{itemize}
    \item $L^1_1 \gets \set{ (a+b)\ \big|\ a\in L_{1} \wedge b\in L_{2} \wedge (a+b) \in \range{mp}}$
    \item $L^1_2 \gets \set{ (a+b)\ \big|\ a\in L_{3} \wedge b\in L_{4} \wedge (a+b) \in \range{mp}}$
    \item $L^2_1 \gets \set{ (a+b)\ \big|\ a\in L^{1}_{1} \wedge b\in L^{1}_{2} \wedge (a+b) \in \range{mp^2}}$
  \end{itemize}
  \item Succeed if $L^2_1$ contains a $0$, fail otherwise.
\end{enumerate}
Above, take the lists $L^d_i$ to each be a multi-set instead of a set -- for instance, if a specific value of $(a+b)\in\range{mp}$ occurs twice as the sum of elements in $L_1$ and $L_2$, then two such values will be present in $L^1_1$. This makes the algorithm sub-optimal in terms of efficiency, but does not change its success probability and makes it easier to analyze. Define the following random variable as described earlier:
\begin{align*}
  C = \text{number of occurrences of } 0 \text{ in the list } L^{2}_1
\end{align*}
The algorithm succeeds exactly when the value of this non-negative variable $C$ is at least $1$. So the probability of this latter event is what we will be concerned with in our analysis.

For any set of input lists $\set{L_i}$, the value of $C$ may be computed by iterating through every possible value of the indices $\ell_1,\dots,\ell_4\in [n]$, and considering whether the tuple $(L_1[\ell_1],\dots,L_4[\ell_4])$ passes all the ``filters'' of the algorithm and also sums to $0$. For each such $\bell = (\ell_1,\dots,\ell_4)$, we capture this by defining a variable $C_\bell$ that is $1$ if all of the following events occur and is $0$ otherwise:
\begin{align*}
  E_1 &\equiv (L_1[\ell_1] + L_2[\ell_2] \in \range{mp})\\
  E_2 &\equiv (L_3[\ell_3] + L_4[\ell_4] \in \range{mp})\\
  E_3 &\equiv (L_1[\ell_1] + L_2[\ell_2] + L_3[\ell_3] + L_4[\ell_4] \in \range{mp^2})\\
  E_4 &\equiv (L_1[\ell_1] + L_2[\ell_2] + L_3[\ell_3] + L_4[\ell_4] = 0)
\end{align*}
Note that $C_\bell$ is a random variable, with randomness coming from the numbers in the input lists. The variable $C$ can now be written as follows:
\begin{align}
  \label{eq:intro-1}
  C = \sum_{\bell \in [n]^4} C_\bell
\end{align}
We will be interested in the moments of this random variable $C$, which are functions of $m$ and $n$ (since we have already fixed $k$).

\paragraph{Heuristic Analysis.} Wagner's heuristic analysis~\cite{Wagner02} of the $k$-Tree algorithm essentially relies on the following three assumptions (for any large enough integer $s$):
\begin{enumerate}[label=(\alph*),topsep=5pt]
  \item When $\exp{C} = 1$, the algorithm succeeds with constant probability.
  \item When $x$ and $y$ are sampled uniformly at random from $\range{s}$, their sum $(x+y)$ is contained in $\range{sp}$ with probability $p$.
  \item With $x$ and $y$ sampled uniformly at random from $\range{s}$, the distribution of $(x+y)$, when conditioned on being contained in $\range{sp}$, is uniform over $\range{sp}$
\end{enumerate}

\medskip
\noindent Fix some tuple of indices $\bell = (\ell_1,\dots,\ell_4)$. Under assumptions (b) and (c), the expectation of $C_\bell$ can be computed by computing the probability of events $E_1,\dots,E_4$ as follows:
\begin{itemize}[topsep=5pt,itemsep=0pt]
  \item Assumption (b) implies that events $E_1$ and $E_2$ each happens with probability $p$.
  \item Conditioned on these happening, assumption (c) implies that the sums $(L_1[\ell_1]+L_2[\ell_2])$ and $(L_3[\ell_3]+L_4[\ell_4])$ are uniformly distributed over $\range{mp}$.
  \item So by (b), $E_3$ also happens with probability $p$.
  \item Finally, appealing to (c) again, conditioned on the first three events happening, the entire sum is distributed uniformly over $\range{mp^2}$, and so $E_4$ happens with probability $(mp^2)^{-1}$.
\end{itemize}
Thus, we have:
\begin{align}
  \label{eq:intro-2}
  \exp{C_\bell} &= \pr{E_1 \wedge \cdots \wedge E_4}\nonumber\\
               &= \pr{E_1} \cdot \pr{E_2} \cdot \pr{E_3\ |\ E_1\wedge E_2} \cdot \pr{E_4\ |\ E_1\wedge E_2\wedge E_3}\nonumber\\
               &= p \cdot p \cdot p \cdot \frac{1}{mp^2} = \frac{p}{m} = p^4
\end{align}
The expected value of $C$ is then:
\begin{align}
  \label{eq:intro-3}
  \exp{C} = \sum_{\bell \in [n]^4} C_\bell = n^4 \cdot p^4 = c^4
\end{align}
Thus, when $c = 1$ (that is, $n = m^{1/3}$), this expectation is $1$, and by assumption (a), the algorithm works with constant probability. This was the conclusion in \cite{Wagner02}. In providing a rigorous analysis of the $k$-Tree algorithm, our objective is to remove the above assumptions. We show that assumptions (b) and (c), while not strictly true, are close enough to being true. Further, we show that a generalization of assumption (a) is true, which allows us obtain bounds on the success probability of the algorithm for a wide range of values of $c$ rather than just $c=1$.

\paragraph{Relaxing the Assumptions.} To compute the actual expectation of $C_\bell$, we start by first showing that slightly relaxed versions of assumptions (b) and (c) above are true. Starting with assumption (b), it can be shown using some elementary computation \iflncs\else(see \cref{sec:tools}) \fi that:
\begin{align}
  \label{eq:intro-4}
  \prob{x,y\gets \range{s}}{x+y \in \range{sp}}  \in \left(p - \frac{p^2}{4}\right) \pm O\left(\frac{1}{s}\right)
\end{align}
The right-hand side above is not exactly $p$, but for typical values of $p$ and $s$, is quite close to it.

\medskip
For assumption (c), we show that while the distribution of $(x+y)$ as described there is not actually uniform, it is close to the uniform distribution. The specific notion of distance we use, which we refer to as the \emph{Max-Ratio (MR) distance} from Uniform, is defined as follows\footnote{This is again a simplification. For the actual definition, see \iflncs\else \cref{sec:tools}\fi .} for a distribution $D$ over a domain $S$:
\begin{align*}
  \mr(U,D) = \frac{\max_{x\in S} D(x)}{\min_{x\in S} D(x)}
\end{align*}
where $D(x)$ is the probability mass placed on $x$ by the distribution $D$. This distance is at least $1$ (which is achieved if $D$ is the uniform distribution over $S$), and is potentially unbounded.

This notion of distance is particularly beneficial for our analysis for a couple of reasons. First, for any event $E$ over the domain $S$, we can bound the probability of $E$ happening under $D$ in terms of the probability of $E$ happening under the uniform distribution over $S$ as follows:
\begin{align}
  \label{eq:intro-5}
  \mr(U,D)^{-1} \cdot \prob{x\gets S}{E(x)} \leq \prob{x\gets D}{E(x)} \leq \mr(U,D) \cdot \prob{x\gets S}{E(x)}
\end{align}
Second, this distance is easy to compute for distributions defined under conditioning, which can otherwise be difficult to reason about. In particular, suppose $D$ is the distribution of $(x+y)$ conditioned on being in $\range{sp}$ when $x$ and $y$ are uniformly drawn from $\range{s}$. Then, this distance is as follows:
\begin{align}
  \label{eq:intro-6}
  \mr(U,D) &= \frac{\max_{z\in \range{sp}} \prob{x,y\gets\range{s}}{x+y=z\ |\ x+y\in\range{sp}}}{\min_{z\in \range{sp}} \prob{x,y\gets\range{s}}{x+y=z\ |\ x+y\in\range{sp}}}\nonumber\\
           &= \frac{\max_{z\in \range{sp}} \prob{x,y\gets\range{s}}{x+y=z}\cdot \pr{x+y\in\range{sp}}^{-1}}{\min_{z\in \range{sp}} \prob{x,y\gets\range{s}}{x+y=z}\cdot \pr{x+y\in\range{sp}}^{-1}}\nonumber\\
           &= \frac{\max_{z\in \range{sp}} \prob{x,y\gets\range{s}}{x+y=z}}{\min_{z\in \range{sp}} \prob{x,y\gets\range{s}}{x+y=z}}\nonumber\\
           &\leq (1+p) + O\left(\frac{1}{s}\right)
\end{align}
where the second equality follows from the Bayes theorem, and the inequality again follows from elementary computations. Again, for typical values of $p$ and $s$, the right-hand side above is quite close to $1$, indicating that the distribution is close to uniform. Putting together (\ref{eq:intro-5}) and (\ref{eq:intro-6}), with some further computation, we get the following effective relaxation of assumption (c). For any event $E(z)$ defined over domain $\range{sp}$, we have:
\begin{align}
  \label{eq:intro-7}
  \prob{x,y\gets\range{s}}{E(x+y)\ |\ (x+y)\in\range{sp}} \in \prob{z\gets\range{sp}}{E(z)} \cdot \left[ 1 \pm \left(p + O\left(\frac{1}{s}\right) \right) \right]
\end{align}
For simplicity, in the rest of this overview we will ignore the $O(1/s)$ parts in the expressions (\ref{eq:intro-4}) and (\ref{eq:intro-7}) above. It is not a crucial part of the big picture, and in most reasonable parameter settings is much smaller than $p$ anyway.

\paragraph{Computing the Expectation.} With (\ref{eq:intro-4}) and (\ref{eq:intro-7}) in hand, we can now compute the expectation of $C_\bell$. We essentially follow the earlier heuristic argument step-by-step, replacing the assumptions there with their true but relaxed versions.

Denote by $\vx$ the vector $(x_1,\dots,x_4)$, where $x_i\in\range{m}$ will later be identified with $L_i[\ell_i]$; the variables $x^1_1$, $x^1_2$, and $x^2_1$ below will similarly initially be identified with $(x_1+x_2)$, $(x_3+x_4)$, and $(x^1_1+x^1_2)$, respectively. We re-state the events in definition of $C_\bell$ to be parameterized as follows for ease of manipulation:
\begin{align*}
  E_1(x_1,x_2) &\equiv (x_1 + x_2 \in \range{mp})\\
  E_2(x_3,x_4) &\equiv (x_3 + x_4 \in \range{mp})\\
  E_3(x^1_1,x^1_2) &\equiv (x^1_1 + x^1_2 \in \range{mp^2})\\
  E_4(x^2_1) &\equiv (x^2_1 = 0)
\end{align*}
Along the lines of (\ref{eq:intro-2}), the expectation of $C_\bell$ for any $\bell$ can be written as follows:
\begin{align}
  \label{eq:intro-8}
  \exp{C_\bell} &= \prob{x_1,\dots,x_4\gets\range{m}}{\substack{E_1(x_1,x_2) \wedge E_2(x_3,x_4) \wedge \\ E_3(x_1+x_2,x_3+x_4) \wedge E_4(x_1+\cdots+x_4)}}\nonumber\\
  &= \prob{x_1,x_2\gets\range{m}}{E_1(x_1,x_2)} \cdot \prob{x_3,x_4\gets\range{m}}{E_2(x_3,x_4)} \nonumber\\ &\quad\quad \cdot \prob{x_1,\dots,x_4\gets\range{m}}{\substack{E_3(x_1+x_2,x_3+x_4) \wedge E_4(x_1+\cdots+x_4)\ \\|\ E_1(x_1,x_2) \wedge E_2(x_3,x_4)}}
\end{align}
The first two terms in the product above can be bounded using (\ref{eq:intro-4}) as follows:
\begin{align}
  \label{eq:intro-9}
  \prob{x_1,x_2\gets\range{m}}{E_1(x_1,x_2)} = \prob{x_3,x_4\gets\range{m}}{E_2(x_3,x_4)} \approx \left(p - \frac{p^2}{4}\right)
\end{align}
The last term is more complex, but a useful observation here is that the events in the probability expression there only depend on the sums $(x_1+x_2)$ and $(x_3+x_4)$ rather than on the $x_i$'s directly. Further, these events are conditioned on these sums being contained in $\range{mp}$ (that is, conditioned on the events $E_1$ and $E_2$). If assumption (c) had been true, it would have implied that these sums are uniformly distributed over $\range{mp}$, simplifying the probability expression considerably. We can do the same using (\ref{eq:intro-7}), albeit with some loss as follows:
\begin{align}
  \label{eq:intro-10}
  &\prob{x_1,x_2,x_3,x_4\gets\range{m}}{\substack{E_3(x_1+x_2,x_3+x_4) \wedge E_4((x_1+x_2)+(x_3+x_4))\ \\ |\ E_1(x_1,x_2) \wedge E_2(x_3,x_4)}}\nonumber\\
  &\quad\in \prob{x^1_1,x^2_1\gets\range{mp}}{E_3(x^1_1,x^1_2) \wedge E_4(x^1_1+x^1_2)} \cdot (1 \pm p)^2
\end{align}
Putting together (\ref{eq:intro-8}-\ref{eq:intro-10}), we get:
\begin{align}
  \label{eq:intro-11}
  \exp{C_\bell} \in p^2 \cdot (1\pm p)^4 \cdot \prob{x^1_1,x^2_1\gets\range{mp}}{E_3(x^1_1,x^1_2) \wedge E_4(x^1_1+x^1_2)}
\end{align}
The probability expression above is similar to the one we started with in (\ref{eq:intro-8}), except that it is for the case of $k=2$ and the range $\range{mp}$ instead of $k=4$ and range $\range{m}$. For general $k$, we get a similar recursive expression in terms of $k/2$ that enables us to compute the overall bound efficiently, both analytically and computationally. In the present case, we simply have to apply (\ref{eq:intro-4}) and (\ref{eq:intro-7}) in turn once more to get the following final bound:
\begin{align}
  \label{eq:intro-12}
  \exp{C_\bell} \in p^3 \cdot (1\pm p)^6 \cdot \frac{1}{mp^2} \subseteq p^4 \cdot (1\pm O(p))
\end{align}
The above is again quite close to its heuristic evaluation in (\ref{eq:intro-2}). The expectation of $C$ can then be computed as in (\ref{eq:intro-3}) to get:
\begin{align}
  \label{eq:intro-13}
  \exp{C} \in c^4 \cdot (1\pm O(p))
\end{align}
The upper bound on expectation above already gives us an upper bound on the success probability of the $k$-Tree algorithm using the Markov inequality. It remains, then, to show a corresponding lower bound.

\paragraph{Concentration.} Next we move on to proving that (a generalization of) assumption (a) is valid. Our hope here is to show that when the expectation of $C$ is significantly larger than $0$, then with somewhat high probability, we have $C \neq 0$ -- that is, the algorithm succeeds. Clearly, this is not true of arbitrary random variables -- if $C$ was $0$ with probability $(1-n^{-4})$ and $n^4$ with probability $n^{-4}$, it would still have an expected value of $1$. So to show this, we will need to rely on additional properties of $C$.

The property we will use is that $C$ is the sum of the $n^4$ identically distributed indicator random variables $C_\bell$. If this set of random variables had been independent, we could have used a Chernoff-Hoeffding bound to show that $C$ strongly concentrates around its expectation, and thus if its expectation is sufficiently larger than $0$, then it will take non-zero values with large probability. However, these variables are not independent.

Consider, for example, tuples of indices $\bell = (\ell_1,\dots,\ell_4)$ and $\bell' = (\ell_1',\dots,\ell_4')$ that differ only on the value of their fourth element. If $C_\bell = 1$, this implies that $(L_1[\ell_1']+L_2[\ell_2']) = (L_1[\ell_1]+L_2[\ell_2]) \in \range{mp}$. Thus, the tuple $(L_1[\ell_1'],\dots,L_4[\ell_4'])$ already passes one of the filters of the algorithm, and so $C_{\bell'}$ is now more likely to be $1$ than it would have been without conditioning on $C_\bell = 1$. In fact, this example illustrates that the set of random variables $\set{C_\bell}$ are not even pairwise independent. This lack of independence is perhaps the most significant challenge in analyzing the performance of the $k$-Tree algorithm.

Our way around this is the following set of observations. The degree of correlation between $C_\bell$ and $C_{\bell'}$ is proportional to the number of co-ordinates on which $\bell$ and $\bell'$ agree. For instance, if $\bell$ and $\bell'$ do not agree on any co-ordinate, then $C_\bell$ and $C_{\bell'}$ are, in fact, independent. Fortunately, for $t\in[0,4]$, the number of pairs $\bell$ and $\bell'$ that agree on $t$ co-ordinates decreases as $t$ (and thus the correlation) increases. Out of the $n^8$ possible pairs, only $O(n^7)$ pairs agree on $1$ co-ordinate, $O(n^6)$ agree on $2$, and $O(n^5)$ agree on $3$.

Taking advantage of this, we are able to carefully bound the second moment of $C$. Then, we use second-moment-based concentration bounds to show that if the expectation of $C$ is large enough, it will indeed be non-zero with substantial probability.

\paragraph{Computing the Second Moment.} The second moment of $C$ can be written as follows:
\begin{align}
  \label{eq:intro-14}
  \exp{C^2} = \sum_{\bell,\bell'\in[n]^4} \exp{C_\bell\cdot C_{\bell'}} = \sum_{\bell,\bell'\in[n]^4} \pr{C_\bell = 1 \wedge C_{\bell'} = 1}
\end{align}
Due to symmetry, the term corresponding to any $\bell$ and $\bell'$ in the sum above is fully determined by the set of co-ordinates that $\bell$ and $\bell'$ agree on (that is, it does not matter what specific values the $\ell_i$'s and $\ell_i'$'s take once the set of $i$'s on which they agree is determined). To capture this, we define the string $\delta(\bell,\bell')\in\bset^4$ whose $i^{\text{th}}$ co-ordinate is defined as follows:
\begin{align*}
  (\delta(\bar{\ell},\bar{\ell}'))_i = \begin{cases}
    0 &\text{if } \ell_i = \ell_i'\\
    1 &\text{if } \ell_i \neq \ell_i'
  \end{cases}
\end{align*}
We then segregate the pairs $(\bell,\bell')$ in the sum in (\ref{eq:intro-14}) according to the value of $\delta(\bell,\bell')$. For any $s\in\bset^4$, the the number of such pairs with $\delta(\bell,\bell') = s$ is $n^4(n-1)^{wt(s)} \approx n^{4+wt(s)}$, where $wt(s)$ is the Hamming weight of $s$. We can then re-write the second moment as follows:
\begin{align}
  \label{eq:intro-15}
  \exp{C^2} &= \sum_{s\in\bset^4} \sum_{\substack{\bell,\bell'\in[n]^4\\s.t.\ \delta(\bell,\bell') = s}} \pr{C_\bell = 1 \wedge C_{\bell'} = 1}\nonumber\\
            &\approx \sum_{s\in\bset^4} n^{4+wt(s)} \cdot f(s)
\end{align}
where $f(s)$ is defined to be the probability that $C_{\bell} = C_{\bell'} = 1$ for any $\bell$ and $\bell'$ such that $\delta(\bell,\bell') = s$. In computing $f(s)$, our broad approach is similar to the one we took when computing the first moment -- to repeatedly replace intermediate list elements with uniformly random values while computing the probabilities that the filters of the algorithm are passed. Only now, we need to do this simultaneously for two partially dependent executions of the algorithm.

\medskip
We start first with the events $E_1$ as defined above, which capture the first filter applied to the first two elements in the lists. The probability that this event happens in both executions is:
\begin{align}
  \label{eq:intro-16}
  \prob{L_1,L_2}{(L_1[\ell_1]+L_2[\ell_2]) \in \range{mp} \wedge (L_1[\ell_1']+L_2[\ell_2']) \in \range{mp}}
\end{align}
Let $s = \delta(\bell,\bell')$, and denote its bits by $s_1,\dots,s_4$. If $s_1 = s_2 = 0$, then $\ell_1 = \ell_1'$ and $\ell_2=\ell_2'$, meaning the two events above are the same and so by (\ref{eq:intro-9}) the above probability is $\approx p$ (ignoring factors of $(1\pm O(p))$ for now). Similarly if $s_1 = s_2 = 1$, the two events are independent and the probability is $\approx p^2$. If exactly one of $s_1$ and $s_2$ is $1$, then again the probability is roughly $\leq p^2$, due the following fact, which again follows from elementary computations:
\begin{align}
  \label{eq:intro-17}
  \prob{w,x,y\gets\range{s}}{w+x\in\range{sp} \wedge w+y\in\range{sp}} \leq p^2 \cdot \left(1 + O\left(\frac{1}{sp}\right) \right)^2
\end{align}
Overall, we have approximately the following (again ignoring $(1\pm O(p))$ factors):
\begin{align}
  \label{eq:intro-18}
  \prob{L_1,L_2}{E_1(L_1[\ell_1],L_2[\ell_2]) \wedge E_1(L_1[\ell_1'],L_2[\ell_2'])} \leq p^{1+\max(s_0,s_1)} = p^{1+s^1_1}
\end{align}
where we define $s^1_1$ to be $\max(s_1,s_2)$.

\medskip
Next, we look at the joint distribution of the sums $(L_1[\ell_1]+L_2[\ell_2], L_1[\ell_1']+L_2[\ell_2'])$. If $s^1_1 = 0$, then these two sums are equal, and by (\ref{eq:intro-6}) are each $(1+O(p))$-close to being uniformly distributed over $\range{sp}$. If $s_1 = s_2 = 1$, then these are independent and so the joint distribution is now close to being uniform over $\range{mp}\times\range{mp}$. We show by computing the MR distances explicitly that even in the case where exactly one of $s_1$ and $s_2$ is $1$, this joint distribution is close to being uniform over $\range{mp}\times\range{mp}$.

So overall, if $s^1_1 = 0$, this pair of sums is equal with close-to-uniform marginals, and if $s^1_1 = 1$, the pair is close to being independently uniformly distributed over $\range{mp}$. Notice that this is exactly the relationship that the pair $(L_1[\ell_1],L_1[\ell_1'])$ had to $s_1$, or $(L_2[\ell_2],L_2[\ell_2'])$ had to $s_2$, except that the marginal distributions there were over $\range{m}$. The same sequence of arguments can be made with the events $E_2$, where in place of $s^1_1$ we use $s^1_2 = \max(s_3,s_4)$. In effect, this lets us set up a recursive argument where we can replace the intermediate sums with uniform elements from $\range{mp}$, and the string $s$ with $(s^1_1,s^1_2)$.

\medskip
We can then repeat the entire argument above once more (essentially with $k=2$ now), with $s^2_1 = \max(s^1_1,s^1_2)$, to show that the probability of each of the pairs of events corresponding to $E_3$ and $E_4$ happening is roughly bounded by $p^{1+s^2_1}$. Overall, we end up with the following bound:
\begin{align}
  \label{eq:intro-19}
  \pr{C_\bell=1 \wedge C_{\bell'}=1} = f(s) &\leq p^{1+s^1_1} \cdot p^{1+s^1_2} \cdot p^{1+s^2_1} \cdot p^{1+s^2_1} \cdot (1+O(p)) \nonumber\\ &= p^{4+ s^1_1+s^1_2+ 2s^2_1} \cdot (1+O(p))
\end{align}
Combining this bound with (\ref{eq:intro-15}), we get:
\begin{align}
  \label{eq:intro-20}
  \exp{C^2} &\leq \sum_{s\in\bset^4} n^{4+wt(s)} \cdot p^{4+ s^1_1+s^1_2+ 2s^2_1} \cdot (1+O(p))\nonumber\\
            &= c^4 \cdot (1+O(p)) \cdot \sum_{s\in\bset^4} c^{wt(s)} \cdot p^{s^1_1+s^1_2+2s^2_1 - wt(s)}
\end{align}
Using graph-theoretic arguments, we can show that for any $s \not\in \set{0^4,1^4}$, the difference $(s^1_1+s^1_2+2s^2_1 - wt(s))$ is at least $1$, and so $p^{s^1_1+s^1_2+2s^2_1 - wt(s)} \leq p$. For $s \in \set{0^4,1^4}$, this difference is $0$. So we can bound the above sum as follows:
\begin{align}
  \label{eq:intro-21}
  \sum_{s\in\bset^4} c^{wt(s)} \cdot p^{s^1_1+s^1_2+2s^2_1 - wt(s)} &\leq 1 + c^4 + \sum_{s\in\bset^4\setminus \set{0^4,1^4}} c^{wt(s)}\cdot p \nonumber\\
                                                                    &\leq 1 + c^4 + p \cdot \sum_{s\in\bset^4} c^{wt(s)}\nonumber\\ 
                                                                    &= 1 + c^4 + p\cdot (1+c)^4
\end{align}
Overall, from (\ref{eq:intro-20}) and (\ref{eq:intro-21}), we get the following bound on the second moment:
\begin{align}
  \label{eq:intro-22}
  \exp{C^2} \leq c^4 \cdot \left(1 + c^4 + p\cdot (1+c)^4 \right) \cdot (1+O(p))  
\end{align}

\paragraph{Computing the Lower Bound.} With the bounds on the first two moments from (\ref{eq:intro-13}) and (\ref{eq:intro-22}), the required lower bound can be obtained using the Paley-Zygmund inequality, which states that for any non-negative random variable $Z$ and any $\theta\in [0,1]$:
\begin{align*}
  \pr{Z > \theta \exp{Z}} \geq (1-\theta)^2 \frac{\exp{Z}^2}{\exp{Z^2}}
\end{align*}
Applying this bound to $C$ with $\theta = 0$, we get:
\begin{align}
  \label{eq:intro-23}
  \pr{\text{$k$-Tree succeeds}} = \pr{C > 0 } &\geq \frac{\exp{C}^2}{\exp{C^2}} \nonumber\\ &\geq \frac{c^8\cdot (1-O(p))}{c^4 (1+c^4+p(1+c)^4) (1+O(p))} \nonumber\\ &\geq \frac{c^4}{1+c^4+p\cdot (1+c)^4} \cdot (1-O(p))
\end{align}

\paragraph{Tightening the Bound.} The above approach naturally generalizes to any $k$ that is a power of $2$ to give the following lower bound:
\begin{align*}
  \pr{\text{$k$-Tree succeeds}} \geq \frac{c^k}{1+c^k+p\cdot (1+c)^k} \cdot (1-O(kp))
\end{align*}
While their asymptotics for any fixed $k$ are the same, for moderately large values of $k$ and some reasonable concrete values of $m$ and $n$, the above bound ends up being much weaker than the lower bound actually stated in \cref{infthm:ktree}. This is because, even when $c = 1$, the $p(1+c)^k$ term in the denominator quickly starts to dominate. The bound stated in \cref{infthm:ktree} remains meaningful for a wider range of concrete values of the parameters, and is obtained by performing a more careful analysis of the sum in (\ref{eq:intro-20}), using a recursive argument to bound the entire sum together rather than each term separately.

The constants in the statement of \cref{infthm:ktree} are obtained by carefully tracking the constants that come up in the course of the above analysis. Even so, the constants stated there are sub-optimal due to limits on human tolerance, and the algorithmic implementation of this proof (as captured in \cref{infthm:program}) improves upon them significantly.


\subsection{Related Work}
\label{sec:related}

The worst-case version of the $k$-SUM problem has been studied extensively in the literature on algorithms~\cite[\dots]{HS74,BDP08,DSW18,Chan20}, data structures~\cite[\dots]{KP19,GGHPV20,CL23}, and complexity theory~\cite[\dots]{GO95,BHP01,SEO03,nseth,ksum_decision_tree,Erickson95,AC05,Pat10,PW10,GP18,ABHS19}. In the study of the fine-grained complexity of problems within $\P$, in particular, connections have been discovered between the complexity of this problem and those of various important problems from computational geometry, graph theory, data structures, etc.~\cite{GO95,BHP01,SEO03,Pat10,AW14,KPP16}. A simple meet-in-the-middle algorithm solves this problem in $\Otilde(n^{\ceil{k/2}})$ time~\cite{HS74}. The best algorithms known are faster than this by only a small polylog factor~\cite{BDP08,GP18,Chan20}, or run in time $\Otilde(m+n)$~\cite{Bri17,JW19}. It has been conjectured that the best worst-case algorithm for this problem runs in time $n^{\ceil{k/2}-o(1)}$~\cite{GO95,AL13}.

\medskip
For average-case $k$-SUM, non-trivial algorithms (beyond known worst-case algorithms) have so far been restricted to Wagner's $k$-Tree algorithm~\cite{Wagner02} and its extensions to using smaller lists (at the cost of running time)~\cite{MS12}, to values of $k$ that are not power of $2$~\cite{NS15,Dinur19}, to generalizations with better time-memory tradeoffs~\cite{NS15,BK17,Dinur19}, and to the quantum setting~\cite{GNS18}. These algorithms have found extensive use across cryptography and cryptanalysis~\cite[\dots]{Schnorr01,Wagner02,DEFKLNS19,LS19,BLLOR22}. They are also closely related to some of the best algorithms for the Learning Parity with Noise problem~\cite{BKW03}.

Lattice-based conditional lower bounds are known that indicate that the $k$-Tree algorithm for $k$-SUM is optimal in its asymptotic dependence on $k$~\cite{BSV21}. Some conditional lower bounds are also known for the complexity of the $k$-SUM problem given lists of size close to $m^{1/k}$~\cite{DKK21,ASSVV24}.

\paragraph{Comparison with Existing Analyses.} As noted earlier, Wagner's original argument~\cite{Wagner02} of the correctness of the $k$-Tree algorithm was heuristic and relied on assumptions regarding the distribution and independence of elements in intermediate lists computed by the algorithm. The conclusion of this argument was that the algorithm works with constant success probability when given input lists of size $m^{1/(\log{k}+1)}$, and that in this case it runs in time $O(k\cdot m^{1/(\log{k}+1)})$.

\medskip
This argument was made rigorous by Minder and Sinclair~\cite{MS12} for the variant of the $k$-Tree algorithm that sovles the $k$-SUM problem over the group $\Int_2^m$ (in which case it is referred to as the $k$-XOR problem). In this case, Wagner's assumption regarding the uniform distribution of intermediate list elements is immediately seen to be true. Minder and Sinclair also follow the approach (described in \cref{sec:overview}) of then computing the first two moments of the random variable counting the number of solutions found, and then using these to prove concentration bounds. In the case of $k$-XOR, the process of computing these moments turns out to be much simpler due to the fact that the XOR of any arbitrary vector with a uniformly random vector results again in a uniformly random vector. Some of the questions that come up during this analysis of $k$-Tree over $\Int_2^m$ also come up in the analysis of an approach to hashing called Simple Tabulation Hashing (see, for example, \cite{PT12}), but the implications of the results there for Minder and Sinclair's analysis are not immediately clear to us.

\medskip
The first rigorous analysis of the $k$-Tree algorithm over integers\footnote{We continue to ignore the distinction between $k$-Tree over integers and over $\Int_m$, following the discussion in \cref{rem:zm}.} was by Lyubashevsky~\cite{Lyubashevsky05}, who showed that the algorithm works with high probability if the input lists are of size at least $\Omega(k^2\log{k} \cdot m^{2/\log{k}})$. This was improved by Shallue~\cite{Shallue08} to lists of size at least $\Omega(k\cdot m^{1/\log{k}})$. Both of these actually study a variant of the $k$-Tree algorithm where the factor by which the permitted range shrinks at each level of the tree is set to $p = m^{-1/\log{k}}$ instead of $m^{-1/(\log{k}+1)}$. In this case, the final list is only permitted to contain $0$'s, and when the input list size is $m^{1/\log{k}}$, the expected number of solutions found by the algorithm is also $\Theta(m^{1/\log{k}})$. Both Lyubashevsky and Shallue were interested in using the $k$-Tree algorithm for very large values of $k$, and so the difference between $\log{k}$ and $(\log{k}+1)$ in the exponent were not significant in their applications. It is possible that Shallue's analysis can be made to work for the original $k$-Tree algorithm, but this is to be verified and even then it is not clear whether the bounds will work for list sizes smaller than $k\cdot m^{1/(\log{k}+1)}$.

Both of the above papers take similar approaches in their analysis, and we briefly describe Shallue's here. Similar to us, Shallue starts by showing that the elements in the intermediate lists are close to uniform; the distance measure used there is also an element-wise bound on differences in probability mass, but the whereas MR distance measures relative difference, they measure absolute difference. We believe our choice is better for concrete parameters, but asymptotically these choices are likely equivalent. The larger difference is in how they deal with the dependence between list elements. They start with the observation that even though there are dependencies between list elements, it is possible to find a large enough subset of elements in each list such that all correlations among them are relatively small. They then use martingale-based concentration bounds to recursively bound the size of each intermediate list.

Their utilization of correlation bounds on larger sets of variables eventually results in a better dependence of the probability bound on the list-size overhead. They show that if the input lists are of size $c\cdot m^{1/\log{k}}$ for some $c \geq k$, then the probability of failure is roughly at most $(m^{1/\log{k}} e^{-c/1024})$, whereas our bound for lists of size $c \cdot m^{1/(\log{k}+1)}$ is around $c^{-k}$. For very large values of $c$, the former will be much smaller, but for moderate values $c= \Theta(1)$, the latter bound is better.

\medskip
A different approach to analyzing the $k$-Tree algorithm was taken by Joux, Kippen, and Loss~\cite{JKL24}. They get around the uniformity and independence issues by modifying the $k$-Tree algorithm in such a way that these properties actually hold as assumed by the heuristic analysis. This makes the algorithm slightly worse than the original, but easier to analyze. They then show that this modified algorithm works with probability roughly larger than $(1-e^{-\poly(k)})$ if the input list sizes are $n = \beta(k) \cdot m^{1/(\log{k}+1)}$, where $\beta(k) \approx (1.26)^{\frac{\ell-1}{\ell+1}} \cdot 6^{\frac{(\ell-1)(\ell+2)}{(\ell+1)}}$, with $\ell = \log{k}$. Further, their modified algorithm runs in time $\Theta(kn)$ in this case. This again results in very good bounds on the probability of success, but the bounds only work for large values of $c$, especially for larger values of $k$ -- for instance, $\beta(4) \approx 3.563$ and $\beta(1024) \approx 7964$. 

\medskip
\noindent In summary, the following are the most significant points of comparison between our results and the existing analyses of Shallue~\cite{Shallue08} and Joux et al.~\cite{JKL24} of the $k$-Tree algorithm over integers:
\begin{itemize}
  \item The results in the other papers do not give meaningful bounds for $c \leq 1$ -- that is, when the size of the lists is $n = m^{1/(\log{k}+1)}$, as suggested by Wagner's heuristic analysis, or smaller -- while ours do. This is also true for some range of values of $c$ larger than $1$, with the range depending on $k$.
  \item For sufficiently larger list sizes ($c \gg 1$), the approaches in the other papers can likely show lower bounds on the success probability that are much closer to $1$ than our bounds can. This gap can be closed using certain optimizations to the $k$-Tree algorithm as described in \cref{rem:optimisations}, though that does not vindicate our analysis itself.
\end{itemize}
We present comparisons of our results with these in \cref{fig:comparison}, making optimistic assumptions regarding the constants in the bounds of the other papers, as well as regarding the regime of their validity, wherever relevant.

\begin{figure}[h!]
  \centering
  \includegraphics[width=0.8\textwidth]{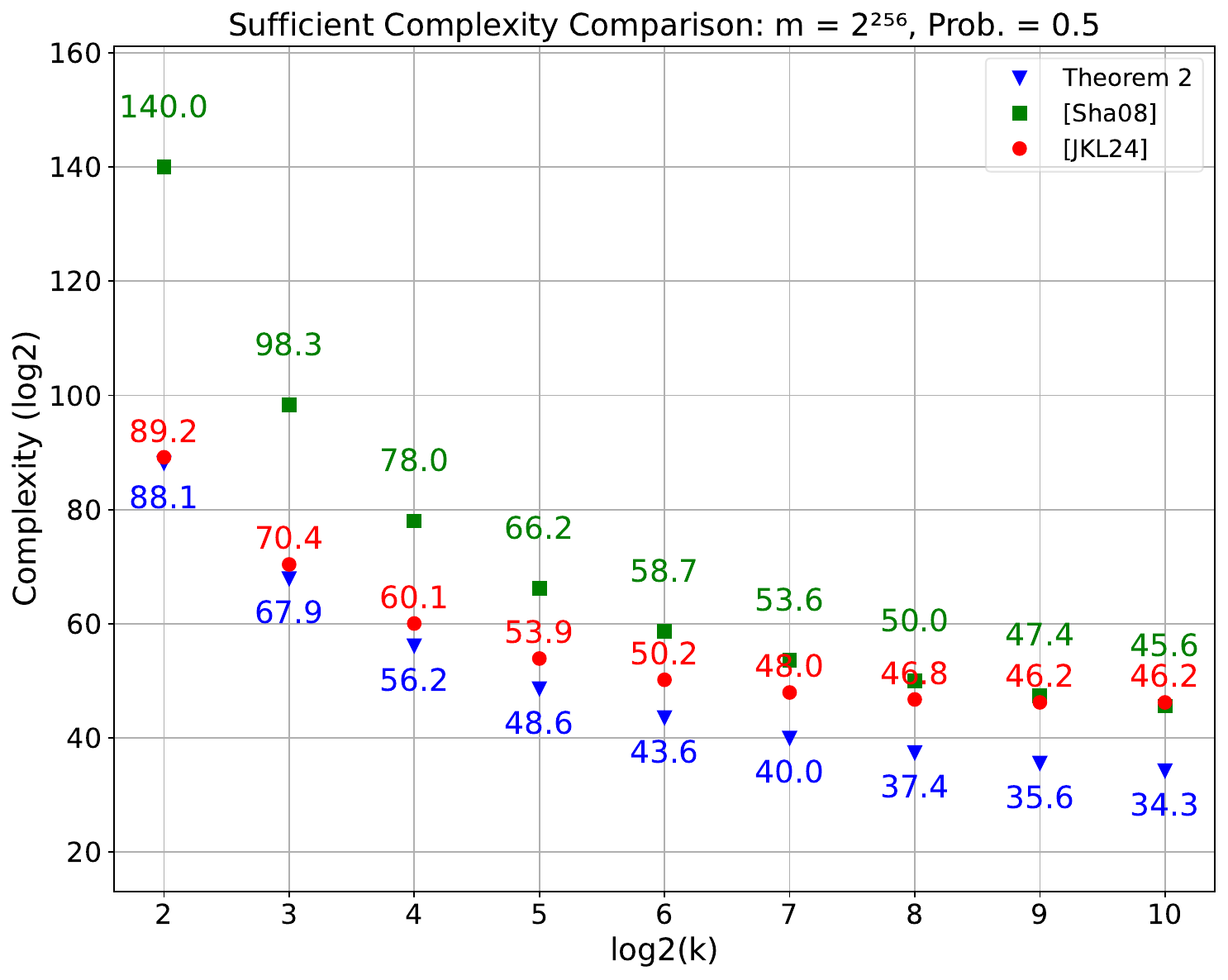}
  \caption{Comparison of bounds produced by \cref{infthm:program} with those of \cite{Shallue08} and \cite{JKL24}. For the other papers, we plot the minimum complexity of $k$-Tree for which they show non-trivial lower bounds on success probability (usually close to $1$). In our case, we plot the smallest complexity of $k$-Tree for which the lower bound on success probability produced by \cref{infthm:program} is at least $0.5$.}
  \label{fig:comparison}
\end{figure}


\section{Analysis}
\label{sec:analysis}


In this section, we present our analysis of Wagner's $k$-Tree algorithm for the $k$-SUM problem over integers~\cite{Wagner02}. We present the algorithm in full detail in \cref{fig:ktree}. It has four parameters -- the number of lists ($k$), the size of each list ($n$), the range the numbers are drawn in ($m$), and an internal filtering parameter $p$.
We then state our main claim about its success probability in \cref{thm:ktree}, state its significant corollaries, and then prove the theorem in the rest of the section.

\paragraph{Notation.} For any $m\in\Nat$, we denote by $[m]$ the set of natural numbers $\set{1,2,\dots,m}$, and by $\range{m}$ the set of integers $\set{-\floor{\frac{m}{2}},\dots,\floor{\frac{m}{2}}}$. Below, the term ``list'' simply refers to an array indexed by natural numbers in some range $[n]$. A list might contain duplicate entries, and in particular is to be distinguished from a set.

\begin{figure}[h!]
  \centering
  \begin{mdframed}
    \begin{center}
      \textbf{\underline{The $k$-Tree algorithm}}
    \end{center}
    \vspace{0.5em}

    \medskip
    \textbf{Parameters:} $k, n, m \in \Nat$, $p\in (0,1]$, with $k$ being a power of $2$

    \medskip
    \textbf{Input:} Lists $L_1,\dots,L_k$, each consisting of $n$ integers from $\range{m}$

    \medskip
    \textbf{Output:} Indices $\ell_1,\dots,\ell_k\in[n]$, or a failure symbol $\bot$

    \medskip
    \textbf{Subroutines:}
    \begin{itemize}[topsep=5pt,itemsep=0pt]
      \item Merge($L_a$, $L_b$, $\tau$): On input lists $L_a$ and $L_b$ of integers and threshold $\tau\in\Reals$, outputs a list consisting of all $(a+b)$ such that $a\in L_a$, $b\in L_b$, and $\abs{a+b} \leq \tau$. This output maintains duplicates -- if multiples pairs $(a\in L_a,b\in L_b)$ have the same sum, a copy of that sum is included in the output list for each pair.
    \end{itemize}
    
    \medskip
    \textbf{Procedure:}
    \begin{enumerate}[topsep=5pt,itemsep=0pt]
      \item For each $i\in[k]$, denote the list $L_i$ by $L^0_i$
      \item Set $\tau \gets m/2$
      \item For $d$ from $1$ to $\log{k}$:
      \begin{itemize}[topsep=0pt,itemsep=5pt]
        \item Set $\tau \gets p \cdot \tau$
        \item For $i \in \left[\frac{k}{2^d}\right]$: set $L^d_i \gets \mathrm{Merge}(L^{d-1}_{2i-1},L^{d-1}_{2i},\tau)$
      \end{itemize}
      \item If $L^{\log{k}}_1$ contains $0$, output the indices in the input lists that led to this sum. Otherwise output $\bot$.
    \end{enumerate}

    \medskip
    \textbf{Remarks:}
    \begin{itemize}[topsep=5pt,itemsep=0pt]
      \item In order to perform its last step, the algorithm additionally needs to keep track of which elements of the input lists contribute to each sum in the intermediate lists. We leave this bookkeeping out of our description for simplicity.
      
      \item Note that if any list considered in the course of the algorithm is empty, the algorithm is eventually bound to output $\bot$.
      \item Above, we only describe the behavior of the $\mathrm{Merge}$ algorithm rather than specifying the procedure it follows because it can be implemented in various ways, with the most efficient choice depending on the parameters of the problem and the execution environment. Also see \cref{rem:complexity}.
    \end{itemize}
    
    \medskip
  \end{mdframed}
  \caption{The $k$-Tree algorithm over Integers}
  \label{fig:ktree}
\end{figure}

\begin{theorem}
  \label{thm:ktree}
  Consider any $k, n, m\in \Nat$, where $k \geq 4$ is a power of $2$ and $m > 30^{\log{k}+1}$. Set $p = m^{\frac{-1}{\log{k}+1}}$ and $c = p\cdot n$. Consider $k$ lists $L_1,\dots,L_k$, each consisting of $n$ uniformly random integers from the range $\range{m}$. The $k$-Tree algorithm (as in \cref{fig:ktree}) with these parameters, denoted by $\ktree$, satisfies the following:
  \begin{itemize}
    \item \textbf{Success Probability.} Its probability of success is bounded as follows:
    \begin{align*}
      \frac{1}{c^{-k} + \left(1+\frac{k}{n} \right)^k}\cdot (1-150p)^{k} \leq \prob{\substack{L_1,\dots,L_k}}{\substack{\ktree(L_1,\dots,L_k)\\\text{ outputs } (\ell_1,\dots,\ell_k)\\\text{such that } \sum_i L_i[\ell_i] = 0 }} \leq c^k \cdot (1 + 37p)^k
    \end{align*}
    \item \textbf{Complexity.} Its expected complexity is bounded as follows:
    \begin{align*}
      \expec{\substack{L_1,\dots,L_k}}{\substack{\text{Total size of all lists}\\\text{involved in}\\\ktree(L_1,\dots,L_k)}} \in k n \cdot \left( 1 + \sum_{d\in[\log{k}]} \frac{c^{2^d-1}}{2^d} \right) \cdot \left(1 \pm 37p \right)^{k-1}
    \end{align*}
  \end{itemize}
\end{theorem}

\begin{corollary}
  \label{cor:ktree}
  Consider functions $k,n:\Nat\ra\Nat$, where for any $m$, $k(m) \geq 4$ is a power of $2$. Set $p(m) = m^{\frac{-1}{\log{k(m)}+1}}$ and $c(m) = p(m)\cdot n(m)$. Further, suppose $k = o(1/p)$ and $k = o(n^{1/2})$. With $k = k(m)$, $n = n(m)$, $m$, and $p(m)$ as its parameters, the probability of success of the $k$-Tree algorithm, denoted $\ktree$, is bounded as follows:
    \begin{align*}
      \frac{c^k}{1 + c^k}\cdot (1-o(1)) \leq \pr{\text{$\ktree$ succeeds}} \leq c^k \cdot (1 + o(1))
    \end{align*}
    Its complexity is bounded as follows:
    \begin{align*}
      \text{Complexity of $\ktree$} \in k n \cdot \left( 1 + \sum_{d\in[\log{k}]} \frac{c^{2^d-1}}{2^d} \right) \cdot \left(1 \pm o(1) \right)
    \end{align*}
\end{corollary}

\medskip
The rest of this section is a proof of \cref{thm:ktree}. Our approach is non-trivial, but elementary. We simply compute the first two moments of the random variable that counts the number of occurrences of $0$ in the final list produced by the algorithm, and then use standard tail bounds to bound the probability that this variable is non-zero.

\medskip
Hereon, we adopt the context of the $k$-Tree algorithm from \cref{fig:ktree}, and use notation established in its description in the general text as well. For any set of parameters $k, n, m, p$, we define the aforementioned random variable as follows (with randomness coming from the choices of the $L_i$'s):
\begin{align*}
  C_{k,n,m,p} = \text{number of occurrences of } 0 \text{ in the list } L^{\log{k}}_1
\end{align*}

We show the following properties of the lower moments of this random variable, use these to prove \cref{thm:ktree}, and then later in the section prove the propositions themselves.

\begin{proposition}
  \label{prop:c-exp}
  Consider any valid set of parameters $k$, $n$, $m$, and $p$ such that $m > 7k$, $p\cdot k > (7/m)^{(k/2-1)}$, $m p^{\log{k}} > 30$, and $p < 1/30$. We have:
  \begin{align*}
    \exp{C_{k,n,m,p}} \in \frac{n^kp^k}{mp^{\log{k}+1}} \cdot (1\pm p)^{k-1} \left(1 \pm \frac{35}{mp^{\log{k}}} \right)^{k-1}
  \end{align*}
\end{proposition}

The second moment of this variable is most conveniently bounded in terms of a function that is again recursively defined, and is a refinement of $T_{\mu,n}$ defined above. For any $\mu,\nu\in (0,1)$ and $n\in \Nat$, the function $T_{\mu,n,\nu}:\Nat\ra\Reals$ is defined on inputs that are powers of $2$. The base is its value on $1$:
\begin{align*}
  T_{\mu,n,\nu}(1) = n\cdot\mu\cdot\nu
\end{align*}
For any $k \geq 2$ that is a power of $2$, it is defined as follows:
\begin{align*}
  T_{\mu,n,\nu}(k) = T_{\mu,n,\nu}(k/2) \left( \frac{T_{\mu,n,\nu}(k/2)}{\nu} + 2\mu \right)
\end{align*}
As the above is a quadratic recursion, there is no general closed-form expression for $T_{\mu,n,\nu}(k)$. We prove upper bounds on its value later in the course of the proof of \cref{thm:ktree}.

\begin{proposition}
  \label{prop:c-var}
  Consider any valid set of parameters $k$, $n$, $m$, and $p$ such that $m p^{\log{k}-1} > 30$ and $p < 1/30$. Define the following:
  \begin{align*}
    \mu = p \cdot \left[ (1-p)^{-1} \left(1 + \frac{35}{mp^{\log{k}}}\right) \right]\quad \text{and} \quad
    \nu = \frac{1}{mp^{\log{k}+1}}
  \end{align*}
  Then, we have: 
  \begin{align*}
    \exp{C_{k,n,m,p}^2} \leq (n^k\mu^k\nu) \cdot (1+T_{\mu,n,\nu}(k)) 
  \end{align*}
\end{proposition}

\medskip
In order to discuss the complexity of the algorithm, for any set of parameters $k, n, m, p$, we define the following random variable that captures the total size of all the lists in an execution of the algorithm:
\begin{align*}
  \Lambda_{k,n,m,p} = \sum_{d=0}^{\log{k}} \size{L^d_i}
\end{align*}

\begin{proposition}
  \label{prop:l}
  Consider any valid set of parameters $k$, $n$, $m$, and $p$ such that $m > 7k$, $p\cdot k > (7/m)^{(k/2-1)}$, $m p^{\log{k}} > 30$, and $p < 1/30$. We have:
  \begin{align*}
    \exp{\Lambda_{k,n,m,p}} &\in k n \cdot \left( 1 + \sum_{d\in[\log{k}]} \frac{(np)^{2^d-1}}{2^d} \right) \cdot (1\pm p)^{k-1} \left(1 \pm \frac{35}{mp^{\log{k}}} \right)^{k-1}
  \end{align*}
\end{proposition}

\medskip
\noindent With these propositions, we can now prove our main theorem.
\begin{proofof}{\cref{thm:ktree}}
  The lower bound on the success probability in \cref{thm:ktree} follows from applying a suitable second-moment-based tail-bound to the random variable $C_{k,n,m,p}$, and then appropriately bounding the function $T_{\mu,n,\nu}$ defined above. The lower bound follows from a simple Markov bound, and the bounds on the complexity follow from \cref{prop:l}.

  \paragraph{Lower bound.} For any $k$, $n$, $m$, and $p$, each entry in the list $L^{\log{k}}_1$ is a sum of one entry each from the input lists $L_1,\dots,L_k$. So whenever the non-negative random variable $C_{k,n,m,p}$ (denoted simply by $C$ hereon) is non-zero, there is at least one entry in $L^{\log{k}}_1$ that is $0$, and the execution $\ktree_p(L_1,\dots,L_k)$ finds indices $(\ell_1,\dots,\ell_k)$ such that $\sum_i L_i[\ell_i] = 0$. So all we need to do is to bound the probability that this random variable is non-zero. This we do using the Paley-Zygmund inequality.

  \begin{lemma}[Paley-Zygmund Inequality, see e.g. {\cite[Section 2.3]{roch_mdp_2024}}]
    \label{lem:paley-zygmund}
    For any non-negative random variable $Z$ and any $\theta\in [0,1]$,
    \begin{align*}
      \pr{Z > \theta \exp{Z}} \geq (1-\theta)^2 \frac{\exp{Z}^2}{\exp{Z^2}}
    \end{align*}
  \end{lemma}

  Substituting the values $p = m^{\frac{-1}{\log{k}+1}}$ and $n = c\cdot m^{\frac{1}{\log{k}+1}}$, and observing that the conditions $k \geq 4$ and $m > 30^{\log{k}+1}$ guarantee the requirements in its hypotheses, we get the following bound from \cref{prop:c-exp}:
  \begin{align}
    \label{eq:thm-ktree-1}
    \exp{C} \in c^k \cdot (1\pm p)^{k} \left(1 \pm 35p \right)^{k}
  \end{align}
  For brevity, denote the function $T_{\mu,n,1}$ that was defined above by $T_{\mu,n}$. We get the following from \cref{prop:c-var}:
  \begin{align}
    \label{eq:thm-ktree-2}
    \exp{C^2} \leq c^k \cdot (1+T_{\mu,n}(k)) \cdot  \left[ (1-p)^{-1} \left(1 + 35p\right) \right]^k
  \end{align}
  where $\mu = p \cdot \left[ (1-p)^{-1} \left(1 + 35p\right) \right]$.
  We bound this more concretely using the claim below, which we prove after the completion of the current proof.
  \begin{claim}
    \label{claim:T}
    For any $\mu\in(0,1)$ and $n,k\in\Nat$, we have:
    \begin{align*}
      T_{\mu,n}(k) \leq n^k\mu^k \cdot \left(1 + \frac{k}{n}\right)^k
    \end{align*}
  \end{claim}
  \noindent Putting together (\ref{eq:thm-ktree-2}) and \cref{claim:T}, we have:
  \begin{align}
    \label{eq:thm-ktree-3}
    \exp{C^2} &\leq c^k \cdot \left(1+c^k \cdot \left[ (1-p)^{-1} \left(1 + 35p\right) \right]^k \cdot \left(1+\frac{k}{n} \right)^k \right) \cdot  \left[ (1-p)^{-1} \left(1 + 35p\right) \right]^k\nonumber\\
    &\leq c^{2k} \left(c^{-k} + \left(1+\frac{k}{n} \right)^k \right)  \left[ (1-p)^{-1} \left(1 + 35p\right) \right]^{2k}
  \end{align}

  \medskip
  \noindent Applying \cref{lem:paley-zygmund} with $\theta = 0$ and the bounds from (\ref{eq:thm-ktree-1}) and (\ref{eq:thm-ktree-3}), we get:
  \begin{align*}
    &\pr{C > 0} \geq \frac{\exp{C}^2}{\exp{C^2}}\\
    &\quad \geq \frac{c^{2k} \cdot (1- p)^{2k} \left(1 - 35p \right)^{2k}}{c^{2k} \left(c^{-k} + \left(1+\frac{k}{n} \right)^k \right)  \left[ (1-p)^{-1} \left(1 + 35p\right) \right]^{2k}}\\
                         &\quad \geq \frac{1}{c^{-k} + \left(1+\frac{k}{n} \right)^k} \cdot (1-p)^{4k} (1-35p)^{2k} (1+35p)^{-2k}\\
    &\quad \geq \frac{1}{c^{-k} + \left(1+\frac{k}{n} \right)^k}\cdot (1-150p)^{k}
  \end{align*}
  as required.
  
  \paragraph{Upper bound.} To get the upper bound on the success probability, we apply the Markov bound using (\ref{eq:thm-ktree-1}) as follows:
  \begin{align*}
    \pr{C \geq 1} \leq \exp{C} \leq c^k \cdot (1 + 37p)^k
  \end{align*}
  
  \paragraph{Complexity.} The bounds on the complexity of the algorithm follow directly from \cref{prop:l}. By our setting of $p$, we have $mp^{\log{k}+1} = 1$, and using the fact that $p < 1/30$, we get:
  \begin{align*}
    \exp{\Lambda_{k,n,m,p}} &\in kn \cdot \left( 1 + \sum_{d\in[\log{k}]} \frac{(np)^{2^d-1}}{2^d} \right) \cdot (1 \pm p)^{k-1} \left(1 \pm 35p \right)^{k-1}\\
                            &\subseteq kn \cdot \left( 1 + \sum_{d\in[\log{k}]} \frac{c^{2^d-1}}{2^d} \right) \cdot \left(1 \pm 37p \right)^{k-1}
  \end{align*}
  This proves the theorem.

\end{proofof}

\begin{proofof}{\cref{claim:T}}
  Fix some values of $n$ and $\mu$, and denote $T_{\mu,n}$ simply by $T$. Recall that $T$ is defined as follows:
  \begin{align*}
    T(1) &= n\cdot \mu\\
    T(k) &= T(k/2) \left( T(k/2) + 2\mu \right)
  \end{align*}
  The claim is that for every $k$ (that is a power of $2$), $T(k) \leq (n\mu + k\mu)^k$. This can be verified to be true for $k = 1$ and $2$. For some $k\geq 4$, suppose this is true for $k/2$. Then, we have:
  \begin{align*}
    T(k) &= T(k/2) (T(k/2)+2\mu) \\
         &\leq (n\mu + (k/2)\mu)^{k/2} (n\mu + (k/2)\mu + 2\mu)^{k/2}\\
         &\leq (n\mu + k\mu)^{k/2} (n\mu + k\mu)^{k/2}\\
         &= (n\mu+k\mu)^k
  \end{align*}
  which proves the claim.
\end{proofof}

\paragraph{Section Outline.} In \cref{sec:tools}, we define a useful notion of distance between probability distributions and establish some facts about distributions over bounded integer intervals. In \cref{sec:c-exp}, we compute the first moment of the above random variable to prove \cref{prop:c-exp}, and in \cref{sec:c-var} we compute its second moment to prove \cref{prop:c-var}.

\subsection{Tools}
\label{sec:tools}

Here, we set up some notational conventions, definitions, and propositions that will be useful at multiple points in the proofs of \cref{prop:c-exp,prop:c-var}.

\paragraph{Distance between distributions.} For any distribution $D$ over a domain $\cX$ and any $x\in\cX$ (or $S\subseteq \cX$), we denote by $D(x)$ (resp. $D(S)$) the probability mass placed by $D$ on $x$ (resp. $S$). Given distributions $D$ and $\hat{D}$, we denote by $(D\otimes \hat{D})$ the (``direct product'') distribution of $(x,y)$ where $x\gets D$ and $y\gets \hat{D}$ are sampled independently.  

We will use the following notion of distance between probability distributions. Related notions of distance have found significant use in differential privacy~\cite{DMNS16}, analysis of lattice algorithms~\cite{BLRLSSS18}, etc.

\begin{definition}[Max-Ratio Distance]
  \label{def:mr-dist}
  Consider two distributions $D_0$ and $D_1$ over a finite domain $\cX$ that have the same support. The \emph{Max-Ratio (MR) distance} between these, denoted by $\mr(D_0,D_1)$, is the smallest $\lambda \geq 1$ such that for all $x$, we have $D_1(x) \in [\lambda^{-1} \cdot D_0(x), \lambda \cdot D_0(x)]$. 
\end{definition}

Note that the above distance is symmetric, and is only defined if both distributions have the same support. All pairs of distributions we will consider in our work will indeed have the same support, and we will leave out stating this condition explicitly in the hypotheses of our statements below. The following are a few other easily observed facts about this distance that we will use.

\begin{fact}
  \label{fact:mr-1}
  Consider any distributions $D_0$ and $D_1$ over $\cX$ with $\mr(D_0,D_1) = \lambda$. For any subset $S\subseteq \cX$, we have $D_1(S) \in [\lambda^{-1} \cdot D_0(S), \lambda \cdot D_0(S)]$.
\end{fact}

\begin{fact}
  \label{fact:mr-2}
  Consider any distributions $D_0$, $D_1$, $\hat{D}_0$, and $\hat{D}_1$ over appropriate domains. We have:
  \begin{align*}
    \mr(D_0\otimes \hat{D}_0, D_1 \otimes \hat{D}_1) \leq \mr(D_0,D_1) \cdot \mr(\hat{D}_0,\hat{D}_1)
  \end{align*}
\end{fact}

The following proposition will be especially useful for us since we will be repeatedly bounding the distance of distributions from the uniform distribution over their support.

\begin{proposition}
  \label{prop:unif-dist}
  Consider any distribution $D$, denote its support by $S$, and let $U$ be the uniform distribution over $S$. Then, we have:
  \begin{align*}
    \mr(U,D) \leq \frac{\max_{x\in S} D(x)}{\min_{x\in S} D(x)}
  \end{align*}
\end{proposition}

\begin{proofof}{\cref{prop:unif-dist}}
  Observe that for any $x\in S$, we necessarily have $U(x) \geq min_{y\in S}D(y)$ and $U(x) \leq \max_{z\in S}D(z)$. If either of these is not the case, the sum of the values of $D(x)$ over all values of $x\in S$ cannot be equal to $1$. Thus, for any $x\in S$,
  \begin{align*}
    D(x) = U(x) \cdot \frac{D(x)}{U(x)} \in \left[ U(x) \cdot \frac{\min_{y\in S} D(y)}{\max_{z\in S} D(z)}, U(x) \cdot \frac{\max_{z\in S} D(z)}{\min_{y\in S} D(y)} \right]
  \end{align*}
  This proves the proposition. Note that above we rely on the fact that $U(x)$ is non-zero for $x\in S$.
\end{proofof}

\paragraph{Distributions of Sums.} For any distribution $D$ over integers (or reals), denote by $(2\cdot D)$ the distribution sampled by taking two independent samples $x$ and $y$ from $D$ and outputting $(x+y)$. Denote by $D|_{\range{s}}$ the distribution $D$ conditioned on the samples being contained in $\range{s}$. For any $s > 0$, let $U_s$ be the uniform distribution over $\range{s}$. That is, $U_s(x)$ is $1/\left(2\cdot\floor{s/2}+1\right)$ if $x\in \range{s}$, and is $0$ otherwise.

We now look at what happens to the distributions of elements in the lists in a single merge step of the $k$-Tree algorithm. The distribution of the sum of two numbers uniform in some range is captured by the following.

\begin{claim}
  \label{claim:tools-1}
  For any $s > 10$ and any integer $z\in [-s,s]$,
  \begin{align*}
    \prob{x,y\gets U_s}{x+y = z} \in \left(\frac{1}{s} - \frac{\abs{z}}{s^2}\right) \pm \frac{5}{s^2}
  \end{align*}
\end{claim}

\begin{proofof}{\cref{claim:tools-1}}
  This follows from the following calculation:
  \begin{align}
    \label{eq:tools-1}
    \prob{x,y\gets U_s}{x+y = z} &= \sum_{x = -\floor{s/2}}^{\floor{s/2}} U_s(x) \cdot U_s(z-x) \nonumber\\
                                 &= \sum_{x = -\floor{s/2}+\max(0,z)}^{\floor{s/2}+\min(0,z)} \frac{1}{(2\floor{s/2}+1)^2}\nonumber\\
                                 &= (2\floor{s/2}+1 - \abs{z}) \cdot \frac{1}{(2\floor{s/2}+1)^2}\nonumber\\
                                 &= \frac{1}{(2\floor{s/2}+1)} - \frac{\abs{z}}{(2\floor{s/2}+1)^2}
  \end{align}
  Using the fact that $(2\floor{s/2}+1) \in s \pm 1$ and $\abs{z} \leq s$, we can compute the ``errors'' in each term as follows:
  \begin{align*}
    \abs{\frac{1}{s} - \frac{1}{2\floor{s/2}+1}} &\leq \frac{1}{s(s-1)}\\
    \abs{\frac{\abs{z}}{s^2} - \frac{\abs{z}}{(2\floor{s/2}+1)^2}} &\leq \frac{\abs{z} (2s+1)}{s^2(s-1)^2} \leq \frac{3}{(s-1)^2}\\
  \end{align*}
  Putting these together with (\ref{eq:tools-1}) and some simple approximations using the fact that $s>10$ gives the claim.
\end{proofof}

\begin{proposition}
  \label{prop:tools-2}
  For any $s > 10$ and $p \in [0,1]$,
  \begin{align*}
    \prob{x,y\gets U_s}{x+y \in \range{sp}} \in \left(p - \frac{p^2}{4}\right) \pm \frac{7}{s}
  \end{align*}
\end{proposition}

\begin{proofof}{\cref{prop:tools-2}}
  Using \cref{claim:tools-1}, we can calculate the relevant probability as follows:
  \begin{align*}
    \prob{x,y\gets U_s}{x+y \in \range{sp}} &= \sum_{z = -\floor{sp/2}}^{\floor{sp/2}} \prob{x,y\gets U_s}{x+y=z}\\
                                            &\in \sum_{z = -\floor{sp/2}}^{\floor{sp/2}} \frac{1}{s} - \frac{\abs{z}}{s^2} \pm \frac{5}{s^2} \\
                                            &= \frac{2\floor{sp/2}+1}{s} - \frac{1}{s^2} \cdot \floor{\frac{sp}{2}} \left(\floor{\frac{sp}{2}} + 1 \right) \pm \frac{5 \cdot (2\floor{sp/2}+1)}{s^2} \\
                                            &\subseteq \left(p \pm \frac{1}{s}\right) - \left(\frac{p^2}{4} \pm \frac{p}{2s}\right) \pm \left(\frac{5p}{s} + \frac{5}{s^2} \right) \\
    &\subseteq \left(p - \frac{p^2}{4}\right) \pm \frac{7}{s}
  \end{align*}
  which proves the proposition.
\end{proofof}

\begin{proposition}
  \label{prop:tools-3}
  For any $s > 20$ and $p\in [0,1]$,
  \begin{align*}
    \mr\left(U_{sp}, 2\cdot U_s|_{\range{sp}} \right) \leq \left(1-\frac{p}{2}\right)^{-1} \left(1+\frac{30}{s} \right)
  \end{align*}
\end{proposition}

\begin{proofof}{\cref{prop:tools-3}}
  Following \cref{prop:unif-dist}, as $U_{sp}$ and $(2\cdot U_s|_{\range{sp}})$ have the same support (that is, $\range{sp}$), it is sufficient to bound the ratio between the maximum and minimum probability masses of the latter to bound its distance from uniform. By Bayes's theorem, the probability mass placed on any $z\in\range{sp}$ by this distribution is as follows:
  \begin{align*}
    (2\cdot U_s|_{\range{sp}})(z) = \prob{x,y\gets U_s}{x+y=z\ |\ x+y\in\range{sp}} = \frac{\prob{x,y\gets U_s}{x+y=z}}{\prob{x,y\gets U_s}{x+y\in\range{sp}}}
  \end{align*}
  As the normalizing factor in the denominator appears in all probability mass values, it is sufficient to bound the ratio between the maximum and minimum values of the numerator above. That is, \cref{prop:unif-dist} implies the following:
  \begin{align*}
    \mr\left(U_{sp}, 2\cdot U_s|_{\range{sp}} \right) \leq \frac{\max_{z\in\range{sp}} \prob{x,y\gets U_s}{x+y=z}}{\min_{z\in\range{sp}} \prob{x,y\gets U_s}{x+y=z}}
  \end{align*}
  We can bound this ratio using \cref{claim:tools-1} as follows:
  \begin{align*}
    \frac{\max_{z\in\range{sp}} \prob{x,y\gets U_s}{x+y=z}}{\min_{z\in\range{sp}} \prob{x,y\gets U_s}{x+y=z}} &\leq \frac{1/s + 5/s^2}{1/s-p/2s-5/s^2}\\
                                                                                                              &\leq \left(1-\frac{p}{2}\right)^{-1} \left(1+\frac{5}{s} \right) \left( 1 - \frac{10}{s} \right)^{-1}\\
    &\leq \left(1-\frac{p}{2}\right)^{-1} \left(1+\frac{30}{s} \right) 
  \end{align*}
  where the second inequality follows from the observation that $(1-p/2) \geq 1/2$, and the third uses the hypothesis that $s>20$. This proves the proposition.
\end{proofof}

Next we show some bounds involving pairs of sums of dependent random variables that come in useful when computing second moments later. Just upper bounds turn out to be sufficient here since we are only interested in upper-bounding these second moments.

\begin{proposition}
  \label{prop:tools-4}
  For any $s > 10$ and $p \in [0,1]$,
  \begin{align*}
    \prob{w,x,y\gets U_s}{w+x \in \range{sp} \wedge w+y \in \range{sp}} \leq p^2 \cdot \left(1 + \frac{3}{sp} \right)^2
  \end{align*}
\end{proposition}

\begin{proofof}{\cref{prop:tools-4}}
  We can write this probability as follows:
  \begin{align}
    \prob{w,x,y\gets U_s}{w+x \in \range{sp} \wedge w+y \in \range{sp}} \nonumber &= \sum_{w = -\floor{s/2}}^{\floor{s/2}} U_s(w) \cdot \prob{x,y\gets U_s}{w+x\in\range{sp} \wedge w+y\in\range{sp}}\\ \nonumber
                                                                        &\leq \max_{w} \prob{x,y\gets U_s}{w+x\in\range{sp} \wedge w+y\in\range{sp}}\\ \nonumber
                                                                        &= \max_{w} \prob{x\gets U_s}{w+x\in\range{sp}}\cdot \prob{y\gets U_s}{w+y\in\range{sp}}\\ \nonumber
                                                                        &= \max_{w} \prob{x\gets U_s}{x\in -w+\range{sp}}\cdot \prob{y\gets U_s}{y\in -w+\range{sp}}\\ \label{eq:prop-tools-4-ineq-2}
                                                                    &\leq \left( \frac{\size{\range{sp}}}{\size{\range{s}}} \right)^2\\ \nonumber
                                                                        &\leq \left( \frac{sp+1}{s-1} \right)^2\\
                                                                        &\leq p^2 \cdot \left( 1 + \frac{3}{sp} \right)^2 \nonumber
  \end{align}
  where the second equality follows from the independence of $x$ and $y$, and the last inequality uses the hypothesis that $s > 10$ and the fact that $p \leq 1$.
\end{proofof}

\begin{proposition}
  \label{prop:tools-5}
  Consider any $s > 20$ and $p\in [0,1/2]$, and let $D$ be the distribution over $\Int\times\Int$ sampled as follows: \textnormal{(Sample $w,x,y\gets U_s$ conditioned on $((w+x),(w+y)\in\range{sp})$, and output $(w+x,w+y)$)}. Then,
  \begin{align*}
    \mr\left(U_{sp}\otimes U_{sp}, D \right) \leq \left(1-p\right)^{-1} \left(1+\frac{4}{s} \right)
  \end{align*}
\end{proposition}

\begin{proofof}{\cref{prop:tools-5}}
  Following \cref{prop:unif-dist}, as $U_{sp}^{\otimes 2}$ (denoting $U_{sp}\otimes U_{sp}$) and $D$ have the same support (that is, $\range{sp}\times\range{sp}$), and $U_{sp}^{\otimes 2}$ is uniform over this support, it is sufficient to bound the ratio between the maximum and minimum probability masses of $D$ to bound this distance. We can bound this ratio as follows:
  \begin{align*}
    \mr\left(U_{sp}^{\otimes 2}, D \right) \leq \frac{\max_{z_1,z_2\in\range{sp}} \prob{w,x,y\gets U_s}{w+x=z_1\wedge w+y=z_2}}{\min_{z_1,z_2\in\range{sp}} \prob{w,x,y\gets U_s}{w+x=z_1\wedge w+y=z_2}}
  \end{align*}
  where again we have ignored the Bayes normalizing factor as it appears in both the numerator and denominator. We can then express this as follows:
  \begin{align*}
    \frac{\max_{z_1,z_2\in\range{sp}} \prob{w,x,y\gets U_s}{w+x=z_1\wedge w+y=z_2}}{\min_{z_1,z_2\in\range{sp}} \prob{w,x,y\gets U_s}{w+x=z_1\wedge w+y=z_2}} &= \frac{\max_{z_1,z_2\in\range{sp}} \sum_{w\in\range{s}} U_s(w) U_s(z_1-w) U_s(z_2-w)}{\min_{z_1,z_2\in\range{sp}} \sum_{w\in\range{s}} U_s(w) U_s(z_1-w) U_s(z_2-w)} \\
     &= \frac{\max_{z_1,z_2\in\range{sp}} \size{\range{s} \cap (z_1-\range{s}) \cap (z_2-\range{s})}}{\min_{z_1,z_2\in\range{sp}} \size{\range{s} \cap (z_1-\range{s}) \cap (z_2-\range{s})}}\\
    &\leq \frac{\size{\range{s}}}{\size{\range{s} \cap (\floor{sp/2}-\range{s}) \cap (-\floor{sp/2}-\range{s})}}\\
    &\leq \frac{\size{\range{s}}}{\size{\range{s}}-2\floor{sp/2}}\\
    &\leq \frac{(s-1)}{(s-1)-sp}\\
    &= (1-p)^{-1} \left(1 - \frac{p}{(1-p)(s-1)} \right)^{-1}\\
    &\leq (1-p)^{-1} \left(1 + \frac{4}{s} \right)
  \end{align*}
  where the second equality follows from the fact that the summand corresponding to each $w$ is a fixed value when all of $w$, $(z_1-w)$, and $(z_2-w)$ are contained in $\range{s}$, and is $0$ otherwise. The first inequality follows by noting that the denominator is the size of the intersection of three intervals of integers, and is minimized when two of them are translated as far as possible in either direction. The last inequality uses the fact that $p \leq 1/2$ and $s > 20$.
\end{proofof}

\subsection{Proof of \cref{prop:c-exp}}
\label{sec:c-exp}

Fix any set of values for $k$, $n$, $m$, and $p$ satisfying the hypothesis of \cref{prop:c-exp}. For each $d\in [0,\log{k}]$ and $i\in [k/2^d]$, denote by $L^d_i$ the corresponding list computed by the $k$-Tree algorithm as described in \cref{fig:ktree}, with $L^0_1, \dots, L^0_k$ being the input lists, each of which contains $n$ uniformly distributed integers from $\range{m}$. Recall that we defined the variable $C_{k,n,m,p}$ to be the number of occurrences of $0$ in $L^{\log{k}}_1$; denote this variable by $C$ for brevity. We will write this $C$ as a sum of indicator variables that indicate whether each tuple of elements in the $k$ input lists pass all the filters of the algorithm.

Each element of $L^{\log{k}}_1$ corresponds to the sum $(L^0_1[\ell_1]+\cdots + L^0_k[\ell_k])$ for some $\ell_1,\dots,\ell_k\in[n]$. For each $\ell_1,\dots,\ell_k \in [n]$, denote the tuple $(\ell_1,\dots,\ell_k)$ by $\bar{\ell}$, and define the following variable:
\begin{align*}
  C_{\bar{\ell}} = \begin{cases}
    1 & \text{if $(L^0_1[\ell_1] + \cdots + L^0_k[\ell_k])$ appears in $L^{\log{k}}_1$, \emph{and} this sum is $0$}\\
    0 & \text{otherwise}
  \end{cases}
\end{align*}
Above, the phrase ``$(L^0_1[\ell_1] + \cdots + L^0_k[\ell_k])$ appears in $L^{\log{k}}_1$'' is to be taken symbolically. That is, it means the following:
\begin{align*}
  \forall d\in[\log{k}]\ \forall i\in \left[ \frac{k}{2^d} \right]: \left( \sum_{j= (i-1)\cdot 2^d+1}^{i\cdot 2^d} L^0_j[\ell_j] \right) \in L^d_i
\end{align*}
We can now write $C$ as:
\begin{align}
  \label{eq:c-exp-0}
  C = \sum_{\bar{\ell}\in[n]^k} C_{\bar{\ell}}
\end{align}
Computing the expectation of each $C_{\bar{\ell}}$ then gives us the expectation of $C$. We do this as follows.

\medskip
Fix any value of $\bar{\ell} = (\ell_1,\dots,\ell_k)$. For each $i\in[k]$ denote by $x^0_i$ the value of $L^0_i[\ell_i]$. When the input lists are uniformly random, each $x^0_i$ is also uniformly random over $\range{m}$. For each $d\in [\log{k}]$ and $i\in[k/2^d]$, we also set up the following notation for the partial sums being considered at each step in the $k$-Tree algorithm:
\begin{align*}
  x^d_i = x^{d-1}_{2i-1} + x^{d-1}_{2i}
\end{align*}
We will later override some of these $x^d_i$'s by sampling them afresh rather than computing them as above, but for any $x^d_i$ for which we do not explicitly say otherwise, the above is to be taken to be its definition. For each $d\in[\log{k}]$, define the following event that captures the set of checks made by the calls to the Merge function in the $d^{th}$ iteration of the algorithm:
\begin{align*}
  E_d(x_1, \dots, x_{k/2^{d-1}}) \equiv \left( \forall i\in \left[ \frac{k}{2^d} \right]: (x_{2i-1} + x_{2i}) \in \range{mp^d} \right) 
\end{align*}
Finally, define the following event that captures the property we are interested in elements of the final list:
\begin{align*}
  E_{\log{k}+1}(x) \equiv \left( x = 0 \right)
\end{align*}
We can now write the expectation of $C_{\bar{\ell}}$ as follows:
\begin{align*}
  \exp{C_{\bar{\ell}}} = \prob{x^0_1,\dots,x^0_k\gets U_m}{ E_1\left(x^0_1,\dots,x^0_k\right) \wedge E_2\left(x^1_1,\dots,x^1_{\frac{k}{2}}\right) \wedge \cdots \wedge E_{\log{k}}\left(x^{\log{k}-1}_1,x^{\log{k}-1}_2\right) \wedge E_{\log{k}+1}\left(x^{\log{k}}_1\right) }
\end{align*}
We define the following sequence of probabilites of subsets of the events in the expression above that we then bound inductively to eventually arrive at a bound for the above expectation. For $d\in[0,\log{k}]$, define the following:
\begin{align*}
  \zeta_d = \prob{x^d_1,\dots,x^d_{k/2^d}\gets U_{mp^d}}{ E_{d+1}\left(x^d_1,\dots,x^d_{k/2^d}\right) \wedge E_{d+2}\left(x^{d+1}_1,\dots,x^{d+1}_{k/2^{d+1}}\right) \wedge \cdots \wedge E_{\log{k}+1}\left(x^{\log{k}}_1\right) }
\end{align*}
From the above two expressions, we have:
\begin{align}
  \label{eq:c-exp-2}
  \exp{C_{\bar{\ell}}} = \zeta_0
\end{align}
The base of our induction is the following observation that follows from noting that for any $z$, $\floor{z} \in (z \pm 1)$:
\begin{align}
  \label{eq:c-exp-3}
  \zeta_{\log{k}} = \prob{x^{\log{k}}_1 \gets U_{mp^{\log{k}}}}{E_{\log{k}+1}(x^{\log{k}}_1)} = \prob{x \gets U_{mp^{\log{k}}}}{x=0} \in \frac{1}{mp^{\log{k}}} \left( 1 \pm \frac{1}{mp^{\log{k}}-1} \right)
\end{align}
The inductive step uses the following claim.
\begin{claim}
  \label{claim:c-exp-1}
  For $d\in [0,\log{k}-1]$, we have: $\zeta_d \in \zeta_{d+1} \cdot \left[p \cdot (1\pm p) \left(1 \pm \frac{30}{mp^{\log{k}}} \right) \right]^{k/2^{d+1}}$
\end{claim}

\noindent We now complete the proof of \cref{prop:c-exp} and then prove \cref{claim:c-exp-1} below.

\begin{proofof}{\cref{prop:c-exp}}
  Inductively applying \cref{claim:c-exp-1} with $d$ going from $(\log{k}-1)$ to $0$, we get the following:
\begin{align}
  \label{eq:c-exp-8}
  \zeta_0 \in \zeta_{\log{k}} \cdot \left[p \cdot (1\pm p) \left(1 \pm \frac{30}{mp^{\log{k}}} \right) \right]^{\sum_{d=0}^{\log{k}-1} k/2^{d+1}} = \zeta_{\log{k}} \cdot \left[p \cdot (1\pm p) \left(1 \pm \frac{30}{mp^{\log{k}}} \right) \right]^{k-1}
\end{align}
Putting together (\ref{eq:c-exp-2}), (\ref{eq:c-exp-3}) and (\ref{eq:c-exp-8}), we get:
\begin{align}
  \label{eq:c-exp-9}
  \exp{C_{\bar{\ell}}} = \zeta_0 &\in \frac{p^k}{mp^{\log{k}+1}} \cdot (1\pm p)^{k-1} \left(1 \pm \frac{30}{mp^{\log{k}}} \right)^{k-1}\left( 1 \pm \frac{1}{mp^{\log{k}}-1} \right)\nonumber\\
  &\subseteq \frac{p^k}{mp^{\log{k}+1}} \cdot (1\pm p)^{k-1} \left(1 \pm \frac{35}{mp^{\log{k}}} \right)^{k-1}
\end{align}

\medskip
\noindent Putting together (\ref{eq:c-exp-0}) and (\ref{eq:c-exp-9}), we get:
\begin{align}\label{eq:c-exp-sum-over-all-k-tuples}
  \exp{C} = \sum_{\bar{\ell}\in[n]^k} \exp{C_{\bar{\ell}}} \in \frac{n^kp^k}{mp^{\log{k}+1}} \cdot (1\pm p)^{k-1} \left(1 \pm \frac{35}{mp^{\log{k}}} \right)^{k-1}
\end{align}
as required in the statement of the proposition. Finally, we prove \cref{claim:c-exp-1}, completing our proof of \cref{prop:c-exp}.
\end{proofof}

\begin{proofof}{\cref{claim:c-exp-1}}
Expanding the expression that defines $\zeta_d$ using Bayes's theorem, we get: 
\begin{align}
  \label{eq:c-exp-4}
  \zeta_d &= \prob{x^d_1,\dots,x^d_{k/2^d}\gets U_{mp^d}}{ E_{d+1}\left(x^d_1,\dots,x^d_{k/2^d}\right) \wedge E_{d+2}\left(x^{d+1}_1,\dots,x^{d+1}_{k/2^{d+1}}\right) \wedge \cdots \wedge E_{\log{k}+1}\left(x^{\log{k}}_1\right) }\nonumber\\
  &= \prob{x^d_1,\dots,x^d_{k/2^d}\gets U_{mp^d}}{ E_{d+1}\left(x^d_1,\dots,x^d_{k/2^d}\right)} \cdot \nonumber\\
  &\qquad \prob{x^d_1,\dots,x^d_{k/2^d}\gets U_{mp^d}}{E_{d+2}\left(x^{d+1}_1,\dots,x^{d+1}_{k/2^{d+1}}\right) \wedge \cdots \ \Big|\ E_{d+1}\left(x^d_1,\dots,x^d_{k/2^d} \right) } \nonumber\\
  &= \prob{x^d_1,\dots,x^d_{k/2^d}\gets U_{mp^d}}{ E_{d+1}\left(x^0_1,\dots,x^0_{k/2^d}\right)} \cdot \nonumber\\
  &\qquad \prob{x^d_1,\dots,x^d_{k/2^d}\gets U_{mp^d}}{E_{d+2}\left(x^{d+1}_1,\dots,x^{d+1}_{k/2^{d+1}}\right) \wedge \cdots \ \Big|\ \forall i\in \left[\frac{k}{2^{d+1}}\right]: (x^d_{2i-1}+x^d_{2i}) \in \range{mp^{d+1}} } \nonumber\\
  &= \prob{x^d_1,\dots,x^d_{k/2^d}\gets U_{mp^d}}{ E_{d+1}\left(x^0_1,\dots,x^0_{k/2^d}\right)} \cdot \nonumber\\
  &\qquad \prob{x^{d+1}_1,\dots,x^{d+1}_{k/2^{d+1}}\gets 2\cdot U_{mp^d}}{E_{d+2}\left(x^{d+1}_1,\dots,x^{d+1}_{k/2^{d+1}}\right) \wedge \cdots \ \Big|\ \forall i\in \left[\frac{k}{2^{d+1}}\right]: x^{d+1}_{i} \in \range{mp^{d+1}} }
\end{align}
where the third equality follows from the definition of $E_{d+1}$, and the fourth from the definition of $x^{d+1}_i$ as being $(x^d_{2i-1}+x^d_{2i})$. We can bound the first term above directly using \cref{prop:tools-2} as follows, noting that for the various values of $i$, the distributions of $(x^d_{2i-1}+x^d_{2i})$ are independent:
\begin{align}
  \label{eq:c-exp-5}
  \prob{x^d_1,\dots,x^d_{k/2^d}\gets U_{mp^d}}{ E_{d+1}\left(x^0_1,\dots,x^0_k\right)} &= \prob{x^d_1,\dots,x^d_{k/2^d}\gets U_{mp^d}}{ \forall i\in \left[\frac{k}{2^{d+1}}\right]: (x^d_{2i-1}+x^d_{2i}) \in \range{mp^{d+1}} } \nonumber\\
  &\in \left( \left(p - \frac{p^2}{4}\right) \pm \frac{7}{mp^d} \right)^{k/2^{d+1}} \nonumber\\
  &\subseteq \left(p - \frac{p^2}{4}\right)^{k/2^{d+1}} \left( 1 \pm \frac{14}{mp^{d+1}} \right)^{k/2^{d+1}}
\end{align}
where the first containment follows from \cref{prop:tools-2} (setting $s$ there to be $mp^d$), and the second from the fact that $(p-p^2/4) > p/2$.

The second term in (\ref{eq:c-exp-4}) can be upper-bounded using \cref{fact:mr-1,fact:mr-2} and \cref{prop:tools-3} (and the hypothesis that $mp^d \geq mp^{\log{k}} > 30$) as follows:
\begin{align}
  \label{eq:c-exp-6}
  &\prob{x^{d+1}_1,\dots,x^{d+1}_{k/2^{d+1}}\gets 2\cdot U_{mp^d}}{E_{d+2}\left(x^{d+1}_1,\dots,x^{d+1}_{k/2^{d+1}}\right) \wedge \cdots \ \Big|\ \forall i\in \left[\frac{k}{2^{d+1}}\right]: x^{d+1}_{i} \in \range{mp^{d+1}} }  \nonumber\\
  &\quad = \prob{x^{d+1}_1,\dots,x^{d+1}_{k/2^{d+1}}\gets 2\cdot U_{mp^d}|_{\range{mp^{d+1}}}}{E_{d+2}\left(x^{d+1}_1,\dots,x^{d+1}_{k/2^{d+1}}\right) \wedge \cdots}  \nonumber\\
  &\quad \leq \mr\left( U_{mp^{d+1}}, 2\cdot U_{mp^d}|_{\range{mp^{d+1}}} \right)^{k/2^{d+1}} \cdot \prob{x^{d+1}_1,\dots,x^{d+1}_{k/2^{d+1}}\gets U_{mp^{d+1}}}{E_{d+2}\left(x^{d+1}_1,\dots,x^{d+1}_{k/2^{d+1}}\right) \wedge \cdots } \nonumber\\
  &\quad = \mr\left( U_{mp^{d+1}}, 2\cdot U_{mp^d}|_{\range{mp^{d+1}}} \right)^{k/2^{d+1}} \cdot \zeta_{d+1}  \nonumber\\
  &\quad \leq \left[ \left(1 - \frac{p}{2}\right)^{-1} \left(1+\frac{30}{mp^d}\right) \right]^{k/2^{d+1}} \cdot \zeta_{d+1}  
\end{align}
We can similarly lower-bound it by:
\begin{align}
  \label{eq:c-exp-7}
  \mr\left( U_{mp^{d+1}}, 2\cdot U_{mp^d}|_{\range{mp^{d+1}}} \right)^{-k/2^{d+1}} \cdot \zeta_{d+1} \geq \left[ \left(1 - \frac{p}{2}\right) \left(1+\frac{30}{mp^d}\right)^{-1} \right]^{k/2^{d+1}} \cdot \zeta_{d+1}  
\end{align}

\medskip
\noindent Putting together (\ref{eq:c-exp-4}), (\ref{eq:c-exp-5}), and (\ref{eq:c-exp-6}), we get the following upper bound:
\begin{align*}
  \zeta_d &\leq \zeta_{d+1} \cdot \left[ \left(p - \frac{p^2}{4}\right) \left(1-\frac{p}{2}\right)^{-1}\left( 1 + \frac{14}{mp^{d+1}} \right) \left(1+\frac{30}{mp^d}\right) \right]^{k/2^{d+1}}\\
  &\leq \zeta_{d+1} \cdot \left[ p (1+p) \left(1 + \frac{29}{mp^{d+1}} \right) \right]^{k/2^{d+1}}
\end{align*}
where the second inequality uses the fact that $(1-p/4)(1-p/2)^{-1} < (1+p)$, and the hypotheses that $mp^d \geq mp^{\log{k}} > 30$ and $p < 1/30$. Similarly, putting together (\ref{eq:c-exp-4}), (\ref{eq:c-exp-5}), and (\ref{eq:c-exp-7}), we get:
\begin{align*}
  \zeta_d &\geq \zeta_{d+1} \cdot \left[ \left(p - \frac{p^2}{4}\right) \left(1-\frac{p}{2}\right)\left( 1 - \frac{14}{mp^{d+1}} \right) \left(1+\frac{30}{mp^d}\right)^{-1} \right]^{k/2^{d+1}}\\
  &\geq \zeta_{d+1} \cdot \left[ p (1-p) \left(1 - \frac{15}{mp^{d+1}} \right) \right]^{k/2^{d+1}}
\end{align*}
where we use the fact that $(1-p/4)(1-p/2) > (1-p)$, and the hypothesis that $p < 1/30$. The above two expressions, along with the fact that $d$ is at most $(\log{k}-1)$, now prove the claim.
\end{proofof}


\subsection{Proof of \cref{prop:c-var}}
\label{sec:c-var}

Here, we prove \cref{prop:c-var}, which bounds the second moment of the variable $C_{k,n,m,p}$ defined earlier. Fix any set of values for $k$, $n$, $m$, and $p$ satisfying the hypothesis of \cref{prop:c-var}. We will continue to use the notation set up in \cref{sec:c-exp}, including denoting $C_{k,n,m,p}$ by $C$. The second moment of $C$ can be written as:
\begin{align*}
  \exp{C^2} = \exp{\left(\sum_{\bar{\ell}\in[n]^k} C_{\bar{\ell}}\right)^2} = \sum_{\bar{\ell},\bar{\ell}'\in[n]^k}{ \exp{C_{\bar{\ell}} \cdot C_{\bar{\ell}'}}}
\end{align*}
For any tuples $\bar{\ell},\bar{\ell}' \in [n]^k$, denote by $\diff(\bar{\ell},\bar{\ell}')$ the \emph{Hamming difference} between them -- that is, $\delta:[n]^k\times [n]^k \ra \bset^k$ outputs a string of length $k$ whose $i^{\text{th}}$ bit is defined as follows:
\begin{align*}
  (\delta(\bar{\ell},\bar{\ell}'))_i = \begin{cases}
    0 &\text{if } \ell_i = \ell_i'\\
    1 &\text{if } \ell_i \neq \ell_i'
  \end{cases}
\end{align*}
Then, we can write the above expectation as:
\begin{align}
  \label{eq:c-var-2}
  \exp{C^2} = \sum_{\bar{\ell},\bar{\ell}'\in[n]^k}{ \exp{C_{\bar{\ell}} \cdot C_{\bar{\ell}'}}} = \sum_{s\in \bset^k} \sum_{\bar{\ell},\bar{\ell}': \diff(\bar{\ell},\bar{\ell}') = s}{\exp{C_{\bar{\ell}} \cdot C_{\bar{\ell}'}}}
\end{align}
Our approach now is to first obtain symbolic bounds on the expectation in the sum for each value of $\diff(\bell,\bell')$, and then bound the sum itself inductively using some convenient structure that it possesses. The extremes have slightly different bounds from the general case, so we first separate them as follows.

\begin{claim}
  \label{claim:c-var-1}
  For any $\bell, \bell' \in [n]^k$ such that $\diff(\bell,\bell') = 0^k$,
  \begin{align*}
    \exp{C_\bell \cdot C_\bell'} = \exp{C_\bell} \in \frac{p^k}{mp^{\log{k}+1}} \cdot (1\pm p)^{k-1} \left(1 \pm \frac{35}{mp^{\log{k}}} \right)^{k-1}
  \end{align*}
\end{claim}
\begin{proofof}{\cref{claim:c-var-1}}
  Because $\bell = \bell'$ and $C_\bell$ and $C_{\bell'}$ are Boolean variables, $C_\bell\cdot C_{\bell'} = C_\bell$, the expectation of which is bounded in (\ref{eq:c-exp-9}) from \cref{sec:c-exp}.
\end{proofof}

The bounds on the expectations in the general case are most conveniently expressed in terms of quantities that are defined on top of the difference $\delta(\bell,\bell')$ in a tree structure that reflects the processing of the list elements by the algorithm. To capture these, we define a ``tree extrapolation'' function $tex:\bset^k\ra\bset^{k/2}\times\bset^{k/4}\times \cdots \times \bset$ that on input an $s^0\in\bset^k$, outputs a tuple of strings $(s^1,\dots, s^{\log{k}})$, where for $d\in[\log{k}]$, the string $s^d$ is contained in $\bset^{k/2^d}$, and its $i^{\text{th}}$ bit is defined as follows for $i\in[k/2^d]$:
\begin{align*}
  s^d_i = \max(s^{d-1}_{2i-1},s^{d-1}_{2i})
\end{align*}
Further, we define the function $stex:\bset^k\ra\Nat$ that on input $s\in\bset^k$, computes the tree extrapolation $(s^1, \dots, s^{\log{k}}) \gets tex(s)$, and then outputs the sum $\sum_{d\in[\log{k}]} \sum_{i\in[k/2^d]} s^d_i$.

\begin{claim}
  \label{claim:c-var-2}
  Consider any $s\in\bset^k\setminus\set{0^k}$, and let $\sigma\gets stex(s)$. For any $\bell, \bell' \in [n]^k$ such that $\diff(\bell,\bell') = s$, we have:
  \begin{align*}
    \exp{C_\bell \cdot C_{\bell'}} \leq \frac{p^{k+\sigma+1}}{(mp^{\log{k}+1})^2} \cdot \left[ (1-p)^{-1} \left(1 + \frac{30}{mp^{\log{k}}}\right) \right]^{(k+\sigma+1)}
  \end{align*}
\end{claim}

We use these claims to prove \cref{prop:c-var}, and then prove \cref{claim:c-var-2}.

\begin{proofof}{\cref{prop:c-var}} 
  Fix any valid set of parameters $k$, $n$, $m$, and $p$ that satisfy the hypothesis of the proposition. Define the following symbols for convenience:
  \begin{align*}
    \mu &= p \cdot \left[ (1-p)^{-1} \left(1 + \frac{35}{mp^{\log{k}}}\right) \right]\\
    \nu &= \frac{1}{mp^{\log{k}+1}}
  \end{align*}
  For $s\in\bset^k$, denote by $(s^1,\dots,s^{\log{k}})$ the outputs of the tree extension function $tex(s)$. Note that the bit $s^{\log{k}}_1$ is $0$ if and only if $s = 0^k$. For any $s$ and $\bell,\bell'$ such that $\diff(\bell,\bell') = s$, we have the following by \cref{claim:c-var-1,claim:c-var-2} (and noting that $(1+p) \leq (1-p)^{-1}$):
  \begin{align}
    \label{eq:prop-c-var-1}
    \exp{C_\bell \cdot C_{\bell'}} \leq \mu^{k+stex(s)} \cdot \nu \cdot (\mu\nu)^{s^{\log{k}}_1}
  \end{align}
  For $s\in\bset^k$, denote by $wt(s)$ the Hamming weight of $s$ -- that is, the sum $\sum_{i\in[k]}s_i$. For any $s\in\bset^k$, the number of pairs $\bar{\ell},\bar{\ell}'\in[n]^k$ such that $\diff(\bar{\ell},\bar{\ell}') = s$ is given by:
  \begin{align}
    \label{eq:prop-c-var-2}
    n^k \cdot (n-1)^{wt(s)} \leq n^{k+wt(s)}
  \end{align}
  Putting together (\ref{eq:c-var-2}, \ref{eq:prop-c-var-1}, \ref{eq:prop-c-var-2}), we have:
  \begin{align}
    \label{eq:prop-c-var-3}
    \exp{C^2} &= \sum_{s\in \bset^k} \sum_{\bar{\ell},\bar{\ell}': \diff(\bar{\ell},\bar{\ell}') = s}{\exp{C_{\bar{\ell}} \cdot C_{\bar{\ell}'}}} \nonumber\\
              &\leq \sum_{s\in\bset^k} n^{k+wt(s)} \cdot \mu^{k+stex(s)} \cdot \nu \cdot (\mu\nu)^{s^{\log{k}}_1} \nonumber\\
              &= (n^k\mu^k\nu) \cdot \sum_{s\in\bset^k} n^{wt(s)} \cdot \mu^{stex(s)} \cdot (\mu\nu)^{s^{\log{k}}_1}
  \end{align}
  We show that the sum in the above expression can be simplified to a recursively defined function of $k$, and to do so we define the following function for every $k \geq 2$ that is a power of $2$:
  \begin{align*}
    T(k) = \sum_{s\in\bset^k} n^{wt(s)} \cdot \mu^{stex(s)} \cdot (\mu\nu)^{s^{\log{k}}_1} - 1 = (\mu\nu) \cdot \sum_{s\in\bset^k: s\neq0^k} n^{wt(s)} \mu^{stex(s)}
  \end{align*}
  where the second equality follows from the fact that $wt(0^k) = stex(0^k) = (0^k)^{\log{k}}_1 = 0$. We also define $T(1)$ to be $(n\cdot \mu \cdot \nu)$. For any $k\geq 2$, by (\ref{eq:prop-c-var-3}) we have:
  \begin{align}
    \label{eq:prop-c-var-4}
    \exp{C^2} \leq (n^k\mu^k\nu) \cdot (1+T(k)) 
  \end{align}
  It may be verified by computation that:
  \begin{align}
    \label{eq:prop-c-var-5}
    T(2) = 2n\mu^2\nu + n^2\mu^2\nu = T(1) \cdot \left(\frac{T(1)}{\nu}+2\mu\right)
  \end{align}
  We show that a similar relation holds between $T(k)$ and $T(k/2)$ for other values of $k$ as well. Generalize the definitions of $wt$ and $stex$ to input strings of length $k/2$ in the natural manner. We then do this as follows:
  \begin{align}
    \label{eq:prop-c-var-6}
    T(k) &= \sum_{s\in\bset^k} n^{wt(s)} \cdot \mu^{stex(s)} \cdot (\mu\nu)^{s^{\log{k}}_1} - 1 \nonumber\\
         &= \sum_{s\in\bset^k:s\neq 0^k} n^{wt(s)} \cdot \mu^{stex(s)} \cdot (\mu\nu)\nonumber\\
         &= \sum_{q,r\in\bset^{k/2}: (q||r) \neq 0^k} n^{wt(q||r)} \cdot \mu^{stex(q||r)} \cdot (\mu\nu)\nonumber\\
         &= \sum_{q,r\in\bset^{k/2}: (q||r) \neq 0^k} n^{wt(q)+wt(r)} \cdot \mu^{stex(q)+stex(r)+1} \cdot (\mu\nu)\nonumber\\
         &= (\mu^2\nu) \cdot \sum_{q,r\in\bset^{k/2}: (q||r) \neq 0^k} n^{wt(q)+wt(r)} \cdot \mu^{stex(q)+stex(r)}\nonumber\\
         &= (\mu^2\nu) \cdot \left( \sum_{q,r\in\bset^{k/2}} n^{wt(q)+wt(r)} \cdot \mu^{stex(q)+stex(r)} - 1 \right)\nonumber\\
         &= (\mu^2\nu) \cdot \left( \Big(\sum_{q\in\bset^{k/2}} n^{wt(q)} \cdot \mu^{stex(q)}\Big)^2 - 1 \right)\nonumber\\
         &= (\mu^2\nu) \cdot \left( \left(\frac{T(k/2)}{\mu\nu}+1\right)^2 - 1 \right)\nonumber\\
         &= \mu \cdot T(k/2) \cdot \left( \frac{T(k/2)}{\mu\nu} + 2  \right)\nonumber\\
         &= T(k/2) \cdot \left( \frac{T(k/2)}{\nu} + 2 \mu \right)
  \end{align}
  Above, in the second, sixth, and eighth equalities, we use the fact that $wt(0^{k}) = stex(0^k) = 0$. The fourth equality follows from the fact that $wt$ is additive under concatenation of strings, and $stex(q||r)$, because of its tree-based definition, is equal to $stex(q)+stex(r)+1$ if $(q||r)$ is not $0^{k}$. The rest are straightforward algebraic manipulations. The definition of $T(1)$ and (\ref{eq:prop-c-var-5},\ref{eq:prop-c-var-6}) imply that this function $T$ is exactly the function $T_{\mu,\nu,n}$ specified in the statement of the proposition. This observation together with (\ref{eq:prop-c-var-4}) now proves \cref{prop:c-var}.
\end{proofof}

\begin{proofof}{\cref{claim:c-var-2}}
  Fix any $s\in \bset^k\setminus\set{0^k}$ with $\sigma \gets stex(s)$, and tuples $\bell^x, \bell^y \in [n]^k$ such that $\diff(\bell^x,\bell^y) = s$. We set up notation similar to that used in \cref{sec:c-exp}. For each $i\in [k]$, let:
  \begin{align*}
    x^i_0 &= L_i[\ell^x_i]\\
    y^i_0 &= L_i[\ell^y_i]
  \end{align*}
  For each $d\in [\log{k}]$ and $i\in[k/2^d]$, unless overridden, define the following:
  \begin{align*}
    x^d_i &= x^{d-1}_{2i-1} + x^{d-1}_{2i}\\
    y^d_i &= y^{d-1}_{2i-1} + y^{d-1}_{2i}
  \end{align*}
  We sometimes denote by $\bar{x}^d$ (similarly $\bar{y}^d$) the tuple $(x^d_1,\dots,x^d_{k/2^d})$. In our proof, we will rely on the fact that for $i \in [k]$ on which $\bell^x$ and $\bell^y$ differ, the values of $x^0_i = L_i[\ell^x_i]$ and $y^0_i = L_i[\ell^y_i]$ are independent. To enable us to argue fluidly about this, we define the following indicator constants\footnote{These are constants because we have already fixed $\bell^x$ and $\bell^y$.} for $i\in[k]$:
  \begin{align*}
    \ind^0_i = \begin{cases}
      1 &\text{if } \ell^x_i \neq \ell^y_i\\
      0 &\text{if } \ell^x_i = \ell^y_i
    \end{cases}
  \end{align*}
  Similarly, to capture the independence of intermediate values computed in the course of the $k$-Tree algorithm, we extend this to define the following for each $d\in [\log{k}]$ and $i\in[k/2^d]$:
  \begin{align*}
    \ind^d_i &= \max(\ind^{d-1}_{2i-1},\ind^{d-1}_{2i})
  \end{align*}
  That is, $\ind^d_i$ indicates whether the values computed at the $i^{\text{th}}$ ``node'' at depth $d$ of the $k$-Tree corresponding to the two tuples of indices $\bell^x$ and $\bell^y$ (that is, $x^d_i$ and $y^d_i$) are symbolically equal (in which case it is $0$). Note that the $\ind^d_i$'s correspond directly to bits in the ``tree extrapolation'' $tex(s)$ of $s = \delta(\bell^x,\bell^y)$ as defined before the statement of \cref{claim:c-var-2}. In particular, we have the following:
  \begin{align}
    \label{eq:c-var-sigma}
    \sum_{d\in[\log{k}]} \sum_{i\in[k/2^d]} \ind^d_i = stex(s) = \sigma
  \end{align}
  As in \cref{sec:c-exp}, for each $d\in[\log{k}]$, define the following event:
  \begin{align*}
    E_d(x_1, \dots, x_{k/2^{d-1}}) \equiv \left( \forall i\in \left[ \frac{k}{2^d} \right]: (x_{2i-1} + x_{2i}) \in \range{mp^d} \right) 
  \end{align*}
  And the following event that captures the property we are interested in elements of the final list:
  \begin{align*}
    E_{\log{k}+1}(x) \equiv \left( x = 0 \right)
  \end{align*}
  Finally, we define the following events that enforce the conditions of independence captured by the $\ind^d_i$'s. For each $d\in[0,\log{k}]$:
  \begin{align*}
    V_d \equiv \left( \forall i \in \left[ \frac{k}{2^d} \right] \text{ such that } (\ind^d_i = 0): x^d_i = y^d_i \right)
  \end{align*}
  The expectation we care about can now be written as follows:
  \begin{align}
    \label{eq:c-var-claim-1}
    \exp{C_{\bell^x} \cdot C_{\bell^y}} = \prob{\bar{x}^0,\bar{y}^0\gets (U_m)^k}{\left(E_1(\bar{x}^0) \wedge E_1(\bar{y}^0)\right) \wedge  \cdots \wedge \left( E_{\log{k}+1}(\bar{x}^{\log{k}}) \wedge E_{\log{k}+1}(\bar{y}^{\log{k}}) \right) \ \big|\ V_0}
  \end{align}
  As in \cref{sec:c-exp}, we will bound this quantity inductively. For $d\in[0,\log{k}]$, define the following:
  \begin{align*}
    \chi_d = \prob{\bar{x}^d,\bar{y}^d\gets (U_{mp^d})^{k/2^d}}{ \left( E_{d+1}(\bar{x}^d) \wedge E_{d+1}( \bar{y}^d ) \right) \wedge \left( E_{d+2}(\bar{x}^{d+1}) \wedge E_{d+2}( \bar{y}^{d+1} ) \right) \wedge \cdots\ \big|\ V_d}
  \end{align*}
  By definition, we have:
  \begin{align}
    \label{eq:c-var-claim-2}
    \exp{C_{\bell^x} \cdot C_{\bell^y}} = \chi_0
  \end{align}
  As the base for our induction, we will show a bound on $\chi_{\log{k}}$. Note that as $\diff(\bell^x,\bell^y) \neq 0^k$, the value of $\ind^{\log{k}}_1$ is always $1$, and so $V_{\log{k}}$ is vacuously a tautology. So we have:
  \begin{align}
    \label{eq:c-var-claim-3}
    \chi_{\log{k}} &= \prob{{x}^{\log{k}}_1,{y}^{\log{k}}_1\gets U_{mp^{\log{k}}}}{E_{\log{k}+1}({x}^{\log{k}}_1) \wedge E_{\log{k}+1}({y}^{\log{k}}_1)\ \big|\ V_{\log{k}} } \nonumber\\
                   &= \prob{{x}^{\log{k}}_1,{y}^{\log{k}}_1\gets U_{mp^{\log{k}}}}{E_{\log{k}+1}({x}^{\log{k}}_1) \wedge E_{\log{k}+1}({y}^{\log{k}}_1) } \nonumber\\
                   &= \prob{{x}^{\log{k}}_1\gets U_{mp^{\log{k}}}}{E_{\log{k}+1}({x}^{\log{k}}_1)} \cdot \prob{{y}^{\log{k}}_1\gets U_{mp^{\log{k}}}}{E_{\log{k}+1}({y}^{\log{k}}_1) } \nonumber\\
                   & \leq \frac{1}{(mp^{\log{k}}-1)^2} = \frac{1}{(mp^{\log{k}})^2} \left( 1 + \frac{1}{mp^{\log{k}}-1} \right)^2
  \end{align}
  The following claim enables our induction.
  \begin{claim}
    \label{claim:c-var-4}
    For $d\in[0,\log{k}-1]$, we have:
    \begin{align*}
      \chi_d \leq \chi_{d+1} \cdot \left[ p (1-p)^{-1} \left(1 + \frac{30}{mp^{\log{k}}}\right) \right]^{\left(\frac{k}{2^{d+1}} + \sum_{i=1}^{k/2^{d+1}} \ind^{d+1}_i \right)}
    \end{align*}
  \end{claim}
  We continue with the proof of \cref{claim:c-var-2} now, and then prove \cref{claim:c-var-4} later below. Applying \cref{claim:c-var-4} repeatedly, we get the following relationship:
  \begin{align}
    \label{eq:c-var-claim-4}
    \chi_0 &\leq \chi_{\log{k}} \cdot \left[ p (1-p)^{-1} \left(1 + \frac{30}{mp^{\log{k}}}\right) \right]^{\sum_{d=0}^{\log{k}-1}\left(\frac{k}{2^{d+1}} + \sum_{i=1}^{k/2^{d+1}} \ind^{d+1}_i \right)} \nonumber\\
           &= \chi_{\log{k}} \cdot \left[ p (1-p)^{-1} \left(1 + \frac{30}{mp^{\log{k}}}\right) \right]^{\left((k-1) + \sum_{d=0}^{\log{k}-1}\sum_{i=1}^{k/2^{d+1}} \ind^{d+1}_i \right)} \nonumber\\
           &= \chi_{\log{k}} \cdot \left[ p (1-p)^{-1} \left(1 + \frac{30}{mp^{\log{k}}}\right) \right]^{\left((k-1) + \sum_{d=1}^{\log{k}}\sum_{i=1}^{k/2^{d}} \ind^{d}_i \right)} \nonumber\\
           &= \chi_{\log{k}} \cdot \left[ p (1-p)^{-1} \left(1 + \frac{30}{mp^{\log{k}}}\right) \right]^{\left(k-1 + \sigma \right)}
  \end{align}
  where the last equality follows from (\ref{eq:c-var-sigma}). Putting together (\ref{eq:c-var-claim-3}) and (\ref{eq:c-var-claim-4}), we have:
  \begin{align*}
    \chi_0 &\leq \chi_{\log{k}} \cdot \left[ p (1-p)^{-1} \left(1 + \frac{30}{mp^{\log{k}}}\right) \right]^{(k+\sigma-1)}\\
    &\leq \frac{p^{k+\sigma-1}}{(mp^{\log{k}})^2} \cdot \left[ (1-p)^{-1} \left(1 + \frac{30}{mp^{\log{k}}}\right) \right]^{(k+\sigma-1)} \left( 1 + \frac{1}{mp^{\log{k}}-1} \right)^2\\
    &\leq \frac{p^{k+\sigma+1}}{(mp^{\log{k}+1})^2} \cdot \left[ (1-p)^{-1} \left(1 + \frac{30}{mp^{\log{k}}}\right) \right]^{(k+\sigma+1)} 
  \end{align*}
  Together with (\ref{eq:c-var-claim-2}), this proves \cref{claim:c-var-2}.
\end{proofof}

\begin{proofof}{\cref{claim:c-var-4}}
  As in the proof of \cref{claim:c-exp-1}, we can write $\chi_d$ as follows:
  \begin{align}
    \label{eq:c-var-4-1}
    \chi_d &= \prob{\bar{x}^d,\bar{y}^d\gets (U_{mp^d})^{k/2^d}}{ \left( E_{d+1}(\bar{x}^d) \wedge E_{d+1}( \bar{y}^d ) \right) \wedge \left( E_{d+2}(\bar{x}^{d+1}) \wedge E_{d+2}( \bar{y}^{d+1} ) \right) \wedge \cdots\ \big|\ V_d} \nonumber\\
           &= \prob{\bar{x}^d,\bar{y}^d\gets (U_{mp^d})^{k/2^d}}{ E_{d+1}(\bar{x}^d) \wedge E_{d+1}( \bar{y}^d )\ \big|\ V_d } \cdot \nonumber\\
           &\qquad\prob{\bar{x}^d,\bar{y}^d\gets (U_{mp^d})^{k/2^d}}{ \left( E_{d+2}(\bar{x}^{d+1}) \wedge E_{d+2}( \bar{y}^{d+1} ) \right) \wedge \cdots\ \big|\ V_d \wedge \left( E_{d+1}(\bar{x}^d) \wedge E_{d+1}( \bar{y}^d ) \right)}
  \end{align}
  As in the proof of \cref{claim:c-exp-1}, we would like to bound the two terms in the right-hand size of (\ref{eq:c-var-4-1}) separately. In each case, we do this by looking at the various nodes at level $(d+1)$ of the $k$-Tree separately. As the values of $x^{d+1}_i$ and $x^{d+1}_{i'}$ are independent for distinct values of $i$ and $i'$ (and the same with the $y$'s), the first term can be written as the following product:
  \begin{align}
    \label{eq:c-var-4-2}
    &\prob{\bar{x}^d,\bar{y}^d\gets (U_{mp^d})^{k/2^d}}{ E_{d+1}(\bar{x}^d) \wedge E_{d+1}( \bar{y}^d )\ \big|\ V_d } \nonumber\\
    &\quad= \prob{\bar{x}^d,\bar{y}^d\gets (U_{mp^d})^{k/2^d}}{\forall i\in \left[ \frac{k}{2^{d+1}} \right]: (x^d_{2i-1} + x^d_{2i}) \in \range{mp^{d+1}} \wedge (y^d_{2i-1} + y^d_{2i}) \in \range{mp^{d+1}} \ \big|\ V_d}\nonumber\\
    &\quad= \prod_{i=0}^{k/2^{d+1}} \prob{x^d_{2i-1},x^d_{2i},y^d_{2i-1},y^d_{2i}\gets U_{mp^d}}{(x^d_{2i-1} + x^d_{2i}) \in \range{mp^{d+1}} \wedge (y^d_{2i-1} + y^d_{2i}) \in \range{mp^{d+1}} \ \big|\ V_d}
  \end{align}
  We bound the terms in the product in (\ref{eq:c-var-4-2}) separately. For $i\in[k/2^{d+1}]$ such that $\ind^{d+1}_i = 0$, due to the conditioning on $V_d$, the events $(x^d_{2i-1} + x^d_{2i}) \in \range{mp^{d+1}}$ and $(y^d_{2i-1} + y^d_{2i}) \in \range{mp^{d+1}}$ are identical, and the probability them both occuring is bounded as follows:
  \begin{align}
    \label{eq:c-var-4-3}
    &\prob{x^d_{2i-1},x^d_{2i},y^d_{2i-1},y^d_{2i}\gets U_{mp^d}}{(x^d_{2i-1} + x^d_{2i}) \in \range{mp^{d+1}} \wedge (y^d_{2i-1} + y^d_{2i}) \in \range{mp^{d+1}} \ \big|\ V_d}\nonumber\\
    &\quad= \prob{x^d_{2i-1},x^d_{2i}\gets U_{mp^d}}{(x^d_{2i-1} + x^d_{2i}) \in \range{mp^{d+1}}} \nonumber\\
    &\quad\leq \left(p - \frac{p^2}{4} \right) + \frac{7}{mp^d} \leq p\cdot \left( 1 + \frac{7}{mp^{d+1}} \right)
  \end{align}
  where the last inequality follows from \cref{prop:tools-2}. Next, consider the $i$ such that $\ind^{d+1}_i = \ind^{d}_{2i-1} = \ind^{d}_{2i} = 1$. In this case, conditioning on $V_d$ does not affect the independence of the $x$ and $y$ variables, and we can write the above probability as follows:
  \begin{align}
    \label{eq:c-var-4-4}
    &\prob{x^d_{2i-1},x^d_{2i},y^d_{2i-1},y^d_{2i}\gets U_{mp^d}}{(x^d_{2i-1} + x^d_{2i}) \in \range{mp^{d+1}} \wedge (y^d_{2i-1} + y^d_{2i}) \in \range{mp^{d+1}} \ \big|\ V_d}\nonumber\\
    &\quad= \prob{x^d_{2i-1},x^d_{2i}\gets U_{mp^d}}{(x^d_{2i-1} + x^d_{2i}) \in \range{mp^{d+1}}} \cdot \prob{y^d_{2i-1},y^d_{2i}\gets U_{mp^d}}{(y^d_{2i-1} + y^d_{2i}) \in \range{mp^{d+1}}} \nonumber\\
    &\quad\leq p^2 \cdot \left( 1 + \frac{7}{mp^{d+1}} \right)^2
  \end{align}
  where the last inequality is from \cref{prop:tools-2}. Finally, consider the $i$'s for which $\ind^{d+1}_i = 1$ and exactly one of $\ind^d_{2i-1}$ and $\ind^d_{2i}$ is $1$. Suppose, without loss of generality, that $\ind^d_{2i-1} = 1$. Then, conditioning on $V_d$, the variables $x^d_{2i-1}$ and $y^d_{2i-1}$ are equal and uniformly distributed, whereas $x^d_{2i}$ and $y^d_{2i}$ are independent and uniformly distributed. Thus, in this case, we can use \cref{prop:tools-4} to bound the probability as follows:
  \begin{align}
    \label{eq:c-var-4-5}
    &\prob{x^d_{2i-1},x^d_{2i},y^d_{2i-1},y^d_{2i}\gets U_{mp^d}}{(x^d_{2i-1} + x^d_{2i}) \in \range{mp^{d+1}} \wedge (y^d_{2i-1} + y^d_{2i}) \in \range{mp^{d+1}} \ \big|\ V_d}\nonumber\\
    &\quad= \prob{x^d_{2i-1},x^d_{2i},y^d_{2i}\gets U_{mp^d}}{(x^d_{2i-1} + x^d_{2i}) \in \range{mp^{d+1}} \wedge (x^d_{2i-1} + y^d_{2i}) \in \range{mp^{d+1}}}\nonumber\\
    &\quad\leq  p^2 \cdot \left( 1 + \frac{3}{mp^{d+1}} \right)^2 \leq p^2 \cdot \left( 1 + \frac{7}{mp^{d+1}} \right)^2
  \end{align}
  Putting together (\ref{eq:c-var-4-2}-\ref{eq:c-var-4-5}), we get:
  \begin{align}
    \label{eq:c-var-4-6}
    \prob{\bar{x}^d,\bar{y}^d\gets (U_{mp^d})^{k/2^d}}{ E_{d+1}(\bar{x}^d) \wedge E_{d+1}( \bar{y}^d )\ \big|\ V_d } \leq \left[ p \cdot \left(1 + \frac{7}{mp^{d+1}} \right) \right]^{\left(\frac{k}{2^{d+1}} + \sum_{i=1}^{k/2^{d+1}} \ind^{d+1}_i \right)}
  \end{align}

  \medskip
  \noindent Next we bound the second term in the right-hand side of (\ref{eq:c-var-4-1}), which, to recall, is the following:
  \begin{align}
    \label{eq:c-var-4-7}
    \prob{\bar{x}^d,\bar{y}^d\gets (U_{mp^d})^{k/2^d}}{ \left( E_{d+2}(\bar{x}^{d+1}) \wedge E_{d+2}( \bar{y}^{d+1} ) \right) \wedge \cdots\ \big|\ V_d \wedge \left( E_{d+1}(\bar{x}^d) \wedge E_{d+1}( \bar{y}^d ) \right)}
  \end{align}
  As in the proof of \cref{claim:c-exp-1}, we bound this in terms of $\chi_{d+1}$ and the distance between the distributions of $x^{d+1}_i$'s and $y^{d+1}_i$'s as sampled in the above expression and the corresponding uniform distributions. Note that $\chi_{d+1}$ is the following:
  \begin{align}
    \label{eq:c-var-4-8}
    \chi_{d+1} = \prob{\bar{x}^{d+1},\bar{y}^{d+1}\gets (U_{mp^{d+1}})^{k/2^{d+1}}}{ \left( E_{d+2}(\bar{x}^{d+1}) \wedge E_{d+2}( \bar{y}^{d+1} ) \right) \wedge \cdots\ \big|\ V_{d+1}}
  \end{align}
  By \cref{fact:mr-1}, the (multiplicative) difference between (\ref{eq:c-var-4-7}) and $\chi_{d+1}$ is bounded by the distance between the conditioned distributions from which $(\bar{x}^{d+1},\bar{y}^{d+1})$ is sampled in the two expressions. These distributions can be written in terms of the following two distributions defined for each $i\in[k/2^{d+1}]$:
  \begin{align*}
    D^1_i &\equiv \left[\begin{array}{l} \text{Sample $x^d_{2i-1}, x^d_{2i}, y^d_{2i-1}, y^d_{2i} \gets U_{mp^d}$}\\\text{If $\ind^d_{2i-1} = 0$, reassign $y^d_{2i-1} = x^d_{2i-1}$}\\\text{If $\ind^d_{2i} = 0$, reassign $y^d_{2i} = x^d_{2i}$}\\\text{Set $x^{d+1}_i = (x^d_{2i-1}+x^d_{2i})$ and $y^{d+1}_i = (y^d_{2i-1}+y^d_{2i})$}\\\text{Condition on $x^{d+1}_i, y^{d+1}_i \in \range{mp^{d+1}}$}\\\text{Output $(x^{d+1}_i,y^{d+1}_i)$}\end{array}\right]
  \end{align*}
  \begin{align*}
    D^0_i &\equiv \left[\begin{array}{l} \text{Sample $x^{d+1}_{i}, y^{d+1}_{i} \gets U_{mp^{d+1}}$}\\\text{If $\ind^{d+1}_{i} = 0$, reassign $y^{d+1}_{i} = x^{d+1}_{i}$}\\\text{Output $(x^{d+1}_i,y^{d+1}_i)$}\end{array}\right]
  \end{align*}
  By the definitions of the $E_d$'s and $V_d$'s, it may be seen that the distribution of $(\bar{x}^{d+1},\bar{y}^{d+1})$ in (\ref{eq:c-var-4-7}) is essentially $(D^1_1\otimes \cdots \otimes D^1_{k/2^{d+1}})$, while that in (\ref{eq:c-var-4-8}) is essentially $(D^0_1\otimes \cdots \otimes D^0_{k/2^{d+1}})$. Thus, if we bound the distance between $D^0_i$ and $D^1_i$ for each $i$, the multiplicative difference between these expressions, by \cref{fact:mr-1,fact:mr-2} will be bounded by the product of these distances.

  Consider any $i\in[k/2^{d+1}]$ for which $\ind^{d+1}_i = 0$, which implies that $\ind^d_{2i-1} = \ind^d_{2i} = 0$, and so the samples of both distributions are fully determined by the $x$ variables. Then, by the definitions of the distributions, for such $i$ we have:
  \begin{align}
    \label{eq:c-var-4-9}
    \mr(D^0_i,D^1_i) &= \mr\left(U_{mp^{d+1}}, 2\cdot U_{mp^d}|_{\range{mp^{d+1}}} \right) \nonumber\\
                     &\leq \left(1 - \frac{p}{2}\right)^{-1}\left(1+\frac{30}{mp^d}\right) \nonumber\\
                     &\leq \left(1 -p\right)^{-1}\left(1+\frac{30}{mp^d}\right)
  \end{align}
  where the first inequality follows from \cref{prop:tools-3}. Next, consider an $i$ such that $\ind^{d+1}_i = \ind^d_{2i-1} = \ind^d_{2i} = 1$. In this case, the $x$ variables and the $y$ variables are independent, and each behaves exactly like the $x$ variables in the previous case. So, using \cref{fact:mr-2}, for such $i$ we have:
  \begin{align}
    \label{eq:c-var-4-10}
    \mr(D^0_i,D^1_i) &= \mr\left(U_{mp^{d+1}}^{\otimes 2}, 2\cdot U_{mp^d}|_{\range{mp^{d+1}}}^{\otimes 2} \right) \nonumber\\
                     &\leq \mr\left(U_{mp^{d+1}}, 2\cdot U_{mp^d}|_{\range{mp^{d+1}}}\right)^2 \nonumber\\
                     &\leq \left[ \left(1 -p\right)^{-1}\left(1+\frac{30}{mp^d}\right) \right]^2
  \end{align}
  Finally, consider $i$ such that $\ind^{d+1}_i = 1$ and exactly one of $\ind^d_{2i-1}$ and $\ind^d_{2i}$ is $1$; without loss of generality, suppose $\ind^d_{2i-1} = 0$. Then, in $D^1_i$, $x^d_{2i-1} = y^d_{2i-1}$ is sampled from $U_{mp^d}$, and is added to $x^d_{2i}$ and $y^d_{2i}$, which are sampled independently from $U_{mp^d}$, to get the outputs, whereas $D^0_i$ simply outputs independent samples from $U_{mp^{d+1}}$. We can then use \cref{prop:tools-5} to bound the distance between these as:
  \begin{align}
    \label{eq:c-var-4-11}
    \mr(D^0_i,D^1_i) &\leq (1-p)^{-1} \left( 1 + \frac{4}{mp^d} \right) \leq \left[ (1-p)^{-1} \left(1 + \frac{30}{mp^d} \right) \right]^2
  \end{align}
  Putting together (\ref{eq:c-var-4-8}-\ref{eq:c-var-4-11}) and applying \cref{fact:mr-1,fact:mr-2}, we get:
  \begin{align}
    \label{eq:c-var-4-12}
    &\prob{\bar{x}^d,\bar{y}^d\gets (U_{mp^d})^{k/2^d}}{ \left( E_{d+2}(\bar{x}^{d+1}) \wedge E_{d+2}( \bar{y}^{d+1} ) \right) \wedge \cdots\ \big|\ V_d \wedge \left( E_{d+1}(\bar{x}^d) \wedge E_{d+1}( \bar{y}^d ) \right)} \nonumber\\
    &\qquad\qquad\leq \chi_{d+1} \cdot \left[ (1-p)^{-1} \left(1 + \frac{30}{mp^d} \right) \right]^{\left(\frac{k}{2^{d+1}} + \sum_{i=1}^{k/2^{d+1}} \ind^{d+1}_i \right)}
  \end{align}

  \medskip
  Putting together (\ref{eq:c-var-4-1}), (\ref{eq:c-var-4-6}), and (\ref{eq:c-var-4-12}), we get:
  \begin{align*}
    \chi_d &\leq \chi_{d+1} \cdot \left[ p (1-p)^{-1} \left(1 + \frac{7}{mp^{d+1}}\right) \left(1 + \frac{30}{mp^d} \right) \right]^{\left(\frac{k}{2^{d+1}} + \sum_{i=1}^{k/2^{d+1}} \ind^{d+1}_i \right)}\\
    &\leq \chi_{d+1} \cdot \left[ p (1-p)^{-1} \left(1 + \frac{30}{mp^{\log{k}}}\right) \right]^{\left(\frac{k}{2^{d+1}} + \sum_{i=1}^{k/2^{d+1}} \ind^{d+1}_i \right)}
  \end{align*}
  where the second inequality uses the hypotheses that $mp^d \geq mp^{\log{k}} \geq 30$, and $p < 1/30$. This proves the claim.
\end{proofof}


\subsection{Proof of \cref{prop:l}}
\label{sec:l}

Much of this proof proceeds along the lines of that of \cref{prop:c-exp}, and parts of the proof below are reproduced from there nearly verbatim. Fix any set of values for $k$, $n$, $m$, and $p$ satisfying the hypothesis of \cref{prop:l}. For each $d\in [0,\log{k}]$ and $i\in [k/2^d]$, denote by $L^d_i$ the corresponding list computed by the $k$-Tree algorithm as described in \cref{fig:ktree}, with $L^0_1, \dots, L^0_k$ being the input lists, each of which contains $n$ uniformly distributed integers from $\range{m}$. Recall that we defined the variable $\Lambda_{k,n,m,p}$ to be the total size of all of these lists; denote this variable by $\Lambda$ for brevity. For each $d\in[0,\log{k}]$, denote by $\Lambda_d$ the size of the list $L^d_1$. For any $d$, due to symmetry, the expected size of all lists $L^d_i$ is the same. Thus, we have:
\begin{align}
  \label{eq:l-1}
  \exp{\Lambda} = \exp{\sum_{d\in[0,\log{k}]} \sum_{i\in[k/2^d]} \size{L^d_i}} = \sum_{d\in[0,\log{k}]} \sum_{i\in[k/2^d]} \exp{\size{L^d_i}} =  \sum_{d\in[0,\log{k}]} \frac{k}{2^d} \cdot \exp{\Lambda_d}
\end{align}
We will then bound each $\Lambda_d$ separately and use the above sum to bound $\Lambda$. By design, for $d = 0$, we have:
\begin{align}
  \label{eq:l-2}
  \exp{\Lambda_0} = n
\end{align}
Fix any $d\in[1,\log{k}]$. As in the proof of \cref{prop:c-exp}, we will write $\Lambda_d$ as a sum of variables that indicate whether each tuple of elements in the $2^d$ input lists that contribute to the list $L^d_1$ pass all the filters of the algorithm until that point. Each element of $L^{d}_1$ corresponds to the sum $(L^0_1[\ell_1]+\cdots + L^0_{2^d}[\ell_{2^d}])$ for some $\ell_1,\dots,\ell_{2^d}\in[n]$. For each $\ell_1,\dots,\ell_{2^d} \in [n]$, denote the tuple $(\ell_1,\dots,\ell_{2^d})$ by $\bar{\ell}$, and define the following variable:
\begin{align*}
  \Lambda_{d,\bar{\ell}} = \begin{cases}
    1 & \text{if $ (L^0_1[\ell_1] + \cdots + L^0_{2^d}[\ell_{2^d}]) $ appears in $ L^{d}_1$ }\\
    0 & \text{otherwise}
  \end{cases}
\end{align*}
We can now write $\Lambda_d$ as:
\begin{align}
  \label{eq:l-3}
  \Lambda_d = \sum_{\bar{\ell}\in[n]^{2^d}} \Lambda_{d,\bar{\ell}}
\end{align}
Computing the expectation of each $\Lambda_{d,\bar{\ell}}$ then gives us the expectation of $\Lambda_d$. We do this as follows.

\medskip
Fix any value of $\bar{\ell} = (\ell_1,\dots,\ell_{2^d})$. For each $i\in[2^d]$ denote by $x^0_i$ the value of $L^0_i[\ell_i]$. When the input lists are uniformly random, each $x^0_i$ is also uniformly random over $\range{m}$. For each $t\in [d]$ and $i\in[2^{d-t}]$, we also set up the following notation for the partial sums being considered at each step in the $k$-Tree algorithm:
\begin{align*}
  x^t_i = x^{t-1}_{2i-1} + x^{t-1}_{2i}
\end{align*}
We will later override some of these $x^t_i$'s by sampling them afresh rather than computing them as above, but for any $x^t_i$ for which we do not explicitly say otherwise, the above is to be taken to be its definition. For each $t\in[d]$, define the following event that captures the set of checks made by the calls to the Merge function in the $t^{th}$ iteration of the algorithm:
\begin{align*}
  E_t(x_1, \dots, x_{2^{d-t+1}}) \equiv \left( \forall i\in \left[ 2^{d-t} \right]: (x_{2i-1} + x_{2i}) \in \range{mp^t} \right) 
\end{align*}
We can now write the expectation of $\Lambda_{d,\bar{\ell}}$ as follows:
\begin{align*}
  \exp{\Lambda_{d,\bar{\ell}}} = \prob{x^0_1,\dots,x^0_{2^d}\gets U_m}{ E_1\left(x^0_1,\dots,x^0_{2^d}\right) \wedge E_2\left(x^1_1,\dots,x^1_{2^{d-1}}\right) \wedge \cdots \wedge E_{d}\left(x^{d-1}_1,x^{d-1}_2\right) }
\end{align*}
We define the following sequence of probabilites of subsets of the events in the expression above that we then bound inductively to eventually arrive at a bound for the above expectation. For $t\in[0,d-1]$, define the following:
\begin{align*}
  \xi_t = \prob{x^t_1,\dots,x^t_{2^{d-t}}\gets U_{mp^t}}{ E_{t+1}\left(x^t_1,\dots,x^t_{2^{d-t}}\right) \wedge E_{t+2}\left(x^{t+1}_1,\dots,x^{t+1}_{2^{d-t-1}}\right) \wedge \cdots \wedge E_{d}\left(x^{d-1}_1,x^{d-1}_2\right) }
\end{align*}
From the above two expressions, we have:
\begin{align}
  \label{eq:l-4}
  \exp{\Lambda_{d,\bar{\ell}}} = \xi_0
\end{align}
The base of our induction is the following bound that follows from \cref{prop:tools-2}:
\begin{align}
  \label{eq:l-5}
  \xi_{d-1} &= \prob{x^{d-1}_1, x^{d-1}_2 \gets U_{mp^{d-1}}}{E_{d}(x^{d-1}_1,x^{d-1}_2)}\nonumber\\
            &= \prob{x_1,x_2 \gets U_{mp^{d-1}}}{x_1+x_2\in \range{mp^d}}\nonumber\\
            &\in \left( p - \frac{p^2}{4} \right) \pm \frac{7}{mp^{d-1}}
\end{align}
The inductive step uses the following claim.
\begin{claim}
  \label{claim:l-1}
  For $t\in [0,d-2]$, we have: $\xi_t \in \xi_{t+1} \cdot \left[p \cdot (1\pm p) \left(1 \pm \frac{30}{mp^{\log{k}}} \right) \right]^{2^{d-t-1}}$
\end{claim}

\noindent The proof of \cref{claim:l-1} is identical to that of \cref{claim:c-exp-1}, so we leave it out and refer the reader to the earlier proof. If $d = 1$, we already have the following from (\ref{eq:l-4},\ref{eq:l-5}):
\begin{align}
  \label{eq:l-6}
  \exp{\Lambda_{1,\bell}} \in \left( p - \frac{p^2}{4} \right) \pm \frac{7}{m} \subseteq p \cdot (1\pm p) \left(1 \pm \frac{35}{mp^{\log{k}}} \right)
\end{align}
where in the last bound we use the hypothesis that $mp^{\log{k}} > 30$ and $p \leq 1$. Suppose $d\in[2,\log{k}]$. Inductively applying \cref{claim:l-1} with $t$ going from $(d-2)$ to $0$, we get the following:
\begin{align}
  \label{eq:l-7}
  \xi_0 \in \xi_{d-1} \cdot \left[p \cdot (1\pm p) \left(1 \pm \frac{30}{mp^{\log{k}}} \right) \right]^{\sum_{t=0}^{d-2} 2^{d-t-1}} = \xi_{d-1} \cdot \left[p \cdot (1\pm p) \left(1 \pm \frac{30}{mp^{\log{k}}} \right) \right]^{2^{d}-2}
\end{align}
Putting together (\ref{eq:l-4}), (\ref{eq:l-5}) and (\ref{eq:l-7}), we get:
\begin{align}
  \label{eq:l-8}
  \exp{\Lambda_{d,\bar{\ell}}} = \xi_0 &\in \left(p - \frac{p^2}{4} \pm \frac{7}{mp^{d-1}} \right) \left(p \cdot (1\pm p) \left(1 \pm \frac{30}{mp^{\log{k}}} \right) \right)^{2^{d}-2} \nonumber\\
  &\subseteq {p^{2^{d-1}}} \cdot (1\pm p)^{2^{d-1}} \left(1 \pm \frac{35}{mp^{\log{k}}} \right)^{2^{d}-1}
\end{align}
Putting together (\ref{eq:l-3},\ref{eq:l-6},\ref{eq:l-8}), and using the shorthand $\alpha = (1\pm p) \left(1 \pm \frac{35}{mp^{\log{k}}} \right)$ we get the following for any $d\in[\log{k}]$:
\begin{align}
  \label{eq:l-9}
  \exp{\Lambda_d} = \sum_{\bar{\ell}\in[n]^{2^d}} \exp{\Lambda_{d,\bar{\ell}}} \in n^{2^{d}} p^{2^{d}-1} \alpha^{2^{d}-1} = n \cdot (np)^{2^d-1} \cdot \alpha^{2^d-1}
\end{align}
Putting together (\ref{eq:l-1},\ref{eq:l-2},\ref{eq:l-9}), we get:
\begin{align}
  \label{eq:l-10}
  \exp{\Lambda} =  \sum_{d\in[0,\log{k}]} \frac{k}{2^d} \cdot \exp{\Lambda_d} &\in k\cdot n + \sum_{d\in[\log{k}]} \frac{k}{2^d} \cdot n \cdot (np)^{2^{d}-1} \alpha^{2^{d}-1} \nonumber\\
  &= kn \cdot \sum_{d\in[0,\log{k}]} \frac{(np)^{2^d-1} \alpha^{2^d-1}}{2^d}\nonumber\\
  &\subseteq kn \cdot \alpha^{k-1} \sum_{d\in[0,\log{k}]} \frac{(np)^{2^d-1}}{2^d}
\end{align}
which proves the proposition.


\subsection{$k$-Tree over $\Int_m$}
\label{sec:zm}

In this section, we present our analysis of the $k$-Tree algorithm (from \cref{fig:ktree}) for the $k$-SUM problem over $\Int_m$ rather than over the integers. For convenience, we will restrict $m$ to be odd. We will identify the elements of $\Int_m$ with integers in the range $\range{m} = \set{-\floor{\frac{m}{2}}, \dots, \floor{\frac{m}{2}}}$ in the natural manner.

The only modification to the $k$-Tree algorithm itself is that the addition operation (as used by the Merge procedure) is addition modulo $m$ rather than addition over integers. A key observation here is that if $p < 1/2$, then this difference in the addition operation is only relevant in the first iteration of the algorithm's loop in step 3 -- that is, only when the initial input lists are merged. Thereafter, all the numbers in all the lists have absolute value at most $mp$, and thus adding them will never cause a ``wrap-around'' where the modulus operation with $m$ actually makes a difference. This lets us adapt our earlier analysis to this case with minimal modification. This is captured by the following theorem. As its proof is almost the same as that of \cref{thm:ktree}, we only provide a sketch that highlights and follows the differences between the proofs.

\begin{theorem}
  \label{thm:ktree-zm}
  Consider any $k, n, m\in \Nat$, where $k \geq 4$ is a power of $2$ and $m > 30^{\log{k}+1}$ is an odd number. Set $p = m^{\frac{-1}{\log{k}+1}}$ and $c = p\cdot n$. Consider $k$ lists $L_1,\dots,L_k$, each consisting of $n$ uniformly random elements from the range $\range{m}$. The $k$-Tree algorithm (as in \cref{fig:ktree}) over $\Int_m$ with these parameters, denoted by $\ktree_m$, satisfies the following:
  \begin{itemize}
    \item \textbf{Success Probability.} Its probability of success is bounded as follows:
    \begin{align*}
      \frac{1}{c^{-k} + \left(1+\frac{k}{n} \right)^k}\cdot (1-150p)^{k} \leq \prob{\substack{L_1,\dots,L_k}}{\substack{\ktree(L_1,\dots,L_k)\\\text{ outputs } (\ell_1,\dots,\ell_k)\\\text{such that } \sum_i L_i[\ell_i] = 0 }} \leq c^k \cdot (1 + 37p)^k
    \end{align*}
    \item \textbf{Complexity.} Its expected complexity is bounded as follows:
    \begin{align*}
      \expec{\substack{L_1,\dots,L_k}}{\substack{\text{Total size of all lists}\\\text{involved in}\\\ktree(L_1,\dots,L_k)}} \in k n \cdot \left( 1 + \sum_{d\in[\log{k}]} \frac{c^{2^d-1}}{2^d} \right) \cdot \left(1 \pm 37p \right)^{k-1}
    \end{align*}
  \end{itemize}
\end{theorem}

\begin{proofsketchof}{\cref{thm:ktree-zm}}
  The proofs of \cref{prop:c-exp,prop:c-var,prop:l} (which together imply the bounds in \cref{thm:ktree}) are built solely on top of \cref{prop:tools-2,prop:tools-3,prop:tools-4,prop:tools-5}, following by appropriate manipulations of the bounds provided by the latter. As observed above, looking at the iterations in step 3 of the $k$-Tree algorithm (as in \cref{fig:ktree}), all iterations from $d = 2$ onwards are identical for the algorithm over integers and over $\Int_m$, as there is no wrap-around in the addition operations involved, and that part of the analysis from the proof of \cref{thm:ktree} can be used as is.

  For the first iteration, the primary differentiating factor in the analysis of the algorithms in the two cases is that addition modulo $m$ of two uniformly random numbers from $\range{m}$ is slightly more likely to lead to an element in $\range{mp}$ than addition of such numbers over the integers. So the statements of \cref{prop:tools-2,prop:tools-3,prop:tools-4,prop:tools-5} might no longer be true. To redeem the proof, we instead provide the following alternatives to each of the propositions that may be used instead in the analysis of the first iteration. It may be verified that the bounds they provide suffice to make the proofs of \cref{prop:c-exp,prop:c-var,prop:l} continue to work even in the case of addition modulo $m$, which can then be used to prove \cref{thm:ktree-zm} in the same manner as \cref{thm:ktree}.

  \medskip
  \noindent\emph{Replacing \cref{prop:tools-2}:} Note that the sum of two uniformly random elements of $\Int_m$ is also a uniformly random element of $\Int_m$. This lets us replace the use of \cref{prop:tools-2} in the analysis of the first iteration with the following statement for any $m$ and $p\in[0,1]$:
  \begin{align}
    \label{eq:zm-1}
    \prob{x,y\gets \range{m}}{(x+y) (\text{mod } m) \in \range{mp}} = \frac{\size{\range{mp}}}{\size{\range{m}}} = \frac{2\floor{mp/2}+1}{m} \in \left( p \pm \frac{1}{m} \right)
  \end{align}

  \medskip
  \noindent\emph{Replacing \cref{prop:tools-3}:} As noted above, the sum of two uniformly random elements of $\Int_m$ is also a uniformly random element of $\Int_m$, and conditioning on this sum being in $\range{mp}$ results in a uniform distribution over $\range{mp}$. Thus the MR distance between this conditioned sum distribution and the uniform distribution over $\range{mp}$ is actually \emph{equal to} $1$, which also happens to satisfy the bound in \cref{prop:tools-3}.

  \medskip
  \noindent\emph{Replacing \cref{prop:tools-4}:} If $w$, $x$, and $y$ are sampled uniformly at random from $\Int_m$, the sums $(w+x)$ and $(w+y)$ are, in fact, independent and uniformly random over $\Int_m$. Thus, we have the following using (\ref{eq:zm-1}):
  \begin{align*}
    &\prob{w,x,y\gets \range{m}}{(w+x) (\text{mod } m) \in \range{mp} \wedge (w+y) (\text{mod } m) \in \range{mp}}\\
    &\quad = \prob{w,x\gets \range{m}}{(w+x) (\text{mod } m) \in \range{mp}} \cdot \prob{w,y\gets \range{m}}{(w+y) (\text{mod } m) \in \range{mp}}\\
    &\quad \leq p^2 \cdot \left( 1 + \frac{1}{mp} \right)^2
  \end{align*}
  which also satisfies the bound in \cref{prop:tools-4}.
  
  \medskip
  \noindent\emph{Replacing \cref{prop:tools-5}:} Following the arguments above, the values $(w+x)$ and $(w+y)$ conditioned on both being in $\range{mp}$ are also independent and uniformly distributed over $\range{mp}$. Thus the MR distance between this joint distribution and the uniform distribution over $\range{mp}\times\range{mp}$ is again equal to $1$, which also satisfies the bound in \cref{prop:tools-5}.
\end{proofsketchof}


\section{Computed Bounds}
\label{sec:program}



\noindent

The theoretical analysis presented in the previous sections provides asymptotically tight bounds for the $k$-Tree algorithm. However, for some practical parameter values that modern computers can handle, these bounds may not be as precise as one might desire due to the approximations used in the proofs. While these approximations do not affect the asymptotic tightness, they can have significant impact on the bounds for the parameter ranges we wish to evaluate empirically.

To address this limitation, we have implemented a set of computer programs that compute these bounds without the approximations used in the theoretical analysis. This approach allows us to obtain much tighter bounds for smaller parameter values, which are particularly relevant for practical applications and our experimental evaluations.

In this section, we present the computational methods used to calculate these precise bounds. We provide pseudo-code for each key function, along with corresponding theorem statements that relate these computations to our theoretical results and the $k$-Tree algorithm. By removing the approximations used in the manual proofs, these computational methods provide a bridge between our asymptotic analysis and the actual performance of the algorithm.
The bounds computed by these programs serve two crucial purposes:
\begin{itemize}
    \item Offer more accurate predictions of the algorithm's behavior for practical parameter ranges.
    \item Provide a means to validate our theoretical predictions against experimental measurements.
\end{itemize}

In the following subsections, we will detail each computational method, its relationship to the theoretical claims, propositions and theorems, and how it contributes to our understanding of the $k$-Tree algorithm's performance.

\subsubsection*{Notations}

To ensure consistency throughout this section, we use the following notations and definitions:

\begin{itemize}
    \item $k$: The number of input lists (always a power of 2).
    \item $n$: The size of each input list.
    \item $m$: The range from which the numbers in the lists are drawn.
    \item $p$: The $p$ parameter in the statement of~\cref{infthm:ktree}, computed as $p = m^{-\frac{1}{\log k + 1}}$.
\end{itemize}

\subsection{Algorithms}
\label{sec:pseudocode}

Here we present the algorithms that implement our approach to proving bounds on the $k$-Tree algorithm. We start with sub-routines that compute basic probabilities and distances before proceeding to the ones that actually implement the proof using them. The final bounds on the success probability are computed by the function \textsc{ProbBounds} (\cref{alg:prob_bounds}), and the bounds on the complexity by \textsc{SizeBounds} (\cref{alg:size_bounds}) \iflncs\else In each case, we state a claim about the behavior of the sub-routine, and if applicable refer to the corresponding analogue in \cref{sec:analysis} in this statement. We provide proof sketches for these claims except in cases where the implementation departs significantly from the proof in \cref{sec:analysis}, in which case we provide more detailed proofs.\fi

\begin{algorithm}[h!]
\caption{\textsc{ProbSumToZ}}
\label{alg:prob_sum_to_z}
\begin{algorithmic}[1]
\Input $s$, $z$ \Comment{Range size and target sum}
\Output Probability of the sum of two random variables in $\left<s\right>$ equaling $z$
\State $d \gets 2 \cdot \lfloor s/2 \rfloor + 1$
\State \Return $\frac{1}{d} - \frac{|z|}{d^2}$
\end{algorithmic}
\end{algorithm}

\iflncs\else
\begin{claim}[Claim~\ref{claim:tools-1}]
\label{prop:alg_prob_sum_to_z}
Given parameters $s$ and an integer $z \in [-s, s]$ as inputs, \textsc{ProbSumToZ} (Algorithm~\ref{alg:prob_sum_to_z}) outputs the exact probability $\Pr_{x,y \leftarrow U_s}[x + y = z]$, where $U_s$ is the uniform distribution over $\range{s}$. And the running time of the algorithm is $O(\polylog(s))$.
\end{claim}

\begin{proof}[Proof Sketch]
The computation is exactly the result of the proof for Claim~\ref{claim:tools-1} up to (\ref{eq:tools-1}):
\[
\Pr_{x,y \leftarrow U_s}[x + y = z] = \frac{1}{(2\floor{s/2}+1)} - \frac{\abs{z}}{(2\floor{s/2}+1)^2}
\]
\end{proof}
\fi


\begin{algorithm}[h!]
\caption{\textsc{ProbSumInRange}}
\label{alg:prob_sum_in_range}
\begin{algorithmic}[1]
\Input $s$, $p$ \Comment{Range size and $p$}
\Output Probability of sum being in specified range
\State $t \gets \lfloor (sp) / 2 \rfloor$
\State $d \gets 2 \cdot \lfloor s/2 \rfloor + 1$
\State $r \gets  \frac{2 \cdot t+1}{d} -\frac{t^2 + t}{d^2}$
\State \Return $r$
\end{algorithmic}
\end{algorithm}

\iflncs\else
\begin{proposition}[Proposition~\ref{prop:tools-2}]
\label{prop:alg_prob_sum_in_range}
For any $s$ and $p \in [0, 1]$ as inputs, \textsc{ProbSumInRange} (Algorithm~\ref{alg:prob_sum_in_range}) computes the probability $\Pr_{x,y \leftarrow U_s}[x + y \in \range{sp}]$ with running time being $O(\polylog(s))$.
\end{proposition}

\begin{proof}[Proof Sketch]
The function calculates this probability by summing the exact probabilities of each possible sum within the specified range as in Proposition~\ref{prop:tools-2}. We simplify it a bit by directly calculating the closed-form sum of $\Pr_{x,y \leftarrow U_s}[x + y = z]$ over $z$ instead of iteratively calling Algorithm~\ref{alg:prob_sum_to_z}:
\[
\Pr_{x,y \leftarrow U_s}[x + y \in \range{sp}] = \sum_{z = -\floor{sp/2}}^{\floor{sp/2}} \prob{x,y\gets U_s}{x+y=z} = \frac{(2\floor{sp/2}+1)}{(2\floor{s/2}+1)} - \frac{\floor{sp/2}(\floor{sp/2}+1)}{(2\floor{s/2}+1)^2}.
\]
\end{proof}
\fi


\begin{algorithm}[h!]
\caption{\textsc{MRDistFromUnif}}
\label{alg:mr_dist_from_unif}
\begin{algorithmic}[1]
\Input $s$, $p$ \Comment{Range size and probability}
\Output Max-ratio distance from uniform distribution
\State $s' \gets sp$
\State $\alpha \gets \textsc{ProbSumToZ}(s, 0)/\textsc{ProbSumInRange}(s, p)$
\State $\beta \gets \textsc{ProbSumToZ}(s, \lfloor s'/2 \rfloor)/\textsc{ProbSumInRange}(s, p)$
\State $u \gets \frac{1}{2 \cdot \lfloor s'/2 \rfloor + 1}$
\State \Return $\max (\alpha/u, u/\beta)$ 
\end{algorithmic}
\end{algorithm}

\iflncs\else
\begin{proposition}[Proposition~\ref{prop:tools-3}]
\label{prop:alg_mr_dist_from_unif}
For any $s$ and $p \in [0, 1]$ as inputs, \textsc{MRDistFromUnif} (Algorithm~\ref{alg:mr_dist_from_unif}) computes $\Delta_{MR}(U_{sp}, 2 \cdot U_s|_{\range{sp}})$ -- the max-ratio distance between the uniform distribution over $\range{sp}$ and the distribution of the sum of two uniform random variables in $\range{s}$ conditioned on being in $\range{sp}$. And the running time of the algorithm is $O(\polylog(s))$.
\end{proposition}

\begin{proof}[Proof Sketch]
The computation follows the exact definition of Max-Ratio Distance:
\[
\Delta_{MR}(U_{sp}, 2 \cdot U_s|_{\range{sp}}) = \max \left( \frac{\max_{z\in\range{sp}} 2 \cdot U_s|_{\range{sp}}(z)}{U_{sp}}, \frac{U_{sp}}{\min_{z\in\range{sp}} 2 \cdot U_s|_{\range{sp}}(z)} \right)
.\]
And one can verify that $2 \cdot U_s|_{\range{sp}}(z)$ is maximized when $z = 0$ and minimized when $z = \pm \floor{sp}.$
\end{proof}
\fi

\begin{algorithm}[h!]
\caption{\textsc{ProbSumWithTwoRVInRange}}
\label{alg:prob_sum_with_two_rv_in_range}
\begin{algorithmic}[1]
\Input $s$, $p$ \Comment{Range size and probability}
\Output Probability of summing one r.v. with two other r.v.'s being in specified range respectively
\Require $0 \leq p \leq 1$
\State $\alpha = 2 \cdot (\floor{s/2} - \floor{sp/2}) \cdot \frac{|\range{sp}|^2}{|\range{s}|^3}$
\State $x = |\range{sp}| - 1$
\State $y = |\range{sp}| - 1 - \floor{sp/2}$
\State $\beta = \frac{2}{|\range{s}|^3} \left( x(x+1)(2x+1) - y(y+1)(2y+1) \right)/6$
\State \Return $\alpha + \beta$
\end{algorithmic}
\end{algorithm}

\iflncs\else
\begin{proposition}[Proposition~\ref{prop:tools-4}]
\label{prop:alg_prob_sum_with_two_rv_in_range}
    For any $s$ and $p \in [0, 1]$ as inputs, \textsc{ProbSumWithTwoRVInRange} (Algorithm~\ref{alg:prob_sum_with_two_rv_in_range}) computes the probability $\prob{w,x,y\gets U_s}{w+x \in \range{sp} \wedge w+y \in \range{sp}}$. And the running time of the algorithm is $O(\polylog(s))$.
\end{proposition}

\begin{proof}[Proof Sketch]
    The algorithm computes the probability following the first equality of (\ref{eq:prop-tools-4-ineq-2}):
    \[
        \prob{w,x,y\gets U_s}{w+x \in \range{sp} \wedge w+y \in \range{sp}} = \sum_{w = -\floor{sp/2}}^{\floor{sp/2}} U_s(w) \cdot \prob{x,y\gets U_s}{w+x\in\range{sp} \wedge w+y\in\range{sp}}.
    \]

    The calculation considers two cases:
    \begin{enumerate}
    \item When $w \in [-\floor{s/2} + \floor{sp/2}, \floor{s/2} - \floor{sp/2}]$, each term on the RHS equals
    $$U_s(w) \cdot \prob{x,y\gets U_s}{w+x\in\range{sp} \wedge w+y\in\range{sp}} = \frac{|\range{sp}|^2}{|\range{s}|^3},$$
    and $\alpha$ is summing this value over all such $w$.
    \item When $w \in [-\floor{s/2}, -\floor{s/2} + \floor{sp/2} - 1]$ or $w \in [\floor{s/2} - \floor{sp/2} + 1, \floor{s/2}]$, let $i = |w| - \floor{s/2} + \floor{sp/2}$. Then the probability on the RHS is
    \[
        \prob{x,y\gets U_s}{w+x\in\range{sp} \wedge w+y\in\range{sp}} = (\frac{|\range{sp}| - i}{|\range{s}|})^2.
    \]
    Summing all these values over $w$ results in a sum of some squares, which is computed as $\beta$ in the algorithm.
    \end{enumerate}
\end{proof}
\fi

\begin{algorithm}[h!]
\caption{\textsc{MRDistFromPairUnif}}
\label{alg:mr_dist_from_pair_unif}
\begin{algorithmic}[1]
\Input $s$, $p$ \Comment{Range size and probability}
\Output Max-ratio distance from uniform distribution for pairs
\State $\alpha \gets \frac{1}{(|\range{s}|)^2 \cdot \textsc{ProbSumWithTwoRVInRange}(s, p)}$
\State $\beta \gets \frac{2(\floor{s/2} - \floor{sp/2})}{(|\range{s}|)^3 \cdot \textsc{ProbSumWithTwoRVInRange}(s, p)}$
\State $u \gets \frac{1}{(|\range{sp}|)^2}$
\State \Return $\max(\alpha/u, u/\beta)$
\end{algorithmic}
\end{algorithm}

\iflncs\else
\begin{proposition}[Proposition~\ref{prop:tools-5}]
\label{prop:alg_mr_dist_from_pair_unif}
For any $s$ and $p \in [0, 0.5]$ as inputs, \textsc{MRDistFromPairUnif} (Algorithm~\ref{alg:mr_dist_from_pair_unif}) computes $\mr\left(U_{sp}\otimes U_{sp}, D \right)$ as defined in Proposition~\ref{prop:tools-5}. And the running time of the algorithm is $O(\polylog(s))$.
\end{proposition}

\begin{proof}[Proof Sketch]
The computation follows the definition of Max-Ratio Distance:
\[
    \mr\left(U_{sp}\otimes U_{sp}, D \right) = \max \left( \frac{\max_{(z_1, z_2)\in\range{sp}^2} D(z_1, z_2)}{U_{sp}\otimes U_{sp}}, \frac{U_{sp}\otimes U_{sp}}{\min_{(z_1, z_2)\in\range{sp}^2} D(z_1, z_2)} \right).
\]
The algorithm computes:
\begin{enumerate}
    \item The maximum probability $\max_{(z_1, z_2)\in\range{sp}^2} D(z_1, z_2)$ ($\alpha$) occurs when $(z_1, z_2) = (0, 0)$.
    \item The minimum probability $\min_{(z_1, z_2)\in\range{sp}^2} D(z_1, z_2)$ ($\beta$) occurs when $$(z_1, z_2) = (-\floor{sp/2}, \floor{sp/2})\textnormal{ or }(\floor{sp/2}, -\floor{sp/2}).$$
    \item The uniform probability ($u$) over $\range{sp}^2$.
\end{enumerate}
It then calculates the results by comparing these terms. The use of \textsc{ProbSumWithTwoRVInRange} in the denominator of $\alpha$ and $\beta$ accounts for the conditioning on the sums being in $\range{sp}$.
\end{proof}
\fi

\begin{algorithm}[h!]
\caption{\textsc{FirstMomentInductFactors}}
\label{alg:first_moment_induct_factors}
\begin{algorithmic}[1]
\Input $m$, $k$, $p$, $d$ \Comment{Input parameters}
\Output $\text{UB}$, $\text{LB}$ \Comment{Upper and lower bounds}
\State $s \gets m \cdot p^d$ 
\State $\alpha \gets \left(\textsc{ProbSumInRange}(s, p)\right)^{\frac{k}{2^{d+1}}}$
\State $\beta \gets \left(\textsc{MRDistFromUnif}(s, p)\right)^{\frac{k}{2^{d+1}}}$
\State $ \text{UB} \gets \alpha \cdot \beta$
\State $ \text{LB} \gets \alpha / \beta$
\State \Return $ \text{UB}$, $ \text{LB}$
\end{algorithmic}
\end{algorithm}

\iflncs\else
\begin{claim}[Claim~\ref{claim:c-exp-1}]
\label{claim:alg_first_moment_induct_factors}
For any $m$ and $k$ being powers of $2$, $p \in [0, 1]$ and $d \in [0, \log k - 1]$ as inputs, \textsc{FirstMomentInductFactors} (Algorithm~\ref{alg:first_moment_induct_factors}) computes factors $UB = \alpha \cdot \beta$ and $LB = \alpha/\beta$ such that $\zeta_{d+1} \cdot \text{LB} \leq \zeta_d \leq \zeta_{d+1} \cdot \text{UB}$, where:
\begin{align*}
    \alpha &= \prob{x^d_1,\dots,x^d_{k/2^d}\gets U_{mp^d}}{ \forall i\in \left[\frac{k}{2^{d+1}}\right]: (x^d_{2i-1}+x^d_{2i}) \in \range{mp^{d+1}} }, \\
    \beta &= \Delta_{MR}(U_{mp^{d+1}}, 2 \cdot U_{mp^d}|_{\range{mp^{d+1}}})^{\frac{k}{2^{d+1}}},
\end{align*}
and $\zeta_d$ is as defined in the proof of Proposition~\ref{prop:c-exp}. And the algorithm runs in $O(\polylog (m,k))$.
\end{claim}

\begin{proof}[Proof Sketch]
Algorithm~\ref{alg:first_moment_induct_factors} computes the exact values of the factors described in the proof of Proposition~\ref{prop:c-exp}. It calls \textsc{ProbSumInRange} to compute $\alpha$ as in (\ref{eq:c-exp-5}), and \textsc{MRDistFromUnif} to compute $\beta$ as in (\ref{eq:c-exp-6}). Then it computes the upper and lower bound factors accordingly.
\end{proof}
\fi


\begin{algorithm}[h!]
\caption{\textsc{FirstMomentBounds}}
\label{alg:first_moment_bounds}
\begin{algorithmic}[1]
\Input $m$, $k$, $n$, $p$ \Comment{Input parameters}
\Output $\text{UB}$, $\text{LB}$ \Comment{Upper and lower bounds}
\State $\text{UB} \gets n^k \cdot \frac{1}{2 \cdot \floor{mp^{\log{k}}} - 1}$
\State $\text{LB} \gets \text{UB}$
\For{$d \gets 0$ to $\log k - 1$}
    \State $\alpha_u, \alpha_l \gets \textsc{FirstMomentInductFactors}(m, k, p, d)$
    \State $\text{UB} \gets \alpha_u \cdot \text{UB}$
    \State $\text{LB} \gets \alpha_l \cdot \text{LB}$
\EndFor
\State \Return $\text{UB}$, $\text{LB}$
\end{algorithmic}
\end{algorithm}

\iflncs\else
\begin{proposition}[Proposition~\ref{prop:c-exp}]
\label{prop:alg_first_moment_bounds}
Given parameters $m$, $k$, $n$ and $p$ satisfying the hypotheses of Proposition~\ref{prop:c-exp} as inputs, \textsc{FirstMomentBounds} (Algorithm~\ref{alg:first_moment_bounds}) computes upper and lower bounds UB and LB such that:
\[
\text{LB} \leq \mathbb{E}[C] \leq \text{UB},
\]
where $C$ is the number of $k$-tuples $(l_1, \ldots, l_k) \in [n]^k$ such that $\sum_{i=1}^k L_i[l_i] = 0$, and $L_1, \ldots, L_k$ are the uniformly random input lists. And UB, LB are at least as tight as the bounds in Proposition~\ref{prop:c-exp}. And the running time of the algorithm is $O(\polylog (m,k,n))$.
\end{proposition}

\begin{proof}[Proof Sketch]
The algorithm first computes the base values of the upper and lower bounds as in (\ref{eq:c-exp-3}) and multiplies them with the scaling factor $n^k$ as in (\ref{eq:c-exp-sum-over-all-k-tuples}). Then it iteratively computes the induction factors using \textsc{FirstMomentInductFactors} as in (\ref{eq:c-exp-8}) and (\ref{eq:c-exp-9}). As the computations in the subroutines have no approximations involved, the bounds computed by Algorithm~\ref{alg:first_moment_bounds} are at least as tight as the ones in Proposition~\ref{prop:c-exp}.
\end{proof}
\fi

\begin{algorithm}[h!]
    \caption{\textsc{SecondMomentUB}}
    \label{alg:second_moment_ub}
    \begin{algorithmic}[1]
        \Input $m$, $k$, $n$, $p$ \Comment{Input parameters}
        \Output Recursive result of the second moment bound
        \If{$k = 1$}
        \State $u \gets \frac{1}{2\cdot\floor{m} + 1}$
        \State \Return $(n \cdot u) \cdot (n \cdot u + 1)$
        \Else
        \State $m' = m \cdot p$
        \State $k' = k/2$
        \State $n' = (\left(\textsc{ProbUnifXMRDist}(m, p)\right)^2 \cdot n^4 + 2 \cdot \left(\textsc{ProbUnifXMRDistPair}(m, p)\right) \cdot n^3)^\frac{1}{2}$
        \State \Return $\textsc{SecondMomentUB}(m', k', n', p)$
        \EndIf
    \end{algorithmic}
\end{algorithm}

\begin{algorithm}[h!]
\caption{\textsc{ProbUnifXMRDist}}
\label{alg:prob_unif_x_mr_dist}
\begin{algorithmic}[1]
\Input $m$, $p$ \Comment{Input parameters}
\Output Product of uniform probability and max-ratio distance
\State \Return $\textsc{ProbSumInRange}(m, p) \cdot \textsc{MRDistFromUnif}(m, p)$
\end{algorithmic}
\end{algorithm}

\begin{algorithm}[h!]
\caption{\textsc{ProbUnifXMRDistPair}}
\label{alg:prob_unif_x_mr_dist_pair}
\begin{algorithmic}[1]
\Input $m$, $p$ \Comment{Input parameters}
\Output Product of uniform probability and max-ratio distance for pairs
\State \Return $\textsc{ProbSumWithTwoRVInRange}(m, p) \cdot \textsc{MRDistFromPairUnif}(m, p)$
\end{algorithmic}
\end{algorithm}

\iflncs\else
\begin{proposition}[Proposition~\ref{prop:c-var}]
\label{prop:alg_second_moment_ub}
Given parameters $m$, $k$, $n$ and $p$ satisfying the hypotheses of Proposition~\ref{prop:c-var}
as inputs, \textsc{SecondMomentUB} (Algorithm~\ref{alg:second_moment_ub}) computes an upper bound on the second moment of the number of solutions, such that:
\[
\mathbb{E}[C^2] \leq \text{UB}
\]
and UB is at least as tight as the bound in Proposition~\ref{prop:c-var}. And the running time of the algorithm is $O(\polylog (m,k,n))$.
\end{proposition}

\begin{proof}
    Algorithm~\ref{alg:second_moment_ub} employs a recursive variant of the calculation presented in the proof of Proposition~\ref{prop:c-var} to compute the upper bound on the second moment. For consistency, we adhere to the notations defined in the proof of Proposition~\ref{prop:c-var}.
    
    We begin by refining several expressions from the proof of Proposition~\ref{prop:c-var} to enable recursive calculation. Recall the expression we want to calculate in (\ref{eq:c-var-2}):
    \[
    \exp{C^2} = \sum_{\bar{\ell},\bar{\ell}'\in[n]^k}{ \exp{C_{\bar{\ell}} \cdot C_{\bar{\ell}'}}} = \sum_{s\in \bset^k} \sum_{\bar{\ell},\bar{\ell}': \diff(\bar{\ell},\bar{\ell}') = s}{\exp{C_{\bar{\ell}} \cdot C_{\bar{\ell}'}}}.
    \]

    We continue to use the notations $(s^1, \dots, s^{\log{k}}) \gets tex(s)$ defined before~\cref{claim:c-var-2}. Unless overridden in the context, we write $s^d$ by assuming it is one of the corresponding outputs generated by $tex(s^0)$ for some $s^0$ either implicitly or explicitly. We further define $ex:\bset^{2^*} \ra \bset^{2^*/2}$ on input $s\in\bset^{2^*}$ that is of length power of $2$, outputs $s' \in \bset^{2^*/2}$ with half the length, where its $i^{\text{th}}$ bit is defined as follows, for $i\in[k/2^d]$:
    \begin{align*}
      s^{'}_{i} = \max(s_{2i-1},s_{2i}).
    \end{align*}
    Then we refine (\ref{eq:c-var-claim-2}) and provide its extended notation $\chi_d(s)$ with $s \in \bset^{\frac{k}{2^{d}}}$ as input, where for $\chi_0$, the input $s$ corresponds to $\diff(\bar{\ell},\bar{\ell}')$:
    \[
        \exp{C_{\bell} \cdot C_{\bell'}} = \chi_0(\delta(\bell,\bell')).
    \]
    The inclusion of input $s$ is by revisiting the definition of $\chi_d$ in the proof of Proposition~\ref{prop:c-var} and verifying that in each step of the recursive calculations, the corresponding $s^d$ is sufficient and necessary to determine the value of $\chi_d$. Putting this back to (\ref{eq:c-var-2}), we have the following expression:
    \begin{align}
        \label{eq:alg_second_moment_ub-main}
        \exp{C^2} = \sum_{s\in \bset^k} \sum_{\bar{\ell},\bar{\ell}': \diff(\bar{\ell},\bar{\ell}') = s}{\chi_0(s)}.
    \end{align}

    Then we recall the inequality stated by Claim~\ref{claim:c-var-4} and its precise form:
    \begin{align}\label{eq:c-var-4-refined}
        \chi_d(s^d) &\leq \chi_{d+1}(s^{d+1}) \cdot (\FuncA{mp^d}{p}\FuncB{mp^d}{p})^{\sum_{i \in \left[\frac{k}{2^{d+1}}\right]} \FuncInd{s^d_{2i-1} = s^d_{2i} = 0}} \nonumber \\
        &\cdot (\FuncA{mp^d}{p}\FuncB{mp^d}{p})^{2 \cdot \sum_{i \in \left[\frac{k}{2^{d+1}}\right]} \FuncInd{s^d_{2i-1} = s^d_{2i} = 1}} \cdot (\FuncC{mp^d}{p}\FuncD{mp^d}{p})^{\sum_{i \in \left[\frac{k}{2^{d+1}}\right]} \FuncInd{s^d_{2i-1} \neq s^d_{2i}}},
    \end{align}
    where $s^{d+1} = ex(s^d)$ and the notations $\FuncA{mp^d}{p}$, $\FuncB{mp^d}{p}$, $\FuncC{mp^d}{p}$, and $\FuncD{mp^d}{p}$ are defined as the shorthands for the following values:
    \begin{align*}
        \FuncA{mp^d}{p} &= \prob{x,y\gets U_{mp^d}}{(x + y) \in \range{mp^{d+1}}}, \\
        \FuncB{mp^d}{p} &= \mr\left(U_{mp^{d+1}}, 2\cdot U_{mp^d}|_{\range{mp^{d+1}}} \right), \\
        \FuncC{mp^d}{p} &= \prob{w,x,y\gets U_{mp^d}}{(w + x) \in \range{mp^{d+1}} \wedge (w + y) \in \range{mp^{d+1}}}, \\
        \FuncD{mp^d}{p} &= \mr(D^0,D^1) \textnormal{ in (\ref{eq:c-var-4-11})}.
    \end{align*}

    Recall that in the proof of Claim~\ref{claim:c-var-4}, we consider three cases regarding the values of $s^d_{2i-1}$ and $s^d_{2i}$: $s^d_{2i-1} = s^d_{2i} = 0$, $s^d_{2i-1} = s^d_{2i} = 1$ and $s^d_{2i-1} \neq s^d_{2i}$.  And in (\ref{eq:c-var-4-refined}), the base terms corresponding to the three cases are by putting together (\ref{eq:c-var-4-3}) and (\ref{eq:c-var-4-9}), (\ref{eq:c-var-4-4}) and (\ref{eq:c-var-4-10}), (\ref{eq:c-var-4-5}) and (\ref{eq:c-var-4-11}) respectively. For simplification, we define the following functions:
    \begin{align*}
        \gamma_{0}(s) &= \sum_{i \in \left[|s|/2\right]} \FuncInd{s_{2i-1} = s_{2i} = 0}, \\
        \gamma_{1}(s) &= \sum_{i \in \left[|s|/2\right]} \FuncInd{s_{2i-1} = s_{2i} = 1}, \\
        \gamma_{\neq}(s) &= \sum_{i \in \left[|s|/2\right]} \FuncInd{s_{2i-1} \neq s_{2i}}.
    \end{align*}

    We would then rewrite the recursive upper bound (\ref{eq:c-var-4-refined}) in terms of the following function $g$:
    \begin{align}\label{eq:alg_second_moment_ub_g_function}
        g(s, m, p) &= \begin{cases}
            \left(\prob{x\gets U_{mp^{\log{k}}}}{x = 0}\right)^{1 + wt(s)} & \text{if } \text{length}(s) = 1,\\
            g(ex(s), m \cdot p, p) \cdot \delta_{s,m,p} & \text{otherwise},
        \end{cases}
    \end{align}
    where $\delta_{s,m,p} = (\FuncA{m}{p}\FuncB{m}{p})^{\Gammazero{s}} \cdot (\FuncA{m}{p}\FuncB{m}{p})^{2 \cdot \Gammaone{s}} \cdot (\FuncC{m}{p}\FuncD{m}{p})^{\Gamman{s}}$. 
    
    It is easy to see that for $\chi_0(s^0)$, the $\delta_{s^0,m,p}$ in $g(s^0, m, p)$ is exactly the RHS of (\ref{eq:c-var-4-refined}) with $d = 0$. And in the subsequent $g(s^1, m \cdot p, p)$ (i.e., $\chi_1(ex(s^0))$) where $s^1 \gets \text{ex}(s^0)$, $\delta_{s^1,mp,p}$ equals
    \[
        (\FuncA{mp}{p}\FuncB{mp}{p})^{\Gammazero{s^1}} \cdot (\FuncA{mp}{p}\FuncB{mp}{p})^{2 \cdot \Gammaone{s^1}} \cdot (\FuncC{mp}{p}\FuncD{mp}{p})^{\Gamman{s^1}}.
    \]
    It is evident that the calculations align with (\ref{eq:c-var-4-refined}) in each inductive step. The base case when $\text{length}(s) = 1$ is derived from (\ref{eq:c-var-claim-3}), where the value of $\chi_{\log{k}}(s)$ can be computed precisely as
    \[
    \prob{x\gets U_{mp^{\log{k}}}}{x = 0} \cdot \prob{y\gets U_{mp^{\log{k}}}}{ y = 0 }
    \]
    if $s = 1$ and
    \[
    \prob{y\gets U_{mp^{\log{k}}}}{y = 0}
    \]
    otherwise.

    Now we substitute $\chi_0(s^0)$ in (\ref{eq:alg_second_moment_ub-main}) with the recursive upper bound $g(s^0, m, p)$ to obtain
    \begin{align}
        \label{eq:alg_second_moment_ub-main-2}
        \exp{C^2} &\leq \sum_{s^0\in \bset^k} \sum_{\bar{\ell},\bar{\ell}': \diff(\bar{\ell},\bar{\ell}') = s^0}{g(s^0, m, p)} \nonumber\\
        &= \sum_{s^0\in \bset^k} n^{(k - wt(s^0))} \cdot \left(n(n-1)\right)^{\text{wt}(s^0)} \cdot g(s^0, m, p) \nonumber\\
        &\leq \sum_{s^0\in \bset^k} n^{(k - wt(s^0))} \cdot \left(n^2\right)^{\text{wt}(s^0)} \cdot g(s^0, m, p).
    \end{align}

    To conclude our computation of the second moment upper bound recursively, we define a function $f(m, k, n, p)$ to capture (\ref{eq:alg_second_moment_ub-main-2}) more generally:
    \begin{align}
        \label{eq:f_function}
        f(m, k, n, p) &= \sum_{s\in \bset^k} n^{(k - wt(s))} \cdot \left(n^2\right)^{\text{wt}(s)} \cdot g(s, m, p).
    \end{align}
    Corresponding to the base case of $g$, the base case of $f$ is when $k=1$:
    \begin{align*}
        f(m, 1, n, p) &= n \cdot g(0, m, p) + n^2 \cdot g(1, m, p) \\
        &= \frac{n}{2 \cdot \floor{m/2} + 1}\left(\frac{n}{2 \cdot \floor{m/2} + 1} + 1\right).
    \end{align*}

    We now rewrite (\ref{eq:alg_second_moment_ub-main-2}) in a slightly different form:

    \begin{align}\label{eq:alg_second_moment_ub-main-4}
        f(m, k, n, p) &= \sum_{s\in \bset^{k}} \sum_{\bar{\ell},\bar{\ell}' \in [n]^{k}: \diff(\bar{\ell},\bar{\ell}') = s}{g(s,m,p)} \\
        &= \sum_{s'\in \bset^{\frac{k}{2}}} \sum_{s: \text{ex}(s) = s'} \sum_{\bar{\ell},\bar{\ell}': \diff(\bar{\ell},\bar{\ell}') = s}{g(s,m,p)} \nonumber\\
        &= \sum_{s'\in \bset^{\frac{k}{2}}} \sum_{s: \text{ex}(s) = s'} \sum_{\bar{\ell},\bar{\ell}': \diff(\bar{\ell},\bar{\ell}') = s}{g(s',mp,p) \cdot \delta_{s,m,p}} \nonumber\\
        &= \sum_{s'\in \bset^{\frac{k}{2}}} g(s',mp,p) \cdot \sum_{s: \text{ex}(s) = s'} \sum_{\bar{\ell},\bar{\ell}': \diff(\bar{\ell},\bar{\ell}') = s}{\delta_{s,m,p}} \nonumber\\
    \end{align}
    In the sum $\sum_{s: \text{ex}(s) = s'} \sum_{\bar{\ell},\bar{\ell}': \diff(\bar{\ell},\bar{\ell}') = s}{\delta_{s,m,p}}$, the value of $\delta_{s,m,p}$ is determined by the values of $s$, $m$ and $p$. We then count the total number of $\bell$ and $\bell'$ that satisfy $\diff(\bell, \bell') = s$ such that $ex(s) = s'$ then multiply the result by $\delta_{s,m,p}$:
    \begin{align}
        \label{eq:alg_second_moment_ub-main-5}
        \sum_{s: \text{ex}(s) = s'}{\delta_{s,m,p}} \sum_{\bar{\ell},\bar{\ell}': \diff(\bar{\ell},\bar{\ell}') = s}{1} &= 
        &&\sum_{s: \text{ex}(s) = s'}{\delta_{s,m,p}} \cdot (n^2)^{\Gammazero{s}}(n^2(n-1)^2)^{\Gammaone{s}}(n^2(n-1))^{\Gamman{s}}\\
        &\leq &&\sum_{s: \text{ex}(s) = s'}(\FuncA{m}{p}\FuncB{m}{p} \cdot n^2)^{\Gammazero{s}} \nonumber\\
        & &&\cdot (\FuncA{m}{p}\FuncB{m}{p} \cdot n^2)^{2 \cdot \Gammaone{s}} \cdot (\FuncC{m}{p}\FuncD{m}{p} \cdot n^3)^{\Gamman{s}}\\
        &= &&\sum_{i = 0}^{wt(s')}\binom{wt(s')}{i} \cdot (\FuncA{m}{p}\FuncB{m}{p} \cdot n^2)^{k/2 - wt(s')} \nonumber \\ 
        & &&\cdot (\FuncA{m}{p}\FuncB{m}{p} \cdot n^2)^{2 \cdot i} \cdot (2 \cdot \FuncC{m}{p}\FuncD{m}{p} \cdot n^3)^{wt(s') - i} \label{eq:alg_second_moment_ub-main-5-2}\\
        &= &&(\FuncA{m}{p}\FuncB{m}{p} \cdot n^2)^{k/2 - wt(s')} \nonumber \\
        & &&\cdot ((\FuncA{m}{p}\FuncB{m}{p} \cdot n^2)^2 + 2 \cdot \FuncC{m}{p}\FuncD{m}{p} \cdot n^3)^{wt(s')}.  \nonumber
    \end{align}
    The first equality is by counting the number of $\bell$ and $\bell'$ that satisfy $\diff(\bell, \bell') = s$ for a fixed $s$ and the following inequality is by upper bounding $n-1$ by $n$. The equality in (\ref{eq:alg_second_moment_ub-main-5-2}) is by counting the number of possible $s$ given $s'$, where $s'_i = 0$ implies $s_{2i-1} = s_{2i} = 0$ ($\Gammazero{s}$) and $s'_i = 1$ implies $s_{2i-1} = s_{2i} = 1$ ($\Gammaone{s}$) or $s_{2i-1} \neq s_{2i}$ ($\Gamman{s}$). Then we denote by $i$ the number of case $s_{2i-1} = s_{2i} = 1$ and sum over all possible $i \in [0, wt(s')]$. Finally, the last equality is by summing the binomial coefficients with the last two terms then applying Binomial Theorem. Combining the result with (\ref{eq:alg_second_moment_ub-main-4}), we have:
    \begin{align}
        \label{eq:alg_second_moment_ub-main-6}
        f(m, k, n, p) \leq& \sum_{s'\in \bset^{\frac{k}{2}}} g(s',mp,p) \cdot (\FuncA{m}{p}\FuncB{m}{p} \cdot n^2)^{k/2 - wt(s')} \nonumber \\
        &\qquad\qquad\qquad \cdot ((\FuncA{m}{p}\FuncB{m}{p} \cdot n^2)^2 + 2 \cdot \FuncC{m}{p}\FuncD{m}{p} \cdot n^3)^{wt(s')}  \nonumber \\
        \leq& f(mp, k/2, ((\FuncA{m}{p}\FuncB{m}{p} \cdot n^2)^2 + 2 \cdot \FuncC{m}{p}\FuncD{m}{p} \cdot n^3)^{1/2}, p).
    \end{align}

    Finally, we note that Algorithm~\ref{alg:prob_unif_x_mr_dist} and Algorithm~\ref{alg:prob_unif_x_mr_dist_pair} computes the products of $\FuncA{m}{p}$ and $\FuncB{m}{p}$, and $\FuncC{m}{p}$ and $\FuncD{m}{p}$. Combining all these elements, the function $f$ precisely follows the procedure defined in \textsc{SecondMomentUB} (Algorithm~\ref{alg:second_moment_ub}). Since in each recursive call, the total running time of all operations is $O(\polylog (m,k,n))$, the running time of the algorithm is determined by multiplying the number of total recursive calls, which gives $O(\log k \cdot \polylog (m,k,n)) = O(\polylog (m,k,n))$. 
\end{proof}
\fi
\begin{algorithm}[h!]
    \caption{\textsc{ProbBounds}}
    \label{alg:prob_bounds}
    \begin{algorithmic}[1]
    \Input $m$, $k$, $n$, $p$
    \Output Bounds on the success probabilities
    \State $\alpha_u, \alpha_l \gets \textsc{FirstMomentBounds}(m, k, n, p)$
    \State $\beta \gets \textsc{SecondMomentUB}(m, k, n, p)$
    \State $\text{UB} \gets \min(\alpha_u, 1)$
    \State $\text{LB} \gets \frac{\alpha_l^2}{\beta}$
    \State \Return $\text{UB}$, $\text{LB}$
    \end{algorithmic}
\end{algorithm}

\iflncs\else
\begin{theorem}
    \label{thm:alg_prob_bounds}
    Given parameters $m$, $k$, $n$, $p$ satisfying the hypotheses in Theorem~\ref{thm:ktree} as inputs, \textsc{ProbBounds} (Algorithm~\ref{alg:prob_bounds}) computes upper and lower bounds UB and LB on the success probability of the $\ktree$ algorithm. And the running time of the algorithm is $O(\polylog (m,k,n))$.
\end{theorem}
    
\begin{proof}[Proof Sketch]
\textsc{ProbBounds} (Algorithm~\ref{alg:prob_bounds}) combines the results from \textsc{FirstMomentBounds} and \textsc{SecondMomentUB} to compute the final bounds by first computing the upper and lower bounds $\text{UB}_{1}$ and $\text{UB}_{2}$ of the first moment and the upper bound $\text{UB}_{2}$ of the second moment. Then applying the Paley-Zygmund inequality to obtain LB $= \pr{C > 0} \geq \frac{\exp{C}^2}{\exp{C^2}} \geq \frac{\text{LB}_{1}^2}{\text{UB}_{2}}$ and Markov's inequality to obtain UB = $\pr{C \geq 1} \leq \exp{C} \leq \text{UB}_{1}$.
\end{proof}
\fi


\begin{algorithm}[h!]
\caption{\textsc{SizeBoundInductFactors}}
\label{alg:size_bound_induct_factors}
\begin{algorithmic}[1]
\Input $m$, $k$, $p$, $d$, $t$ \Comment{Input parameters}
\Output $\text{UB}$, $\text{LB}$ \Comment{Upper and lower bounds}
\State $s \gets mp^t$
\State $\alpha \gets \textsc{ProbSumInRange}(s, p)$
\State $\beta \gets \textsc{MRDistFromUnif}(s, p)$
\State $(\alpha \cdot \beta)^{2^{d-t-1}}$
\State $(\alpha / \beta)^{2^{d-t-1}}$
\State \Return $\text{UB}$, $\text{LB}$
\end{algorithmic}
\end{algorithm}




\begin{algorithm}[h!]
\caption{\textsc{SizeBounds}}
\label{alg:size_bounds}
\begin{algorithmic}[1]
\Input $m$, $k$, $n$, $p$ \Comment{Input parameters}
\Output $\text{UB}$, $\text{LB}$ \Comment{Upper and lower bounds}
\State $\text{UB} \gets 0$, $\text{LB} \gets 0$
\For{$d \gets 0$ to $\log k$}
    \If{$d = 0$}
        \State $\text{UB} \gets \text{UB} + k \cdot n$
        \State $\text{LB} \gets \text{LB} + k \cdot n$
    \ElsIf{$d = 1$}
        \State $\alpha \gets \frac{k}{2} \cdot \textsc{ProbSumInRange}(m, p) \cdot n^2$
        \State $\text{UB} \gets \text{UB} + \alpha$
        \State $\text{LB} \gets \text{LB} + \alpha$
    \Else
        \State $s' \gets mp^{d-1}$
        \State $\gamma_u \gets \textsc{ProbSumInRange}(s', p)$, $\gamma_l \gets \textsc{ProbSumInRange}(s', p)$
        \For{$t \gets 0$ to $d-2$}
            \State $\beta_u, \beta_l \gets \textsc{SizeBoundInductFactors}(m, k, p, d, t)$
            \State $\gamma_u \gets \gamma_u \cdot \beta_u$
            \State $\gamma_l \gets \gamma_l \cdot \beta_l$
        \EndFor
        \State $\text{UB} \gets \text{UB} + \frac{k}{2^d} \cdot n^{2^d} \cdot \gamma_u$
        \State $\text{LB} \gets \text{LB} + \frac{k}{2^d} \cdot n^{2^d} \cdot \gamma_l$
    \EndIf
\EndFor
\State \Return $\text{UB}$, $\text{LB}$
\end{algorithmic}
\end{algorithm}

\iflncs\else
\begin{proposition}[Proposition~\ref{prop:l}]
\label{prop:alg_size_bounds}
Given parameters $m$, $k$, $n$ and $p$ satisfying the hypotheses in Proposition~\ref{prop:l} as inputs, \textsc{SizeBounds} (Algorithm~\ref{alg:size_bounds}) computes upper and lower bounds UB and LB on the expected total size of all lists involved in the execution of the $k$-Tree algorithm, such that $\text{LB} \leq \mathbb{E}[\Lambda] \leq \text{UB}$ and
\begin{align*}
    \text{UB} &= nk + \frac{kn^2}{2} \cdot \prob{x_1,x_2 \gets U_{m}}{x_1+x_2\in \range{mp}} \\
    &+ \sum_{d\in[2,\log{k}]} \frac{k}{2^d} \cdot n^{2^{d}} \cdot \prob{x_1,x_2 \gets U_{mp^{d-1}}}{x_1+x_2\in \range{mp^d}} \cdot \left[ \alpha \cdot \beta \right]^{\sum_{t=0}^{d-2} 2^{d-t-1}},\\
    \text{LB} &= nk + \frac{kn^2}{2} \cdot \prob{x_1,x_2 \gets U_{m}}{x_1+x_2\in \range{mp}} \\
    &+ \sum_{d\in[2,\log{k}]} \frac{k}{2^d} \cdot n^{2^{d}} \cdot \prob{x_1,x_2 \gets U_{mp^{d-1}}}{x_1+x_2\in \range{mp^d}} \cdot \left[ \alpha / \beta \right]^{\sum_{t=0}^{d-2} 2^{d-t-1}},\\
\end{align*}
where
\begin{align*}
    \alpha &= \prob{x^d_1,\dots,x^d_{2^{d-t}}\gets U_{mp^t}}{ \forall i\in \left[2^{d-t-1}\right]: (x^d_{2i-1}+x^d_{2i}) \in \range{mp^{t+1}} }, \\
    \beta &= \Delta_{MR}(U_{mp^{t+1}}, 2 \cdot U_{mp^t}|_{\range{mp^{t+1}}})^{2^{d-t-1}}.
\end{align*}
$\Lambda$ is the sum of the sizes of all lists generated during the algorithm's execution (as defined above in Proposition~\ref{prop:l}). And the running time of the algorithm is $O(\polylog (m,k,n))$.
\end{proposition}

\begin{proof}[Proof Sketch]
Similar to the proof of Proposition~\ref{prop:l}, the proof of the property of \textsc{SizeBoundInductFactors} (Algorithm~\ref{alg:size_bounds}) can be referred to the proof sketch of Proposition~\ref{prop:alg_first_moment_bounds}. In Algorithm \textsc{SizeBounds}, the total size of all lists is computed by summing the expected sizes of the lists generated at each level of the $k$-Tree algorithm. For $d=0$, the size is $k \cdot n$ as each list has $n$ elements as in (\ref{eq:l-2}). For $d=1$, the size is computed based on the probability of the sum being in the specified range as in (\ref{eq:l-6}). For $d > 1$, the size is calculated using the induction factors computed by \textsc{SizeBoundInductFactors}, with the initial value $\gamma_u, \gamma_l$ ($\xi_{d-1}$) calculated by \textsc{ProbSumInRange} following (\ref{eq:l-5}). Then the results are scaled according to (\ref{eq:l-9}) and (\ref{eq:l-10}).
\end{proof}
\fi


\iflncs
\else
\begin{algorithm}[h!]
\caption{\textsc{MainTheorem}}
\label{alg:main_theorem}
\begin{algorithmic}[1]
\Input $m$, $k$, $n$
\Output Bounds on the probabilities and sizes
\State $p \gets m^{-\frac{1}{\log k + 1}}$
\State $\text{UB}_{success}, \text{LB}_{success} \gets \textsc{ProbBounds}(m, k, n, p)$
\State $\text{UB}_{size}, \text{LB}_{size} \gets \textsc{SizeBounds}(m, k, n, p)$
\State \Return $\text{UB}_{success}$, $\text{LB}_{success}$, $\text{UB}_{size}$, $\text{LB}_{size}$
\end{algorithmic}
\end{algorithm}

The properties of \textsc{MainTheorem} (Algorithm~\ref{alg:main_theorem}) follows directly from Theorem~\ref{thm:alg_prob_bounds} and Proposition~\ref{prop:alg_size_bounds}, and the proof is omitted for brevity.
\fi


\iflncs\else
\subsection{Computed Bounds for $\Int_m$}

To compute bounds for the $k$-SUM problem over $\Int_m$, we provide several slightly modified subroutines to account for the differences in probability calculations when working with modular arithmetic. The key modifications are based on the observations in Section~\ref{sec:zm}, particularly for the first iteration of the algorithm.

Notice that in the computed bounds, the only subroutines affected are \textsc{ProbSumInRange} (Algorithm~\ref{alg:prob_sum_in_range}), \textsc{MRDistFromUnif} (Algorithm~\ref{alg:mr_dist_from_unif}), \textsc{ProbSumWithTwoRVInRange} (Algorithm~\ref{alg:prob_sum_with_two_rv_in_range}) and \textsc{MRDistFromPairUnif} (Algorithm~\ref{alg:mr_dist_from_pair_unif}) during the first iteration of the algorithm. We first introduce the following modifications to \textsc{ProbSumInRange}:

\begin{algorithm}[h!]
\caption{\textsc{ProbSumModSInRange}}
\label{alg:prob_sum_in_range_mod}
\begin{algorithmic}[1]
\Input $s$, $p$ \Comment{Modulus and $p$}
\Output Probability of sum being in specified range
\State \Return $\frac{2 \cdot \floor{sp/2}+1}{s}$
\end{algorithmic}
\end{algorithm}

\textsc{ProbSumModSInRange} (Algorithm~\ref{alg:prob_sum_in_range_mod}) follows direct from (\ref{eq:zm-1}) in Section~\ref{sec:zm}, and we will omit a formal statement of its correctness here. Regarding \textsc{ProbSumWithTwoRVInRange} (Algorithm~\ref{alg:prob_sum_with_two_rv_in_range}), its $\Int_m$ calculation is simply the square of \textsc{ProbSumModSInRange}. And for \textsc{MRDistFromUnif} (Algorithm~\ref{alg:mr_dist_from_unif}) and \textsc{MRDistFromPairUnif} (Algorithm~\ref{alg:mr_dist_from_pair_unif}), we relies on the conclusion from Section~\ref{sec:zm} that the max-ratio distance from the uniform distribution is always $1$. Then we can simply replace any calls to \textsc{MRDistFromUnif} and \textsc{MRDistFromPairUnif} with constant $1$ in the first iteration of the algorithm.

To realize the modified calculations during the first iteration, we introduce the following modification to \textsc{FirstMomentInductFactors}, \textsc{SecondMomentUB}, \textsc{ProbUnifXMRDist}, \textsc{SizeBoundInductFactors} and \textsc{SizeBounds}.

\begin{algorithm}[h!]
    \caption{\textsc{FirstMomentInductFactorsMod}}
    \label{alg:first_moment_induct_factors_mod}
    \begin{algorithmic}[1]
    \Input $m$, $k$, $p$, $d$ \Comment{Input parameters}
    \Output $\text{UB}$, $\text{LB}$ \Comment{Upper and lower bounds}
    \State $s \gets m \cdot p^d$
    \If{$d = 0$}
        \State $\alpha \gets \left(\textsc{ProbSumModSInRange}(s, p)\right)^{\frac{k}{2^{d+1}}}$
        \State $\beta \gets 1$
    \Else
        \State $\alpha \gets \left(\textsc{ProbSumInRange}(s, p)\right)^{\frac{k}{2^{d+1}}}$
        \State $\beta \gets \left(\textsc{MRDistFromUnif}(s, p)\right)^{\frac{k}{2^{d+1}}}$
    \EndIf
    \State $ \text{UB} \gets \alpha \cdot \beta$
    \State $ \text{LB} \gets \alpha / \beta$
    \State \Return $ \text{UB}$, $ \text{LB}$
    \end{algorithmic}
\end{algorithm}

\begin{algorithm}[h!]
    \caption{\textsc{SecondMomentUBMod}}
    \label{alg:second_moment_ub_mod}
    \begin{algorithmic}[1]
        \If{$k = 1$}
        \State $u \gets \frac{1}{2\cdot\floor{m} + 1}$
        \State \Return $(n \cdot u) \cdot (n \cdot u + 1)$
        \Else
        \State $m' = m \cdot p$
        \State $k' = k/2$
        \State $n' = (\left(\textsc{ProbUnifXMRDist}(m, p, b)\right)^2 \cdot n^4 + 2 \cdot \left(\textsc{ProbUnifXMRDistPair}(m, p, b)\right) \cdot n^3)^\frac{1}{2}$
        \State \Return $\textsc{SecondMomentUB}(m', k', n', p, 0)$
        \EndIf
    \end{algorithmic}
\end{algorithm}

\begin{algorithm}[h!]
\caption{\textsc{ProbUnifModMXMRDist}}
\label{alg:prob_unif_x_mr_dist_mod}
\begin{algorithmic}[1]
\Input $m$, $p$, $b$ \Comment{Input parameters}
\Output Product of uniform probability and max-ratio distance
\If{$b = 0$}
    \State \Return $\textsc{ProbSumInRange}(m, p) \cdot \textsc{MRDistFromUnif}(m, p)$
\Else
    \State \Return $\textsc{ProbSumModSInRange}(m, p)$
\EndIf
\end{algorithmic}
\end{algorithm}

\begin{algorithm}[h!]
    \caption{\textsc{ProbUnifXMRDistPair}}
    \label{alg:prob_unif_x_mr_dist_pair_mod}
    \begin{algorithmic}[1]
    \Input $m$, $p$, $d$ \Comment{Input parameters}
    \Output Product of uniform probability and max-ratio distance for pairs
    \If{$b = 0$}
        \State \Return $\textsc{ProbSumWithTwoRVInRange}(m, p) \cdot \textsc{MRDistFromPairUnif}(m, p)$
    \Else
        \State \Return $(\textsc{ProbSumModSInRange}(m, p))^2$
    \EndIf
    \end{algorithmic}
\end{algorithm}

The modification in \textsc{SecondMomentUBMod} is realize by passing a flag $b$ to indicate the first iteration of the algorithm. Only the initial call to \textsc{SecondMomentUB} will be given $b = 1$, and it will remain $0$ in all the subsequent recursive calls.

\begin{algorithm}[h!]
    \caption{\textsc{SizeBoundInductFactorsMod}}
    \label{alg:size_bound_induct_factors_mod}
    \begin{algorithmic}[1]
    \Input $m$, $k$, $p$, $d$, $t$ \Comment{Input parameters}
    \Output $\text{UB}$, $\text{LB}$ \Comment{Upper and lower bounds}
    \State $s \gets mp^t$
    \If{$t = 0$}
        \State $\alpha \gets \left(\textsc{ProbSumModSInRange}(s, p)\right)$
        \State $\beta \gets 1$
    \Else
        \State $\alpha \gets \textsc{ProbSumInRange}(s, p)$
        \State $\beta \gets \textsc{MRDistFromUnif}(s, p)$
    \EndIf
    \State $(\alpha \cdot \beta)^{2^{d-t-1}}$
    \State $(\alpha / \beta)^{2^{d-t-1}}$
    \State \Return $\text{UB}$, $\text{LB}$
    \end{algorithmic}
\end{algorithm}

\begin{algorithm}[h!]
    \caption{\textsc{SizeBoundsMod}}
    \label{alg:size_bounds_mod}
    \begin{algorithmic}[1]
    \Input $m$, $k$, $n$, $p$ \Comment{Input parameters}
    \Output $\text{UB}$, $\text{LB}$ \Comment{Upper and lower bounds}
    \State $\text{UB} \gets 0$, $\text{LB} \gets 0$
    \For{$d \gets 0$ to $\log k$}
        \If{$d = 0$}
            \State $\text{UB} \gets \text{UB} + k \cdot n$
            \State $\text{LB} \gets \text{LB} + k \cdot n$
        \ElsIf{$d = 1$}
            \State $\alpha \gets \frac{k}{2} \cdot \textsc{ProbSumModSInRange}(m, p) \cdot n^2$ \Comment{Modified for $\Int_m$}
            \State $\text{UB} \gets \text{UB} + \alpha$
            \State $\text{LB} \gets \text{LB} + \alpha$
        \Else
            \State $s' \gets mp^{d-1}$
            \State $\gamma_u \gets \textsc{ProbSumInRange}(s', p)$, $\gamma_l \gets \textsc{ProbSumInRange}(s', p)$
            \For{$t \gets 0$ to $d-2$}
                \State $\beta_u, \beta_l \gets \textsc{SizeBoundInductFactorsMod}(m, k, p, d, t)$
                \State $\gamma_u \gets \gamma_u \cdot \beta_u$
                \State $\gamma_l \gets \gamma_l \cdot \beta_l$
            \EndFor
            \State $\text{UB} \gets \text{UB} + \frac{k}{2^d} \cdot n^{2^d} \cdot \gamma_u$
            \State $\text{LB} \gets \text{LB} + \frac{k}{2^d} \cdot n^{2^d} \cdot \gamma_l$
        \EndIf
    \EndFor
    \State \Return $\text{UB}$, $\text{LB}$
    \end{algorithmic}
\end{algorithm}

These modifications address the calculation of \textsc{ProbSumInRange} and \textsc{MRDistFromUnif} when the input $s = m$, corresponding to the first iteration of the algorithm over $\Int_m$. They allow us to compute bounds for the $k$-SUM problem over $\Int_m$ that share the same tightness as the bounds for the $k$-SUM problem over integers.

\fi

\subsection{Visualizations and Interpretations}

In this section, we present examples of our computing bounds for various parameter settings. These visualizations help us better understand the tightness of the upper and lower bounds of the success probability as key parameters, such as $m$ and $k$, change. We will also directly compare our bounds with our experimental results to validate the trends we observe in this section (see \iflncs\else \cref{sec:experiments} \fi for details).


\begin{figure}[h!]
    \centering
    \begin{subfigure}[b]{0.48\textwidth}
        \includegraphics[width=\textwidth]{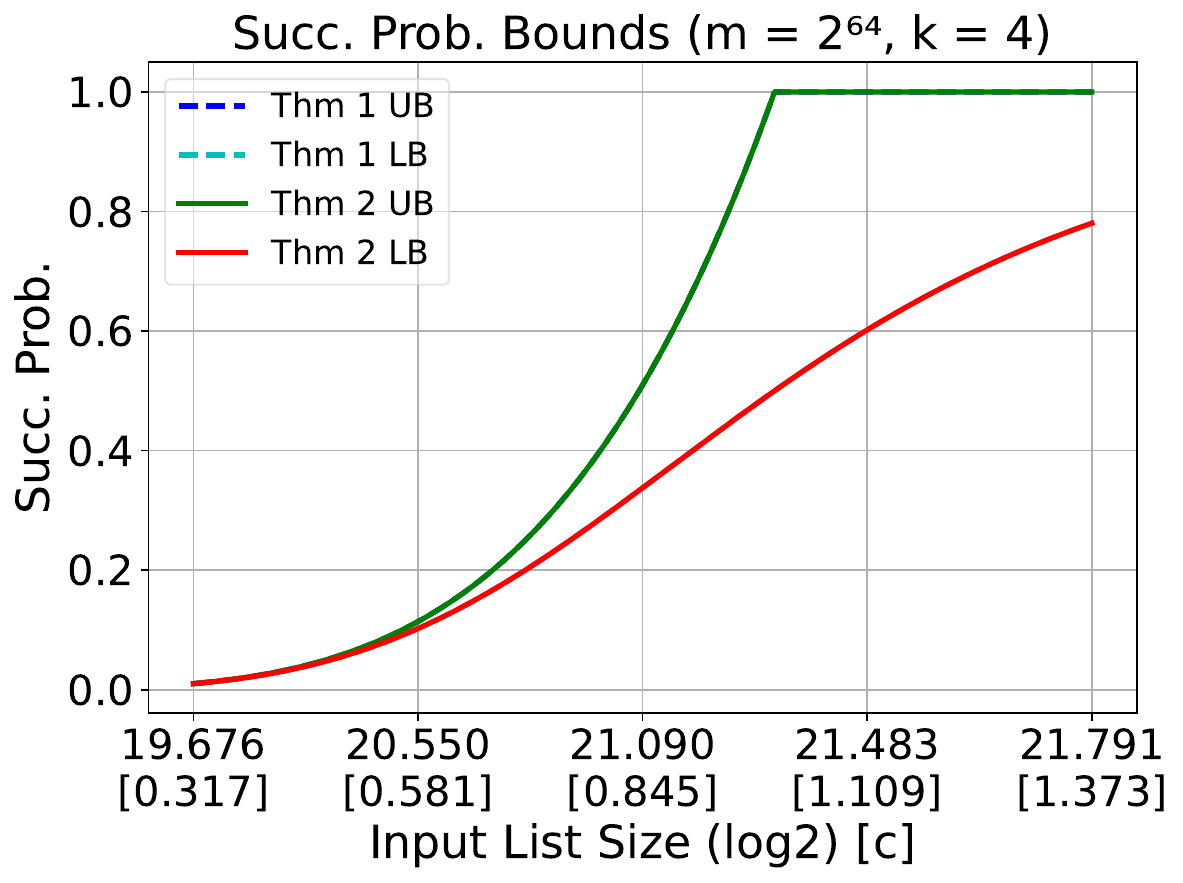}
        \caption{
            }
        \label{fig:comp_bound_type1-64-4}
    \end{subfigure}
    \hfill
    \begin{subfigure}[b]{0.48\textwidth}
        \includegraphics[width=\textwidth]{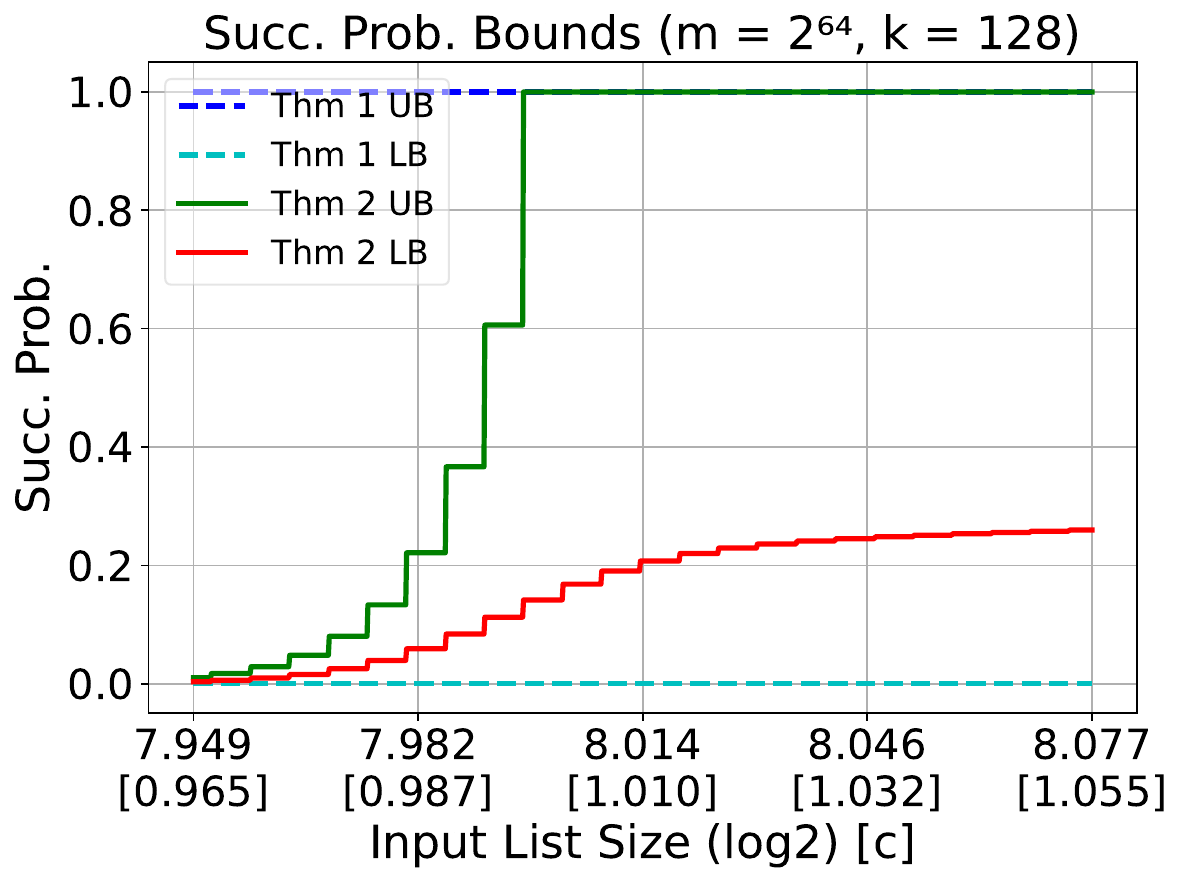}
        \caption{
            }
        \label{fig:comp_bound_type1-64-128}
    \end{subfigure}
    \vspace{0.5em}
    \begin{subfigure}[b]{0.48\textwidth}
        \includegraphics[width=\textwidth]{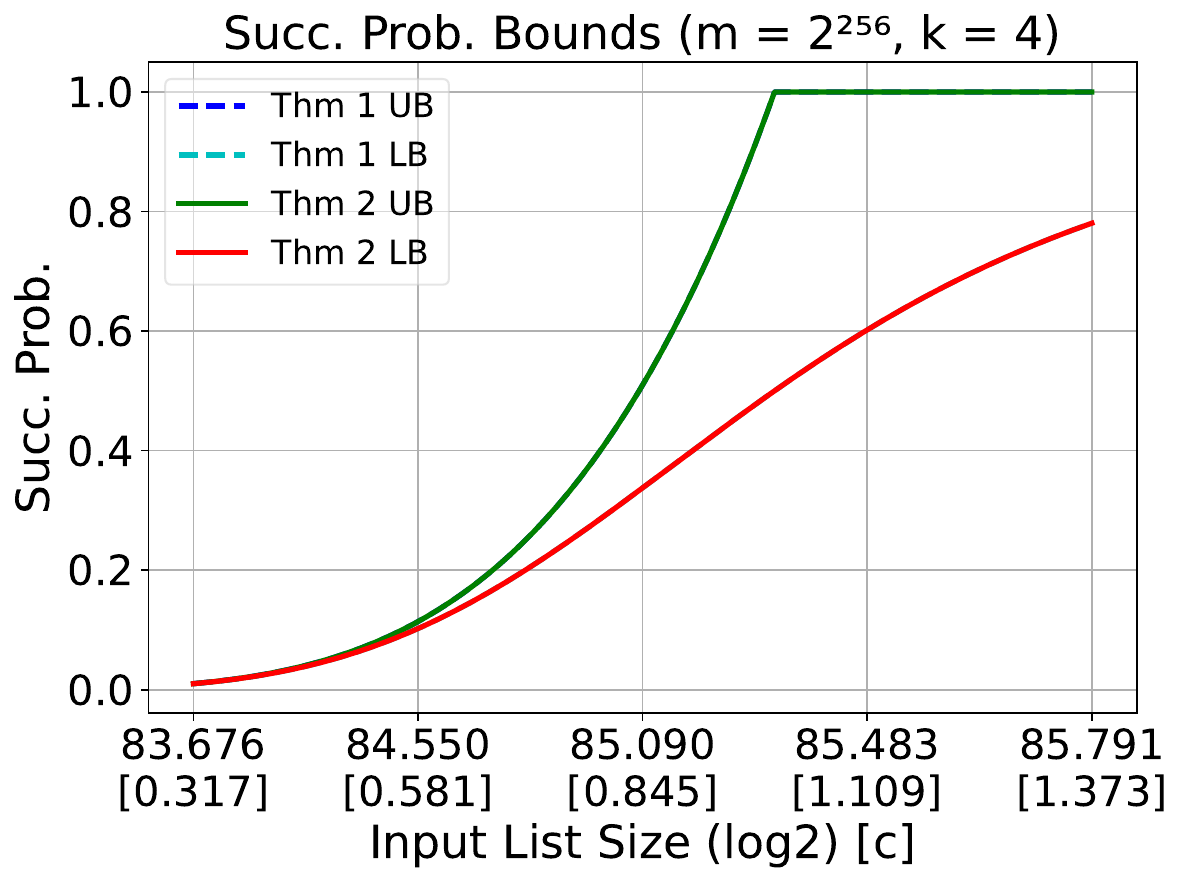}
        \caption{
            }
        \label{fig:comp_bound_type1-256-4}
    \end{subfigure}
    \hfill
    \begin{subfigure}[b]{0.48\textwidth}
        \includegraphics[width=\textwidth]{figures/intro/type1_m256_k1024.pdf}
        \caption{
            }
        \label{fig:comp_bound_type1-256-1024}
    \end{subfigure}
    \caption{Success probability upper and lower bounds when varying the input list size under fixed $m$'s and $k$'s. The values in the square brackets report $c = n/m^{1/(\log{k}+1)}$ as described in~\cref{infthm:ktree}. Note that the cascading patterns of the bounds for certain values of $m$ and $k$ are due to rounding the input list size to the nearest integer.}
    \label{fig:comp_bound_type1}
\end{figure}

\textbf{Success Probability Bounds.}
In Figures~\ref{fig:comp_bound_type1}, the most notable observation is that for a fixed $m$, the lower bounds on the success probability become looser as $k$ increases. This trend will be further illustrated when we later compare the bound with our empirical measurements of success probabilities. Such phenomenon is clearly visualized in the case where $m = 2^{64}$ and $k = 256$ (see Figure~\ref{fig:comp_bound_type1-64-128}). In this example, the lower bound flattens early and is far from $1$, indicating a quick relaxation of the bound as $k$ increases. This behavior shows that when $k$ is large, the success probability bounds are less effective for moderately small values of $m$.

However, as $m$ increases, the bounds become significantly tighter for all values of $k$. For instance, in the case where $m = 2^{256}$, the bounds $k = 1024$ (Figures~\ref{fig:comp_bound_type1-256-1024}) exhibits much better tightness compared to when $m = 2^{64}$ and $k = 128$. This improvement in the tightness of the bounds align with our conclusion \iflncs\else in Section~\ref{sec:analysis} \fi that our bounds become asymptotically tight as $m$ approaches infinity.

Overall, the visualizations indicate that larger values of $m$ lead to more refined bounds on success probability. And the impact of increasing $k$ is more pronounced at smaller values of $m$, while the bounds become more stabilized as $m$ increases.

\begin{figure}[h!]
    \centering
    \begin{subfigure}[b]{0.48\textwidth}
        \includegraphics[width=\textwidth]{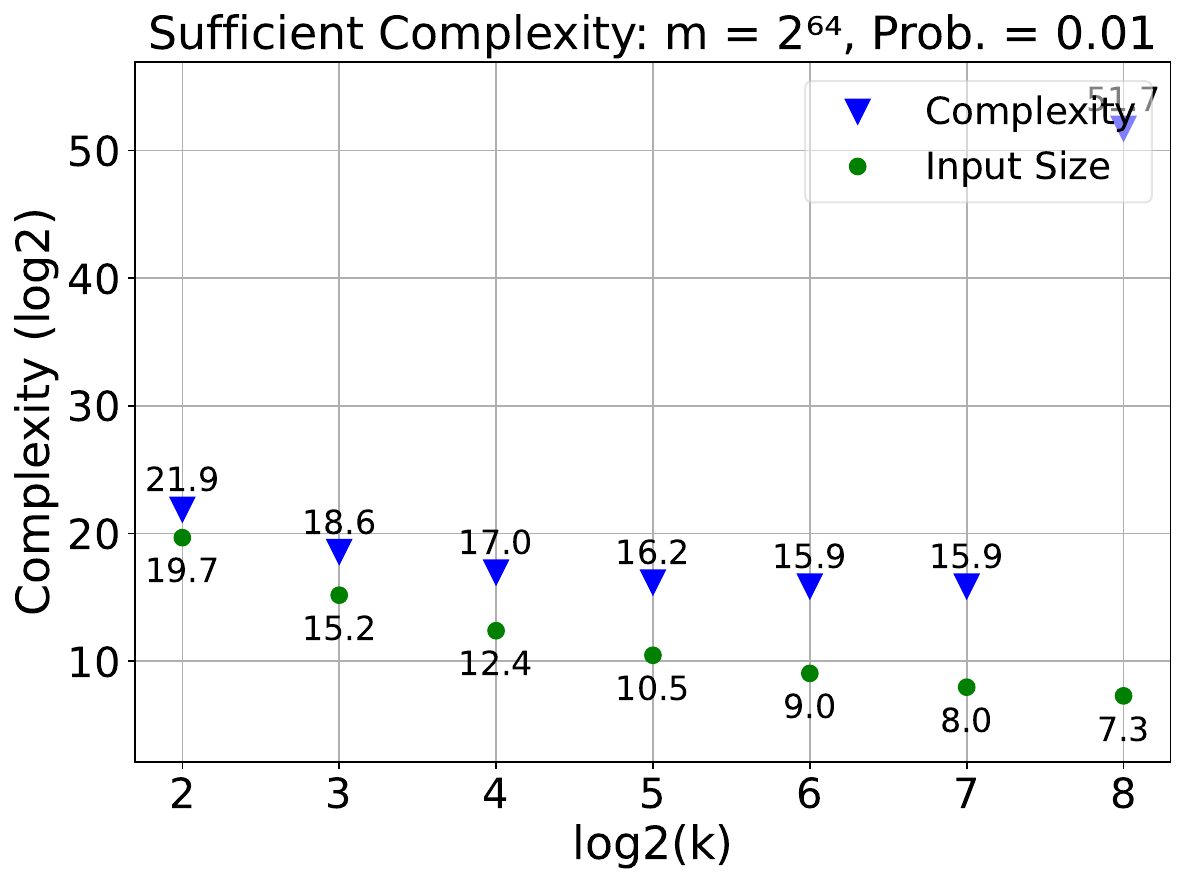}
        \caption{
            }
        \label{fig:comp_bound_type2-64-0.01-ub}
    \end{subfigure}
    \hfill
    \begin{subfigure}[b]{0.48\textwidth}
        \includegraphics[width=\textwidth]{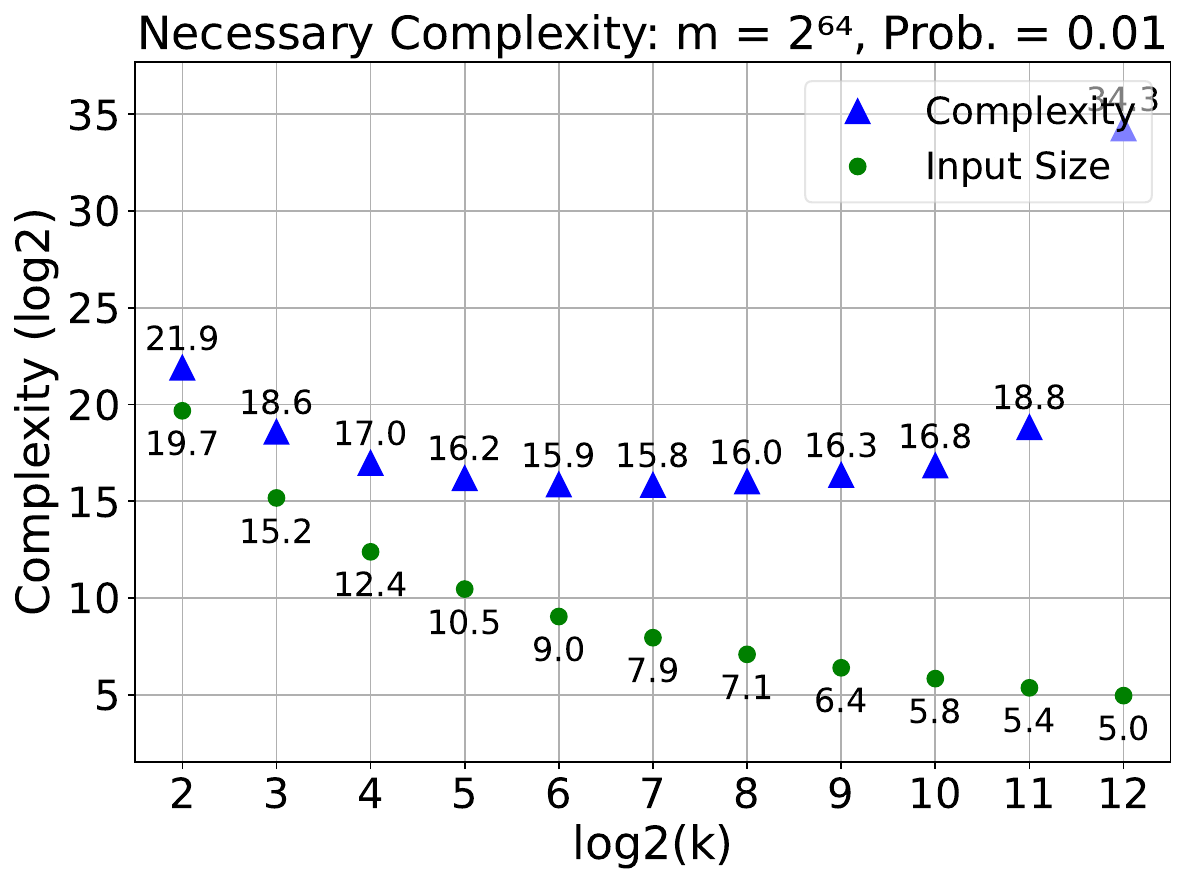}
        \caption{
            }
        \label{fig:comp_bound_type2-64-0.01-lb}
    \end{subfigure}
    \vspace{0.5em}
    \begin{subfigure}[b]{0.48\textwidth}
        \includegraphics[width=\textwidth]{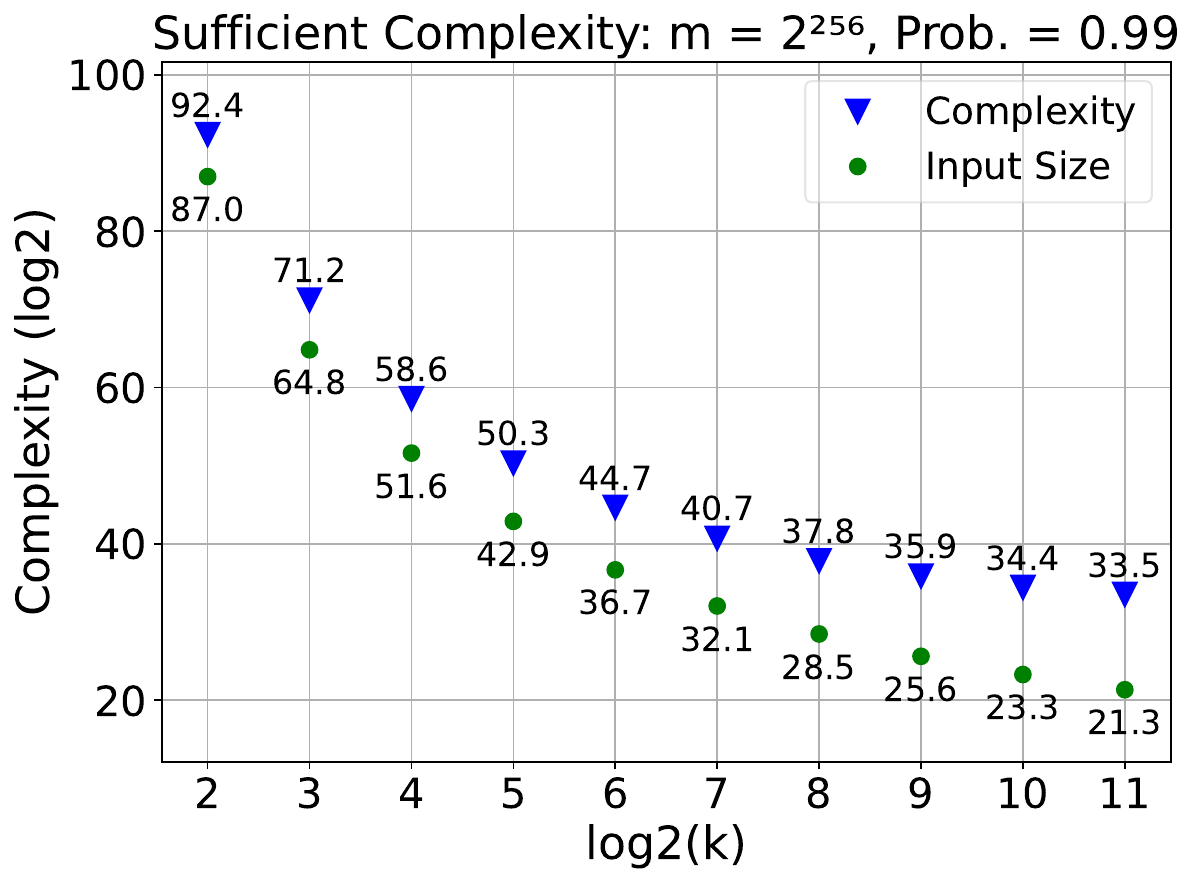}
        \caption{
            }
        \label{fig:comp_bound_type2-128-0.99-ub}
    \end{subfigure}
    \hfill
    \begin{subfigure}[b]{0.48\textwidth}
        \includegraphics[width=\textwidth]{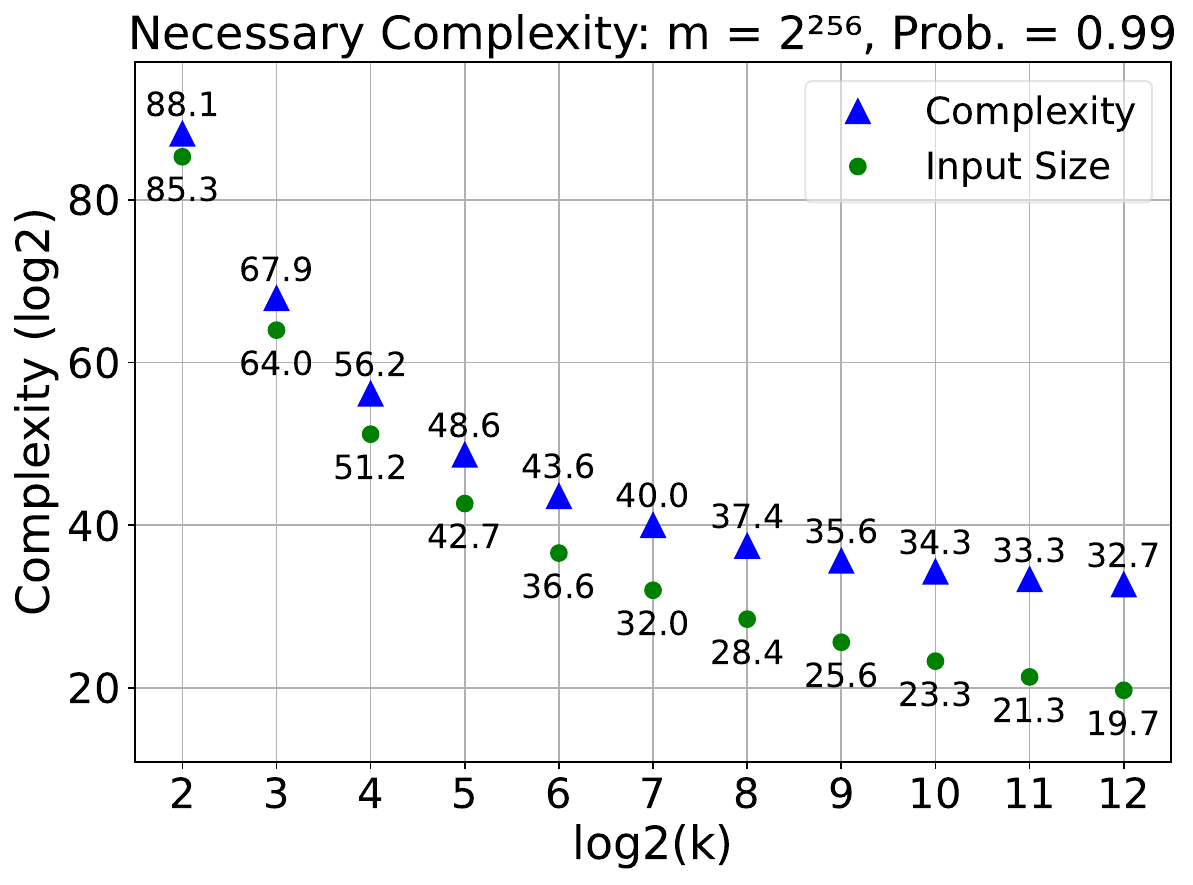}
        \caption{
            }
        \label{fig:comp_bound_type2-128-0.99-lb}
    \end{subfigure}
    \caption{Size upper and lower bounds when varying $k$ for fixed $m$'s (the fewer amount of $k$ in some of the Sufficient Complexity plots is due to the lower bound cannot reach the probability threshold even for very large $c$).}
    \label{fig:comp_bound_type2}
\end{figure}

\textbf{Complexity Bounds.}
In the examples shown in Figure~\ref{fig:comp_bound_type2}, we present the sufficient and necessary complexities (total list sizes) for achieving a target success probability. For the sufficient complexity, given a fixed $m$, we compute the minimum input list size via binary search such that our success probability lower bound equals or slightly larger than the specified probability threshold. We then use this minimum input list size as an input to compute our upper bound on the expected total size of all lists. For the necessary complexity, we first search for the minimum input list size such that our success probability upper bound equals or slightly smaller than the specified probability threshold. We then use this minimum input list size to compute our lower bound on the expected total size of all lists.

As observed in these figures, when $k$ increases, the upper bound on the total list size first decreases accordingly then increases again as $k$ continues to grow. This trend is more pronounced for small $m$'s, as shown in Figures~\ref{fig:comp_bound_type2-64-0.01-ub}. This behavior is consistent with the results shown in the Figures~\ref{fig:comp_bound_type1}, where the lower bounds on the success probability become looser as $k$ increases. The immediate consequence is that the input list size has to be substantially larger to achieve the same success probability when $k$ is large. That accounts for the surge of total list size observed in Figure~\ref{fig:comp_bound_type2-64-0.01-ub}. As $m$ increases, the total list sizes become more stable. In \iflncs\else Section~\ref{sec:experiments}\fi , we see a similar trend of decrease-then-increase of the complexities in our empirical measurements.

Overall, these visualizations provide a straightforward understanding of the efficacy of our method in estimating the computational resource for achieving a target success probability. \iflncs\else We will further validate these results through empirical evaluation in Section~\ref{sec:experiments}.\fi

\subsection{Practical Implications}

The precise bounds computed by our implementation have several significant implications for the practical use of the $\ktree$ algorithm:

\begin{enumerate}
    \item \textbf{Accurate Performance Prediction}: The tighter bounds allow for more accurate predictions of the algorithm's performance on real-world problem sizes. This is particularly crucial when deciding whether to employ the $\ktree$ algorithm for a specific task, as it provides a more reliable estimate of the required computational resources and expected success probability.

    \item \textbf{Optimal Parameter Selection}: With our precise bounds, practitioners can make better-informed decisions when selecting algorithm parameters. For instance, the optimal choice of the number of lists ($k$) can be more accurately determined to achieve desired success probabilities without actually running the algorithm for various $k$.

    \item \textbf{Confidence in Cryptographic Applications}: In cryptographic settings, where the $\ktree$ algorithm might be used for attacks or analysis, having precise bounds allows for more accurate estimation of the computational effort required for potential attacks, leading to better-informed decisions about key sizes and security parameters.
\end{enumerate}

\iflncs
\else
In summary, these computing bounds enhance our understanding of the $\ktree$ algorithm's behavior in practical settings and provide a crucial bridge between theoretical analysis and real-world application. As we move forward to the evaluation of the $\ktree$ algorithm, these computed bounds will serve as a valuable reference point and will be further validated against empirical observations.
\fi

\subsection{Limitations}



For small values of $m$, the tightness of our computed bounds degrade noticeably even when $k$ is in a reasonable range, despite they are still much tighter than the analytical bounds. This limitation can also be inferred from our analytical results, where our bounds loosen considerably after the point around $k > m^{1/(\log k + 1)}$.

Besides, for very large values of $m$ (e.g., $2^{512}$), our current Python implementation may encounter numerical stability issues due to overflow or underflow. To address this limitation, one can update our implementation by using Python Decimal library that handles large numbers with arbitrary precision with reasonable overhead.

\section{Experiments}
\label{sec:experiments}




The theoretical analysis and bounds computed in the previous sections are to offer insights into the expected performance of the $\ktree$ algorithm. However, empirical experiments are crucial to understand the algorithm's behavior under practical conditions and to validate these theoretical predictions. In this section, our primary objectives are as follows:
\begin{enumerate}
    \item Evaluate the behavior of the $\ktree$ algorithm under various parameter configurations;
    \item Compare the empirical results with our computed bounds.
\end{enumerate}

The results we demonstrate in this section provide empirical evidence on how different parameters influence the algorithm's success probability, time, and space complexity.

\subsection{Experimental Setup}

In discussions below, we use symbols defined in the context of the description of the $\ktree$ algorithm in \iflncs\else \cref{fig:ktree}\fi .

\paragraph{Algorithm Implementation.} We implemented the $\ktree$ algorithm in C++, closely following the description provided in \iflncs\else \cref{fig:ktree}\fi . The Merge subroutine was implemented by conducting a binary search over the sorted second list for each element in the first list.

\paragraph{Parameter Configurations.} We investigate the impact of varying the following parameters:

\begin{itemize}
    \item $m$: taking values from $\{2^{64}, 2^{96}, 2^{128}\}$. Due do limitations in computational resources, we were unable to explore larger values of $m$.
    \item $k$: The number of input lists, with values ranging from $4$ to $1024$, varied in powers of $2$. Note that not all $k$'s are feasible for all $m$'s. For example, when $m=2^{128}$, any $k$ smaller than $512$ results in a large input list size that exceeds the computational resources available. We report all the feasible results we have obtained.
    \item $n$: The size of each input list, dynamically adjusted via binary search to achieve target success probabilities varied from $0.01$ to $1.0$. Notice that not all the target probabilities are feasible due to either the nature of the problem (e.g., when $m$ is relatively small and $k$ is large, increasing the input list size by $1$ will increase the success probability from below $0.2$ to above $0.8$) or the limit of computational resources. We report all the feasible results we have obtained.
\end{itemize}


\paragraph{Evaluation Metrics.}
\begin{itemize}
    \item \textbf{Success Probability:} The fraction of trials where the algorithm successfully identifies a set of indices such that the indexed elements of the $k$ input lists sum to $0$.
    \item \textbf{Total Size}: The sum of the sizes of all lists generated when running the $\ktree$ algorithm, measured in terms of the number of elements:
    $$\sum_{d = 0}^{\log k} \sum_{i \in \left[\frac{k}{2^d}\right]} \left| L^d_i \right|.$$
    We use this measurement as a proxy for the running time of the algorithm (see \cref{rem:complexity}).
    \item \textbf{Max Level Size}: The maximum sum of the sizes of all lists generated at one level when running the $\ktree$ algorithm:
    $$\max_{d} \sum_{i \in \left[\frac{k}{2^d}\right]} \left| L^d_i \right|.$$
    This measurement corresponds to the space complexity of the algorithm.
\end{itemize}


\paragraph{Hardware and Software Environment.} All experiments were conducted on a machine running Ubuntu 24.04 LTS (64-bit) system on an Intel Core i7-12700 CPU @ 4.90GHz with 16 GB of RAM. The $\ktree$ algorithm was implemented in C++ 17, and the experimental results were managed and analyzed using Python 3.9 with libraries such as NumPy and Matplotlib for plotting results.



\subsection{Success Probability}









\paragraph{Objectives.} This section is to explore and measure the effect of varying the input list size on the success probability of the $\ktree$ algorithm under different configurations of $m$ and $k$.

\paragraph{Methodology.} For each combination of $m$ and $k$, we conducted a series of trials with varying input list sizes. We started with a list size corresponding to $c=1$, where $c = n/m^{1/(\log{k}+1)}$ as defined in Theorem~\ref{infthm:ktree}. We then used binary search to iteratively adjust the list size, exploring a range of $c$ values that span a typical set of success probabilities (i.e., percentage points). For each parameter configuration, we performed 1000 independent trials to calculate the empirical success probability. Then we further compare the empirical results with our theoretical bounds.

\begin{figure}[h!]
    \centering
    \begin{subfigure}[b]{0.48\textwidth}
        \includegraphics[width=\textwidth]{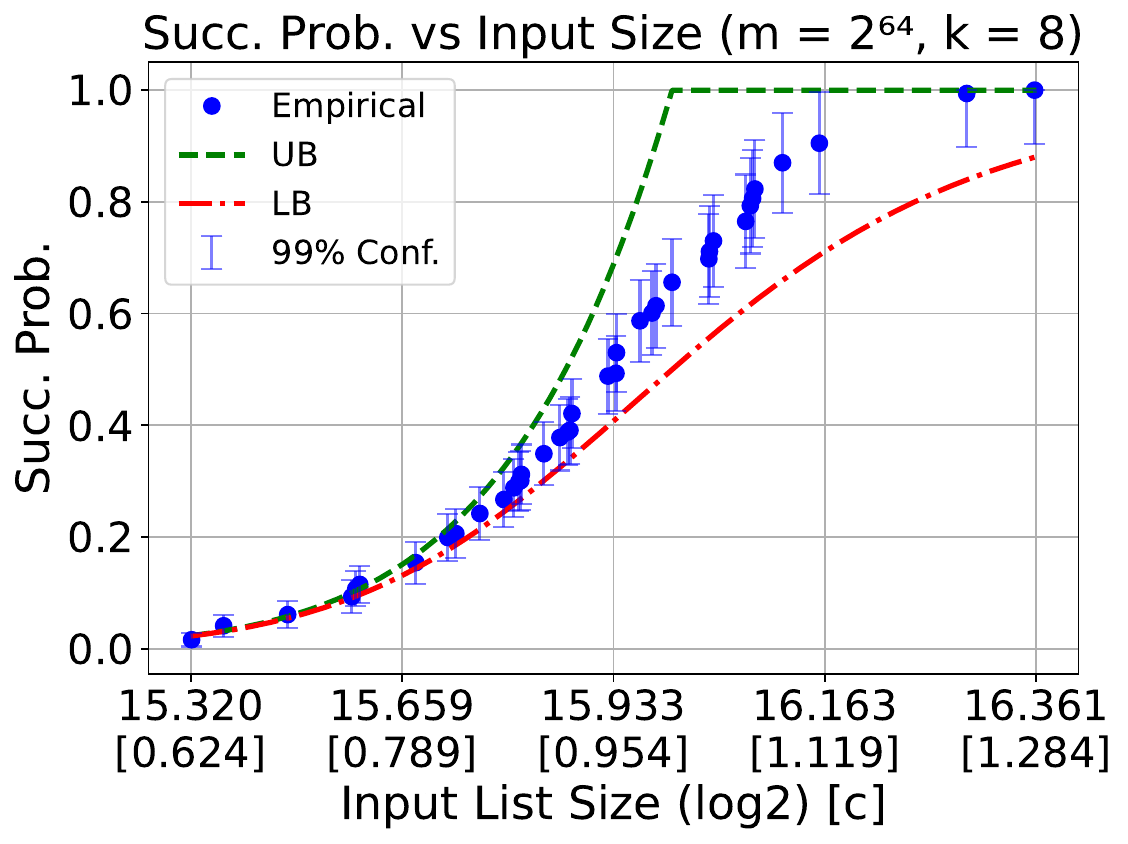}
        \caption{
            }
        \label{fig:type1-64-8}
    \end{subfigure}
    \hfill
    \begin{subfigure}[b]{0.48\textwidth}
        \includegraphics[width=\textwidth]{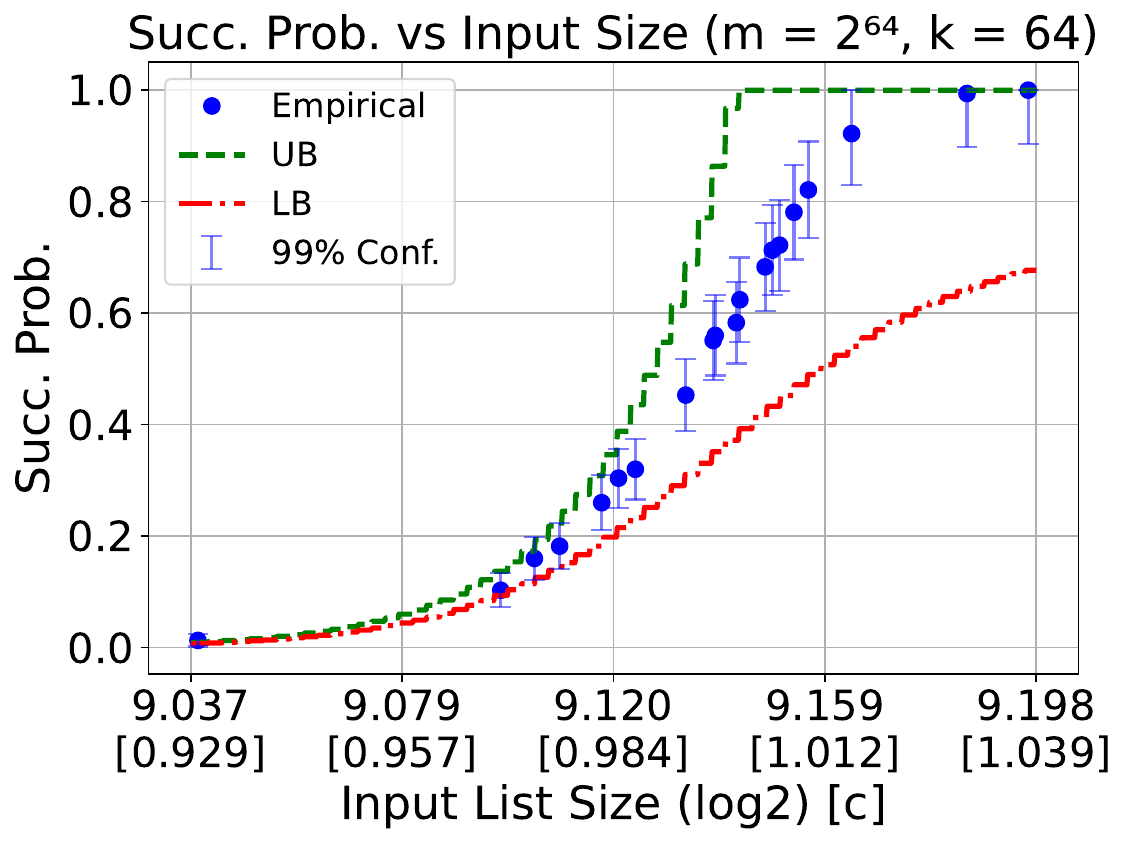}
        \caption{
            }
        \label{fig:type1-64-64}
    \end{subfigure}

    \vspace{1em}

    \begin{subfigure}[b]{0.48\textwidth}
        \includegraphics[width=\textwidth]{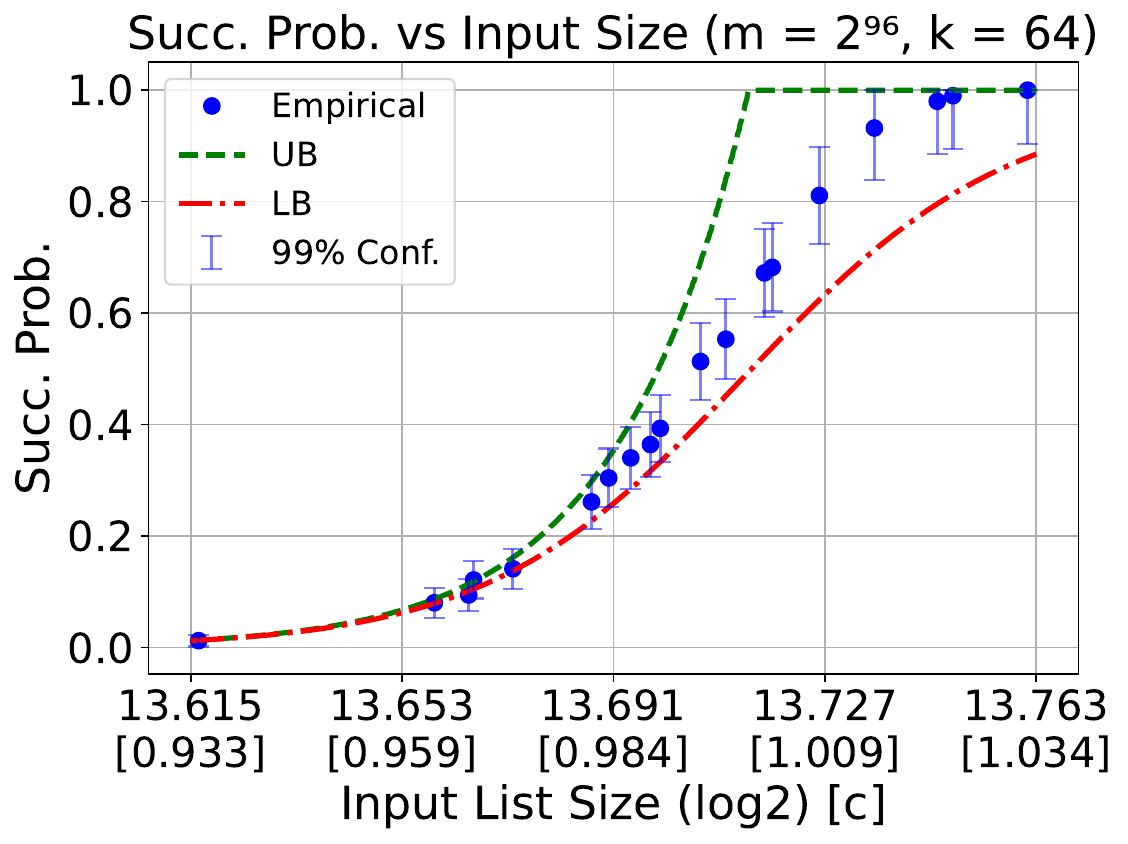}
        \caption{
            }
        \label{fig:type1-96-64}
    \end{subfigure}
    \hfill
    \begin{subfigure}[b]{0.48\textwidth}
        \includegraphics[width=\textwidth]{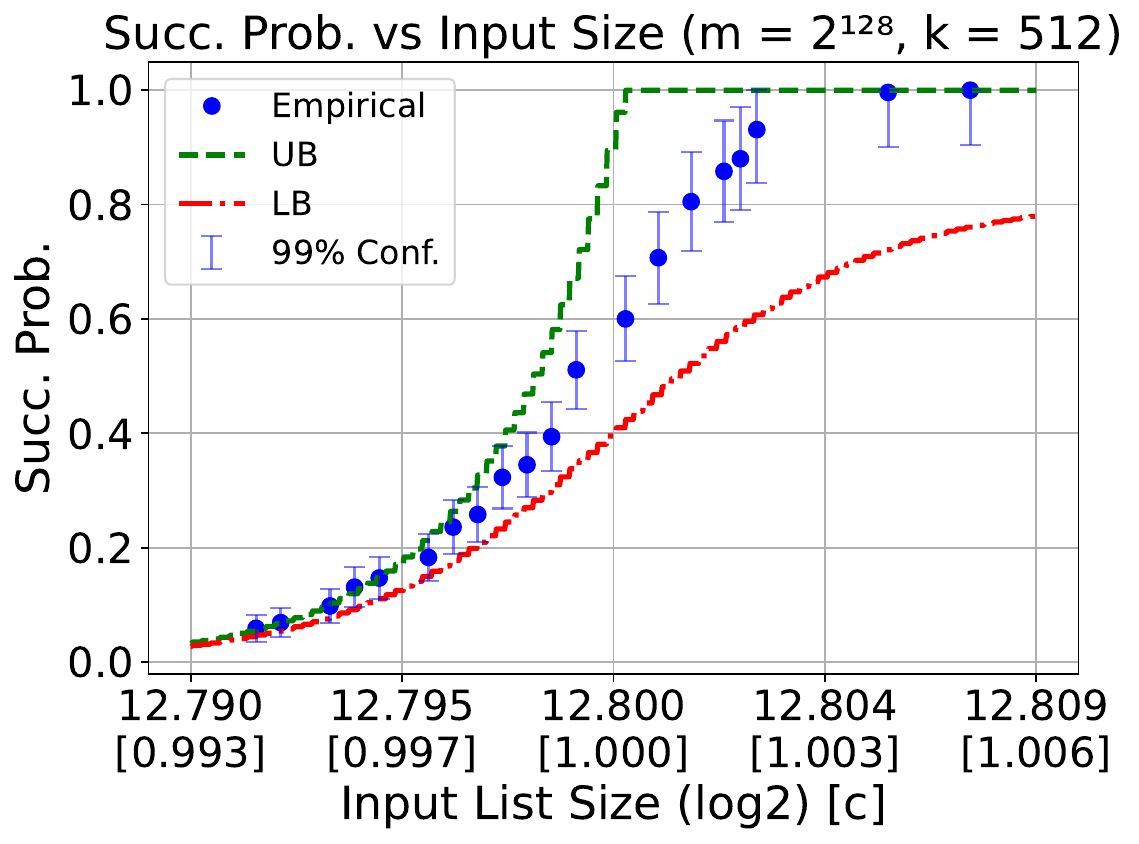}
        \caption{
            }
        \label{fig:type1-128-512}
    \end{subfigure}
    \caption{Success probability when varying the input list size under fixed $m$'s and $k$'s. The values in the square brackets report $c = n/m^{1/(\log{k}+1)}$ as described in~\cref{infthm:ktree}.}
    \label{fig:type1}
\end{figure}

\paragraph{Results.} Figure~\ref{fig:type1} presents the results of our experiments. The blue dots represents the measured success probability, while the dashed red and green lines represent the upper and lower bounds derived from our theoretical analysis, respectively. The error bars around the empirical measurements indicates the 99\% confidence interval, calculated using the Chernoff inequality. The parameter $c$, as defined earlier, serves as a normalized measure of the input list size relative to $m$ and $k$. It allows us to compare results across different parameter configurations and relates directly to our theoretical bounds. These plots enable us to visualize how the success probability transitions from near 0 to near 1 as the input list size increases, and how this transition varies for different values of $m$ and $k$. They also allow us to assess the tightness of our theoretical bounds under various conditions.

The results shown in Figure~\ref{fig:type1} illustrate several key trends in the success probability of the $\ktree$ algorithm as the input list size varies. We report both the actual input list size $n$ and the corresponding $c$ value.
We make the following general observations:

\begin{itemize}
    \item Generally, after fixing reasonable values for $m$ and $k$, the actual success probability approaches $0$ and $1$ at the left and right extremes when the $c$ value is slightly below or above $1$.

    \item For a fixed $m$, as $k$ increases, the success probability converges to $0$ and $1$ more rapidly as $c$ is decreased or increased linearly.

    \item In terms of the bounds we provided, for a fixed $m$, the tightness of the bounds improves as $k$ decreases. Conversely, for a fixed $k$, increasing $m$ significantly improves the tightness of our bounds, which aligns with the earlier conclusion that our bounds are asymptotically tight.

\end{itemize}

As expected from our theoretical model, the $c$ value has a direct and significant impact on the success probability. When $c$ is below $1$, the success probability approaches $0$ gradually. Conversely, when $c$ exceeds $1$, the success probability converges to $1$. When $m$ is fixed, increasing $k$ accelerates the convergence of the success probability toward $0$ and $1$. Larger $k$ increases the rate at which the success probability transitions between failure and success. This faster convergence for larger $k$ highlights a sharper threshold between the regions where the algorithm succeeds and fails, which could be beneficial in applications where precise control over success rates is desired.

The empirical results closely align with our theoretical predictions regarding the behavior of the $\ktree$ algorithm. Specifically, the success probability's dependence on the $c$ value and the role of $k$ in determining the speed of convergence are both well-captured by the theoretical bounds. The observation that the bounds are tighter for smaller $k$ (for fixed $m$) and for larger $m$ (for fixed $k$) is consistent with the asymptotic analysis provided earlier. Moreover, the empirical data validates the hypothesis that our bounds are asymptotically tight: as $m$ grows relative to $k$, the bounds become increasingly accurate, which confirms that the algorithm's performance can be effectively predicted using the theoretical framework developed.

\subsection{Complexity}

\begin{figure}[h!]
    \centering
    \begin{subfigure}[b]{0.48\textwidth}
        \includegraphics[width=\textwidth]{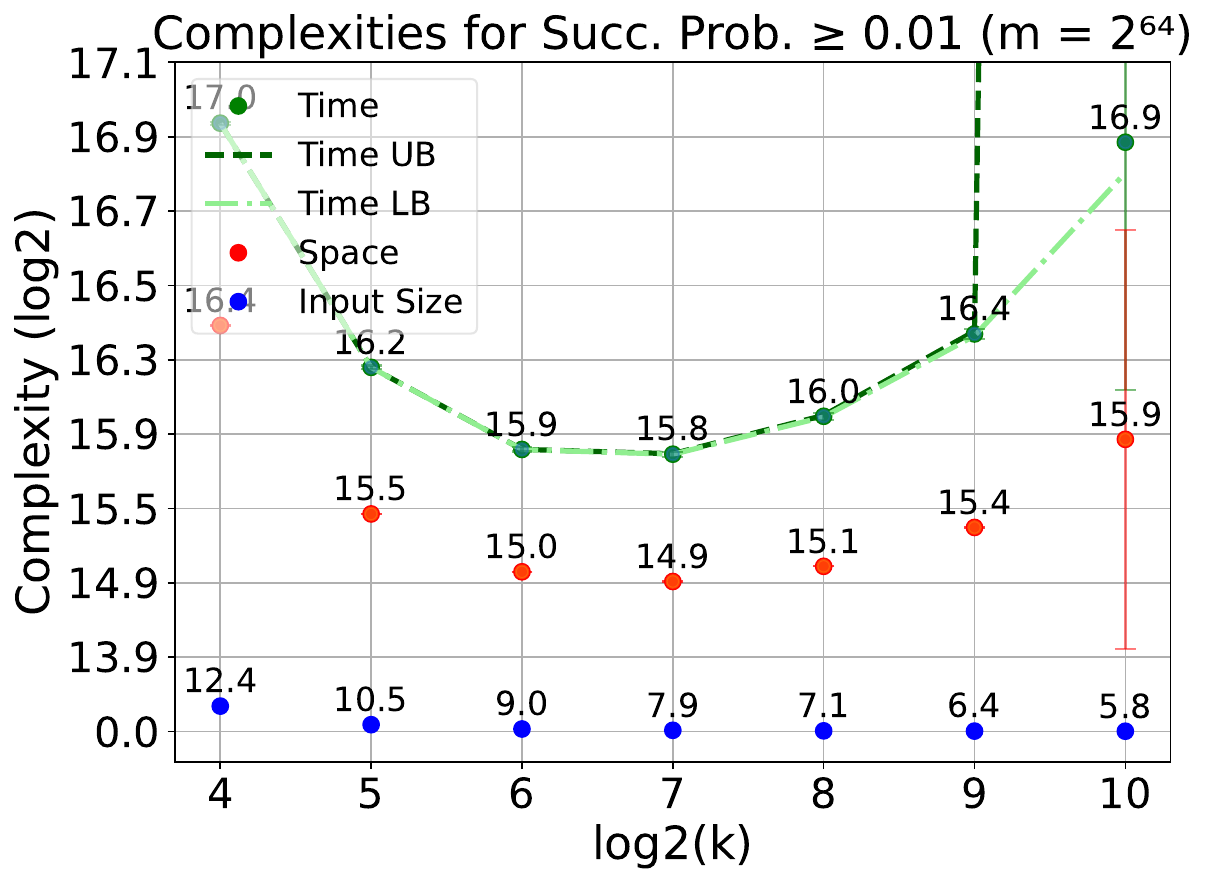}
        \caption{
        }
        \label{fig:type2-64-0.01}
    \end{subfigure}
    \hfill
    \begin{subfigure}[b]{0.48\textwidth}
        \includegraphics[width=\textwidth]{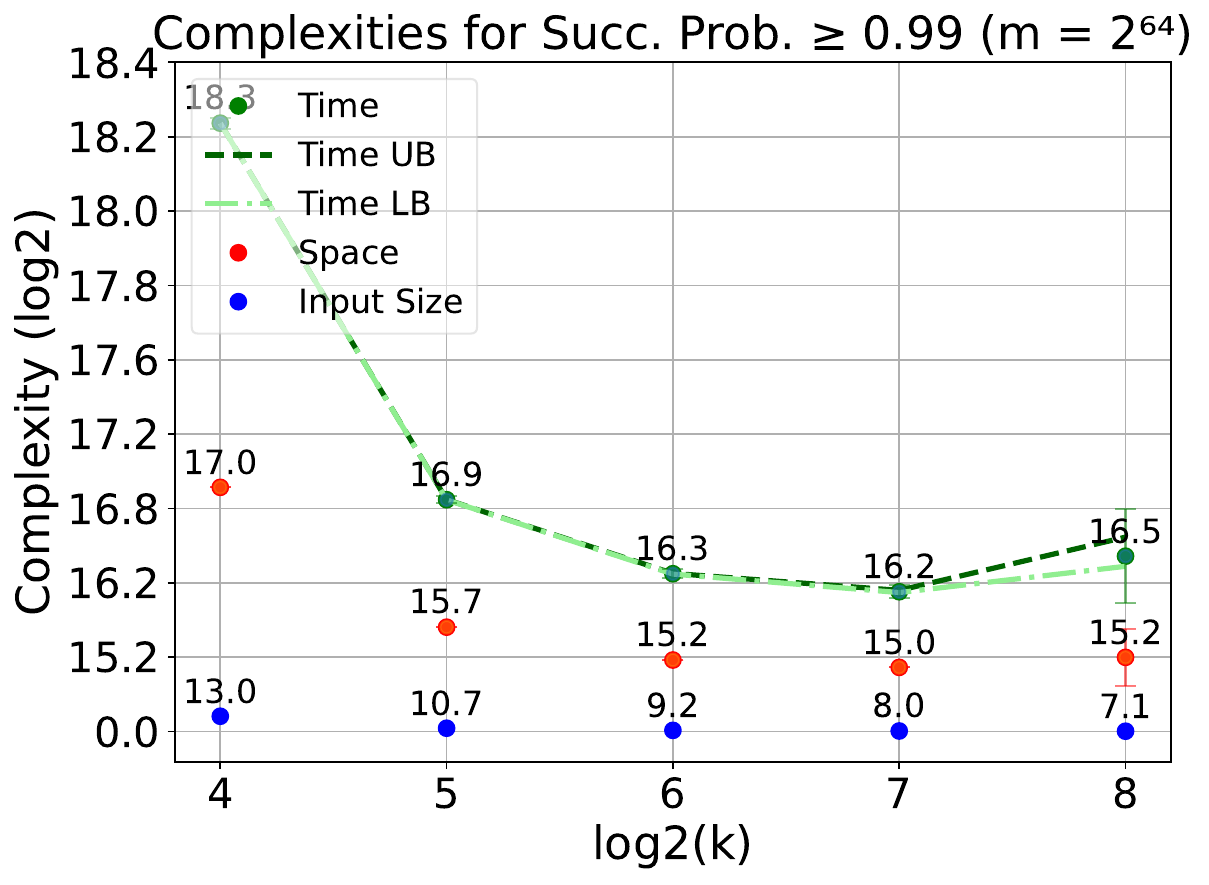}
        \caption{
            }
        \label{fig:type2-64-0.99}
    \end{subfigure}

    \vspace{1em}

    \begin{subfigure}[b]{0.48\textwidth}
        \includegraphics[width=\textwidth]{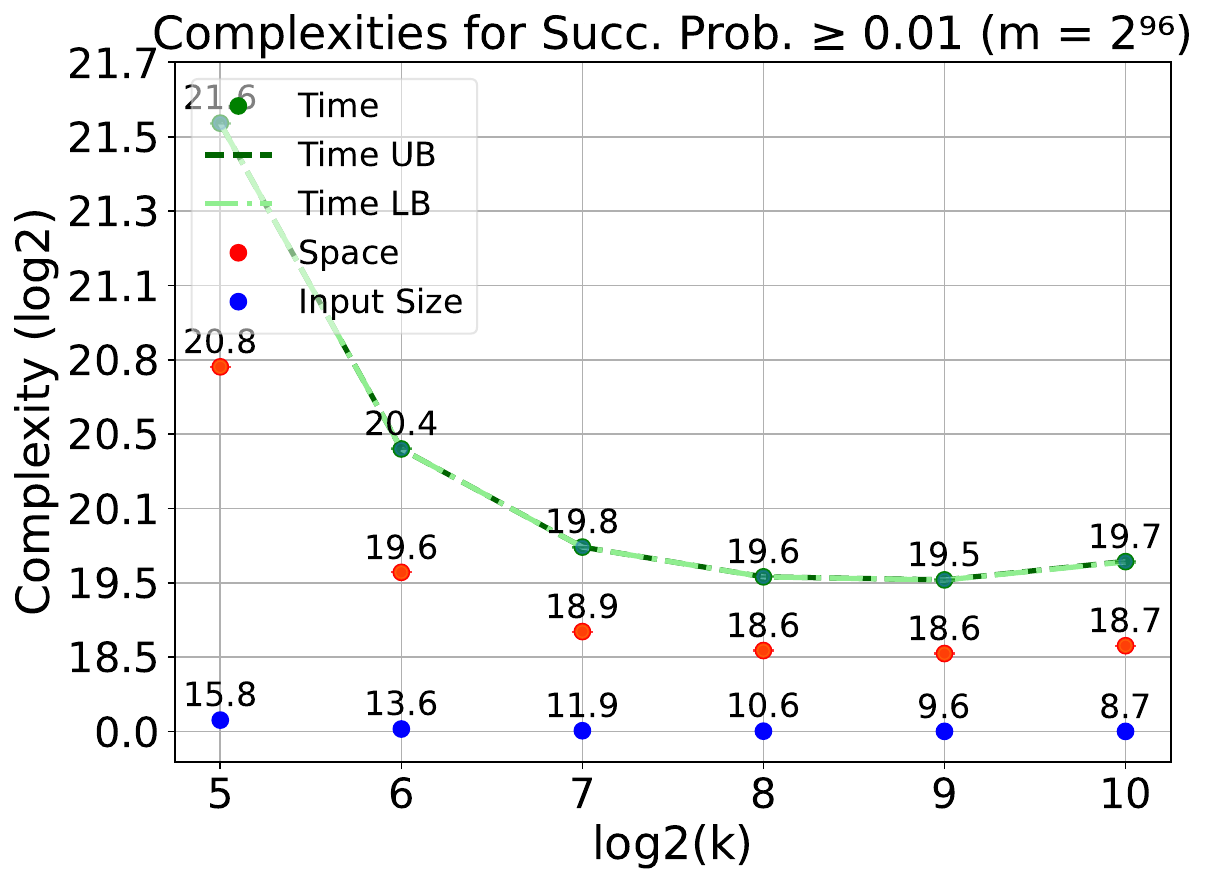}
        \caption{
            }
        \label{fig:type2-96-0.01}
    \end{subfigure}
    \hfill
    \begin{subfigure}[b]{0.48\textwidth}
        \includegraphics[width=\textwidth]{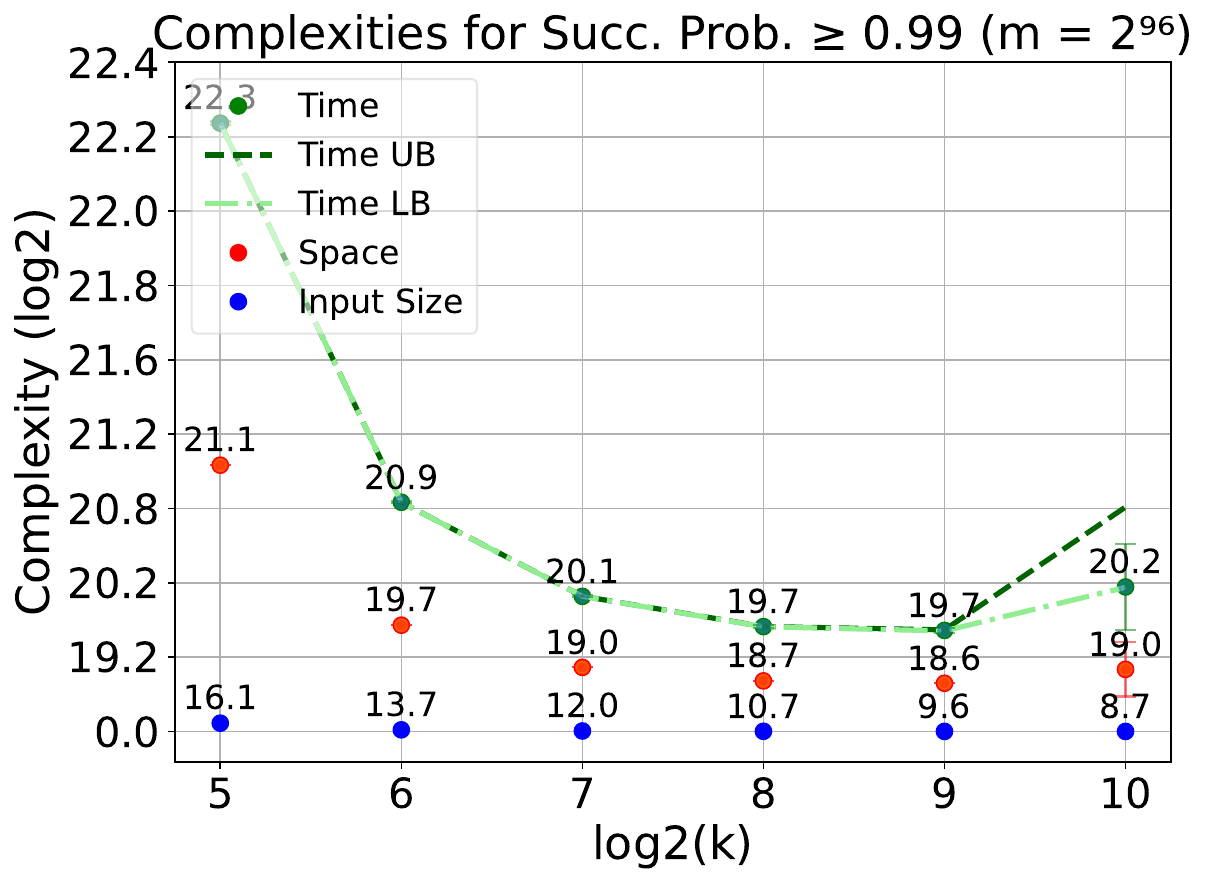}
        \caption{
            }
        \label{fig:type2-96-0.99}
    \end{subfigure}
    \caption{Measurements of $\ktree$'s complexities for different $k$ under fixed $m$'s. Due to the limitation of computational resources, we were only able to measure the $k=512$ and $k=1024$ cases for $m=128$. Therefore its plot is not included.}
    \label{fig:type2}
\end{figure}

\paragraph{Objectives.} The primary objectives of this section are to measure and analyze the time and space complexity of the $\ktree$ algorithm under various parameter configurations. Specifically, we aim to:
\begin{itemize}
    \item Observe how the total list size (a proxy for time complexity) and maximum level size (space complexity) vary with different values of $k$ for fixed $m$ and success probability.
    \item Investigate the impact of different success probability thresholds on the algorithm's complexity.
    \item Identify optimal parameter configurations that minimize time and space usage.
\end{itemize}

\paragraph{Methodology.} For each $m$ ($2^{64}$ and $2^{96}$) and success probability threshold ($0.01$ and $0.99$), we varied $k$ from $4$ to $1024$ in the powers of 2. For each configuration:
\begin{enumerate}
    \item We used binary search to find the minimum input list size $n$ that achieves a success probability above the specified threshold.
    \item We measured the total size and the maximum level size to represent the time and space complexity.
    \item We conduct $1000$ trials and report the average and standard deviation to ensure reliability.
\end{enumerate}

\paragraph{Results.} Figure~\ref{fig:type2} provide a detailed analysis of the algorithm's complexities. For fixed $m$ and a given probability threshold, we observe that both time and space complexities initially decrease and then increase as $k$ grows. This reflects a trade-off: from the results we observed in Figure~\ref{fig:type1}, increasing $k$ boosts the value of $p = m^{\frac{-1}{\log k + 1}}$ and significantly decrease the threshold $c = np$ that achieve success probability $1$. This further leads to a drastically decrease in the input list size. Nevertheless, at the same time, a larger $k$ increases the number of total lists and the number of levels in the $\ktree$ algorithm. This indicates the importance of finding an optimal $k$ value that minimizes complexity.

The impact of the probability threshold (0.01 vs 0.99) on the optimal complexity is relatively small, particularly for larger $k$ values. This suggests that achieving a constantly higher success probability doesn't necessarily incur a significant additional cost. This can also be inferred from the analytical bounds, where the success probability changes exponentially (where $k$ is the exponent) as $n$ vary.

We also observe that as $m$ increases, the minimum complexity (in both time and space) occurs at a higher value of $k$. This indicates that for larger $m$, the optimal choice of $k$—in terms of minimizing both time and space complexity—shifts upwards. This implies that for larger problem sizes, using more input lists can be beneficial up to a certain point. And the insight we obtained here allows users of the algorithm to select the most efficient parameter configuration using our bounds.

  
\section*{Acknowledgements}

This work was supported by the National Research Foundation, Singapore, under its NRF Fellowship programme, award no. NRF-NRFF14-2022-0010.


\iflncs
\else
\bibliographystyle{alpha}
\fi
\bibliography{refs}

\end{document}